\documentclass[fleqn,a4paper]{article}
\usepackage{a4wide}
\usepackage{amsmath}
\usepackage{amssymb}
\usepackage{comment}
\usepackage{hyphenat}
\usepackage{multicol}
\usepackage{times}
\usepackage{bm}
\usepackage{dsfont}
\usepackage{caption}
\usepackage{graphicx}
\usepackage[colorlinks,citecolor=blue,urlcolor=blue]{hyperref}
\usepackage{geometry}
\usepackage{natbib}
\usepackage[ruled,norelsize,vlined]{algorithm2e}

\makeatletter
\renewcommand{\algocf@captiontext}[2]{#1\algocf@typo. \AlCapFnt{}#2} % text of caption
\def\@algocf@capt@plain{top}
\renewcommand{\algocf@makecaption}[2]{%
	\addtolength{\hsize}{\algomargin}%
	\sbox\@tempboxa{\algocf@captiontext{#1}{#2}}%
	\ifdim\wd\@tempboxa >\hsize%     % if caption is longer than a line
	\hskip .5\algomargin%
	\parbox[t]{\hsize}{\algocf@captiontext{#1}{#2}}% then caption is not centered
	\else%
	\global\@minipagefalse%
	\hbox to\hsize{\box\@tempboxa}% else caption is centered
	\fi%
	\addtolength{\hsize}{-\algomargin}%
}
\makeatother

\newtheorem{theorem}{Theorem}

\newtheorem{corollary}{Corollary}
\newtheorem{proposition}{Proposition}
\newtheorem{lemma}{Lemma}
\newtheorem{assumption}{Assumption}
\newtheorem{remark}{Remark}

\newcommand{\qed}{\hfill$\Box$}
\newenvironment{proof}[1][Proof. ]{\addvspace{.2cm} \noindent{\bf #1}}{\qed\vspace{.3cm}}

\def\intercal{{ \mathrm{\scriptscriptstyle T} }}

\newcommand{\E}{{E}}

\newcommand{\bzero}{0 }
\newcommand{\bH}{ H }

\newcommand{\bM}{ M }

\newcommand{\bx}{ x }

\newcommand{\bS}{S}

\newcommand{\by}{y }
\newcommand{\bX}{X}

\newcommand{\bY}{ Y }
\newcommand{\bu}{ u }

\newcommand{\bI}{ I }

\newcommand{\bV}{V }

\newcommand{\bgamma}{ { \gamma} }

\newcommand{\bLambda}{ { \Lambda} }

\newcommand{\bGamma}{ {\Gamma} }
\newcommand{\bDelta}{ { \Delta} }
\newcommand{\bbeta}{ { \beta} }
\newcommand{\bmu}{ {\mu} }
\newcommand{\bSigma}{ {\Sigma} }
\newcommand{\bsigma}{ {\sigma} }
\newcommand{\bOmega}{ { \Omega} }
\newcommand{\bomega}{ { \omega} }

\newcommand{\bxi}{ { \xi} }
\newcommand{\bz}{ z }

\newcommand{\bdelta}{ { \delta} }

\title{Scalable and accurate variational Bayes \\ for  high-dimensional binary regression models}
\author{AUGUSTO FASANO, DANIELE DURANTE \and GIACOMO ZANELLA}
\author{Augusto Fasano\footnote{\mbox{Department of Decision Sciences, Bocconi University, 20136 Milan, and Collegio Carlo Alberto, 10122 Turin, Italy,}  \hspace*{0.5cm}\texttt{augusto.fasano@carloalberto.org}.} \and Daniele Durante\footnote{Department of Decision Sciences and Institute for Data Science and Analytics, Bocconi University, 20136 Milan, Italy.} \and  Giacomo Zanella\footnotemark[\value{footnote}]}

\begin{document}
\maketitle

\begin{abstract}
Modern methods for Bayesian regression beyond the Gaussian response setting are often computationally impractical or inaccurate in high dimensions. In fact, as discussed in recent literature, bypassing such a trade-off is still an open problem even in routine binary regression models, and there is limited theory on the quality of variational approximations in high-dimensional settings.  To address this gap, we study the approximation accuracy of routinely-used mean-field variational Bayes solutions in high-dimensional probit regression with Gaussian priors, obtaining novel and practically relevant results on the pathological behavior of such strategies  in uncertainty quantification, point estimation  and prediction. Motivated by these results, we further develop a new partially-factorized variational approximation for the posterior distribution of the probit coefficients which leverages a representation with global and local variables but, unlike for classical mean-field assumptions,  it avoids a fully factorized approximation, and instead assumes a factorization only for the local variables. We prove that the resulting approximation  belongs to a tractable class of unified skew-normal densities that crucially incorporates skewness and, unlike for state-of-the-art mean-field solutions, converges to the exact posterior density as  $p \rightarrow \infty$. To solve the variational optimization problem, we derive a tractable coordinate ascent variational inference algorithm that easily scales to $p$ in the tens of thousands, and provably requires a  number of iterations converging to $1$ as $p \rightarrow \infty$. Such findings are also illustrated in extensive empirical studies where our novel solution is shown to  improve the approximation accuracy of  mean-field variational Bayes for any $n$ and $p$, with the magnitude of these gains being remarkable in those high-dimensional $p>n$ settings where state-of-the-art methods are computationally impractical.\\

\noindent	\textit{Some key words:}
Bayesian computation; Data augmentation; High-dimensional probit regression;  Truncated normal distribution;  Variational Bayes; Unified skew-normal distribution.
\end{abstract}

\section{Introduction}\label{sec_1}

The absence of a tractable  posterior distribution in several Bayesian models and the  abundance of high-dimensional datasets  have motivated a growing interest in strategies for scalable learning of approximate posteriors, beyond classical sampling-based Markov chain Monte Carlo  methods.  Deterministic approximations, such as variational Bayes  \citep{blei_2017} and expectation-propagation  \citep{minka2001expectation} are routinely-implemented solutions to address this goal. However, in high-dimensional models  such methods still face open problems in terms of scalability, quality of the  approximation and theoretical support \citep{chopin_2017,blei_2017}.  

Notably, such issues also arise in basic regression models for binary responses,  that are routinely implemented and provide a building block in several hierarchical  models \citep[e.g.,][]{chipman_2010, rodrig_2011}. Recalling a review by \citet{chopin_2017}, the problem of posterior computation in binary regression is particularly challenging when the number of predictors $p$ grows. While standard sampling-based algorithms and accurate deterministic approximations can easily deal with small-to-moderate $p$ problems, these strategies are impractical when $p$ is large; e.g.,\ $p> 1000$. The solution to these scalability issues is typically obtained at the cost of lower accuracy in the approximation of the posterior distribution,  and addressing this trade-off, both in theory and in practice, is still an open area of  research \citep{blei_2017}. In this article we prove, both theoretically and empirically, that in $p \gg n$ settings such a lower accuracy cost has dramatic implications not only for uncertainty quantification, but also for point estimation and prediction under routine-use mean-field variational Bayes approximations of the coefficients in probit models with Gaussian priors. To address this issue, we further develop a novel partially-factorized variational approximation for the posterior distribution of the probit coefficients which, unlike for state-of-the-art solutions, combines computational scalability and high accuracy, both theoretically and empirically, in large $p$ settings, especially when $p>n$.  

Classical specifications of Bayesian regression models for dichotomous data assume that the responses $y_{i} \in \{0;1\}$ $(i=1, \ldots, n)$, are conditionally independent realizations from a Bernoulli variable $\mbox{Bern}\{g(\bx_i^{\intercal}\bbeta)\}$, given a fixed $p$-dimensional vector of predictors $\bx_i=(x_{i1}, \ldots, x_{ip})^{\intercal} \in \mathbb{R}^p$ $(i=1, \ldots, n)$, and the associated coefficients $\bbeta=(\beta_1, \ldots, \beta_p)^{\intercal} \in \mathbb{R}^p$. The mapping $g(\cdot): \mathbb{R} \rightarrow (0,1)$ is commonly specified to be either the logit or  probit  link, thus obtaining $\mbox{pr}(y_{i} =1 \mid \bbeta)=\{1+\exp(-\bx_i^{\intercal}\bbeta)\}^{-1}$ in the first case, and $\mbox{pr}(y_{i} =1 \mid \bbeta)=\Phi(\bx_i^{\intercal}\bbeta)$ in the second, where $\Phi(\cdot)$ is the cumulative distribution function of a standard normal.  In performing Bayesian inference under these models, it is common  to specify Gaussian priors for the coefficients in $\bbeta$, and then update such priors with the likelihood of the observed data $\by=(y_1, \ldots, y_n)^{\intercal}$ to obtain the posterior $p(\bbeta \mid \by)$, which is used for point estimation, uncertainty quantification and prediction. However, the apparent lack of conjugacy motivates computational methods relying either on Monte Carlo integration or deterministic approximations \citep[][]{chopin_2017}.

A popular class of Markov chain Monte Carlo methods   for binary regression models leverages augmented data representations that allow the implementation of tractable Gibbs samplers based on  conjugate full-conditional distributions. In probit regression this strategy exploits the possibility of expressing the binary data $y_i \in \{0;1\}$ $(i=1, \ldots, n)$ as dichotomized versions of an underlying regression model for Gaussian responses $z_i \in \mathbb{R}$ $(i=1, \ldots, n)$, thereby restoring conjugacy between the Gaussian prior for the coefficients in $\bbeta$ and the augmented data $z_i$, which are in turn sampled from truncated normal full-conditionals \citep{Albert_1993}. More recently, \citet{Polson_2013} proposed a similar strategy for logit regression which relies on P\'olya-gamma augmented variables $z_i \in \mathbb{R}^+$ $(i=1, \ldots, n)$. Despite their simplicity, these methods can face computational and mixing issues in high-dimensional settings, especially with imbalanced datasets \citep{john_2017}. We refer to \citet{chopin_2017} for a discussion of related strategies \citep{holmes_2006, fru_2007} and alternative sampling methods \citep[e.g.,][]{haario_2001,hoff_2014}.  While  these strategies address some issues of data-augmentation samplers, they are still computationally impractical in large $p$ settings \citep[e.g.,][]{chopin_2017,durante_2019}.

An exception to the above methods is provided by the contribution of \citet{durante_2019}, which proves that in Bayesian probit regression with Gaussian priors the posterior actually belongs to the class of unified skew-normal (\textsc{sun}) distributions \citep{arellano_2006}. These  variables have closure properties that facilitate posterior inference in large $p$ settings. However, the calculation of functionals for inference and prediction  requires the evaluation of cumulative distribution functions of $n$-variate Gaussians or sampling from $n$-variate truncated normals, thus making these results impractical in studies with $n$ greater than a few hundreds \citep{durante_2019}.

A possible solution to scale-up computations is to consider deterministic approximations of the posterior distribution. In binary regression, two popular strategies  are expectation-propagation \citep{chopin_2017} and mean-field variational Bayes with global parameters and local latent variables, typically available for each unit in the form of augmented data \citep{consonni_2007,durante_2019_var}; see also \citet{blei_2017}. Expectation-propagation approximates the exact posterior density by matching the prior and likelihood terms which define $p(\bbeta \mid \by)$ with Gaussian marginals that often yield notable approximation accuracy. These accuracy gains come, however, with a computational cost that limits the applicability of expectation-propagation to settings with moderate $p$ \citep[e.g.,][]{chopin_2017}. Recalling a final remark by \citet{chopin_2017}, this lack of scalability to large $p$ (e.g., $p > 1000$) is common to most  state-of-the-art accurate methods for Bayesian computation in binary regression, including also \citet{knowles2011} and \citet{marlin2011} amongst others.  Mean-field variational Bayes with global and local variables addresses these computational bottlenecks by approximating the posterior density $p(\bbeta,\bz \mid \by)$ for the global parameters $\bbeta=(\beta_1, \ldots, \beta_p)^\intercal$ and the local augmented data $\bz=(z_1, \ldots, z_n)^\intercal$ with an optimal factorized density ${q}^*_{\textsc{mf}}(\bbeta)\prod_{i=1}^n {q}^*_{\textsc{mf}}(z_i) $ which is closest, in Kullback--Leibler divergence \citep{Kullback_1951}, to  $p(\bbeta,\bz \mid \by)$,  among all the approximating densities in the  mean-field family $\mathcal{Q}_{\textsc{mf}}=\{q_{\textsc{mf}}(\bbeta,\bz): q_{\textsc{mf}}(\bbeta,\bz)=q_{\textsc{mf}}(\bbeta) q_{\textsc{mf}}(\bz) \}$. Optimization proceeds via  coordinate ascent variational methods which  scale easily to large $p$. However, the mean-field class underestimates  uncertainty, and  often yields Gaussian approximations that affect inference accuracy if the posterior is non-Gaussian \citep{kuss_2005}. 

In Section \ref{sec_2.1}, we deepen the above results by proving that in high-dimensional probit regression, with Gaussian priors, the mean-field assumption has dramatic implications for uncertainty quantifications, point estimation and prediction. In particular, we prove that, when $p \rightarrow \infty$, with fixed $n$, the  mean-field solution overshrinks the location of the exact posterior distribution and yields to an excessively rapid concentration of the predictive probabilities around $0.5$, thereby leading to unreliable inference and prediction. While there has been recent focus in the literature on studying  consistency properties of mean-field variational Bayes \citep[][]{wang2019,ray2020,yang2020}, there is much less theory on the accuracy in approximating the  whole exact posterior distribution in high dimensions. Our results  in this direction yield novel and practically relevant findings, while motivating improved approximations, as discussed below.

To address the above issues, we propose in Section~\ref{sec_2.2}  a new partially-factorized variational approximation   which crucially drops the  mean-field independence assumption between  $\bbeta$ and the augmented data $\bz$. Unlike for standard implementations of expectation-propagation \citep{chopin_2017} and Hamiltonian Monte Carlo  \citep{hoff_2014}, the proposed solution scales easily to  $p \gg 1000$, while, differently from the methods in \citet{durante_2019}, it only requires the evaluation of distribution functions of univariate Gaussians. Moreover, despite having a computational cost comparable to mean-field strategies \citep{consonni_2007}, the proposed solution yields a substantially improved approximation of the posterior that crucially incorporates skewness, and effectively reduces bias in locations, variances and predictive probabilities, especially when $p>n$. In fact, we prove that the error made by the partially-factorized solution  in approximating the exact posterior density goes to $0$ as $p \rightarrow \infty$. Optimization proceeds via a simple and scalable coordinate ascent variational algorithm. This routine has  $\mathcal{O}(pn \cdot\min\{p,n\})$ overall cost, requires a number of iterations  to find the optimum that converges to $1$ as $p \rightarrow \infty$, and provides a tractable \textsc{sun} approximating density.  Quantitative assessments in Sections \ref{sec_simu}--\ref{sec_3} confirm that our proposed  approximation is orders of magnitude faster than state-of-the-art accurate methods in high-dimensional empirical studies and, unlike for classical mean-field methods, pays almost no cost in inference accuracy and predictive performance, when $p>n$. Codes and tutorials  are available at  \texttt{https://github.com/augustofasano/Probit-PFMVB}.

\section{Approximate Bayesian inference for probit models}\label{sec_2}
\subsection{Model formulation and augmented data representation}\label{sec_2.0}
Recalling Section \ref{sec_1}, we focus on posterior  inference for  the classical Bayesian probit  model 
\begin{eqnarray}
(y_i \mid \bbeta) \sim \mbox{Bern}\{\Phi(\bx_i^{\intercal}\bbeta)\} \quad  (i=1, \ldots, n), \qquad \bbeta \sim \mbox{N}_p(0, \nu_p^2 \bI_p). 
\label{eq1}
\end{eqnarray}
In \eqref{eq1}, each $y_i$ is a Bernoulli variable whose success probability depends on a $p$-dimensional vector of predictors $\bx_i=(x_{i1}, \ldots, x_{ip})^{\intercal}$ under probit mapping. The coefficients $\bbeta=(\beta_1, \ldots, \beta_p)^{\intercal}$ regulate the effect of each predictor and  are assigned independent Gaussian priors $\beta_j \sim \mbox{N}(0, \nu_p^2)$, for $j=1, \ldots, p$. Besides providing a default specification, these priors are also relevant building-blocks which can be naturally extended to scale mixtures of Gaussians. Although our contribution can be readily extended to a generic multivariate Gaussian prior for $\bbeta$, we consider here the common setting with $ \bbeta \sim \mbox{N}_p(0, \nu_p^2 \bI_p)$ to ease notation, and allow  $\nu_p^2$ to possibly change with $p$. This choice incorporates not only  routine implementations of Bayesian probit models that rely on constant prior variances $\nu_p^2=\nu^2$ for the coefficients \citep[e.g.,][]{chopin_2017},  but also more structured formulations for high-dimensional problems which define $\nu_p^2=\nu^2/p$ to control  prior variance of the entire linear predictor and induce increasing shrinkage \citep[e.g.,][]{simpson_2017,fuglstad2018intuitive}; see Assumption \ref{ass_2} for details on the asymptotic behavior of $\nu_p^2$.

Model \eqref{eq1} also has  a simple constructive representation relying on Gaussian augmented data, which has been broadly used in the development of sampling-based   \citep{Albert_1993}  and variational  \citep{consonni_2007,girolami2006}  methods. More specifically,  \eqref{eq1} can be obtained  by marginalizing out the augmented data $\bz=(z_1, \ldots, z_n)^{\intercal}$ in the model
\begin{eqnarray}
y_i = {\rm 1}(z_i>0),  \qquad  (z_i \mid \bbeta) \sim \mbox{N}(\bx^{\intercal}_i \bbeta,1)  \quad (i=1, \ldots, n), \qquad   \bbeta \sim \mbox{N}_p(0, \nu_p^2 \bI_p). 
\label{eq2}
\end{eqnarray}
Recalling \citet{Albert_1993}, this construction yields closed-form full-conditionals for $\bbeta$ and $\bz$, thus allowing the implementation of a  Gibbs sampler where $p(\bbeta \mid \bz, \by)=p(\bbeta \mid \bz)$ is a Gaussian density, and each $p(z_i \mid \bbeta, \by)=p(z_i \mid \bbeta, y_i)$ is the density of a truncated normal. Our focus here is on large $p$ settings that have motivated a growing interest in scalable optimization-based solutions relying on approximate posteriors, beyond classical sampling-based strategies \citep{blei_2017}. As clarified in Section \ref{sec_2.1}, also these strategies exploit representation~\eqref{eq2}, but are based on overly-restrictive assumptions which yield to poor approximations in high dimensions.

\subsection{Mean-field variational Bayes with global and local variables}\label{sec_2.1}
Mean-field variational Bayes with global and local variables aims at providing a tractable approximation for the joint posterior density $p(\bbeta,\bz \mid \by)$ of the global parameters $\bbeta=(\beta_1, \ldots, \beta_p)^{\intercal}$ and the local variables $\bz=(z_1, \ldots, z_n)^{\intercal}$, within the mean-field family  of factorized densities $\mathcal{Q}_{\textsc{mf}}=\{q_{\textsc{mf}}(\bbeta,\bz): q_{\textsc{mf}}(\bbeta,\bz)=q_{\textsc{mf}}(\bbeta) q_{\textsc{mf}}(\bz) \}$. The optimal solution  ${q}^*_{\textsc{mf}}(\bbeta) {q}^*_{\textsc{mf}}(\bz) $
within this class is the one that minimizes the  Kullback--Leibler (\textsc{kl}) \citep{Kullback_1951} divergence    
\begin{eqnarray}
{\normalfont  \textsc{kl}}\{q_{\textsc{mf}}(\bbeta, \bz) \mid \mid p(\bbeta,\bz \mid \by)\}={E}_{q_{\textsc{mf}}(\bbeta, \bz)}\{\log q_{\textsc{mf}}(\bbeta, \bz)\}-{E}_{q_{\textsc{mf}}(\bbeta, \bz)}\{\log p(\bbeta, \bz \mid \by)\},
\label{eq3}
\end{eqnarray}
with $q_{\textsc{mf}}(\bbeta, \bz) \in \mathcal{Q}_{\textsc{mf}}$. Alternatively, it is possible to obtain ${q}^*_{\textsc{mf}}(\bbeta){q}^*_{\textsc{mf}}(\bz) $ by maximizing 
\begin{eqnarray}
{\normalfont  \textsc{elbo}}\{q_{\textsc{mf}}(\bbeta, \bz)\}={E}_{q_{\textsc{mf}}(\bbeta, \bz)}\{\log p(\bbeta, \bz, \by)\}-{E}_{q_{\textsc{mf}}(\bbeta, \bz)}\{\log q_{\textsc{mf}}(\bbeta, \bz)\},
\label{eq4}
\end{eqnarray}
since the \textsc{elbo} coincides with the negative \textsc{kl} up to an additive constant. Recall also that the \textsc{kl} divergence is always non-negative and refer to \citet{armagan_2011} for the expression of ${\normalfont  \textsc{elbo}}\{q_{\textsc{mf}}(\bbeta, \bz)\}$ under \eqref{eq2}. The maximization of \eqref{eq4} is typically easier  and can be performed via a coordinate ascent variational algorithm \citep[e.g.,][]{blei_2017} cycling among steps
\begin{eqnarray}
q_{\textsc{mf}}^{(t)}(\bbeta)  &\propto \exp\left[{E}_{q_{\textsc{mf}}^{(t-1)}(\bz)}\{\log p(\bbeta\mid \bz,  \by)\}\right], \ \
q_{\textsc{mf}}^{(t)}(\bz) &\propto \exp\left[{E}_{q_{\textsc{mf}}^{(t)}(\bbeta)}\{\log p(\bz\mid  \bbeta, \by)\}\right],\quad\quad
\label{eq5}
\end{eqnarray}
where $q_{\textsc{mf}}^{(t)}(\bbeta)$ and $ q_{\textsc{mf}}^{(t)}(\bz)$ are the solutions at iteration $t$. We refer to \citet{blei_2017} for why the updating in \eqref{eq5}  iteratively optimizes the \textsc{elbo} in \eqref{eq4}, and highlight here how \eqref{eq5} is particularly simple to implement in Bayesian models with tractable full-conditional densities $p(\bbeta\mid \bz, \by)$ and $p(\bz\mid\bbeta, \by)$. This is the case of the augmented-data representation \eqref{eq2} for model \eqref{eq1}. Indeed, recalling  \citet{Albert_1993} it easily follows that the full-conditionals in model \eqref{eq2} are
\begin{eqnarray}
\begin{split}
(\bbeta\mid \bz, \by) &\sim \mbox{N}_p(\bV\bX^\intercal\bz, \bV), \quad \bV=(\nu_p^{-2}\bI_p + \bX^\intercal \bX)^{-1}, \\
(z_i\mid  \bbeta, \by) &\sim \textsc{tn}(\bx_i^{\intercal}\bbeta,1, \mathbb{A}_{z_i}) \quad (i=1, \ldots, n),
\label{eq6}
\end{split}
\end{eqnarray}
where $\bX$ is the $n \times p$ design matrix with rows $\bx^{\intercal}_i$, whereas $\textsc{tn}(\bx_i^{\intercal}\bbeta,1, \mathbb{A}_{z_i})$ denotes a univariate normal distribution having mean $\bx_i^{\intercal}\bbeta$, variance $1$, and truncation to the interval $\mathbb{A}_{z_i}=\{ z_i: (2y_i-1)z_i>0\}$.  A key consequence of the conditional independence of $z_1,\dots,z_n$, given $\bbeta $ and $\by$, is that ${q}^*_{\textsc{mf}}(\bz)=\prod_{i=1}^n{q}^*_{\textsc{mf}}(z_i)$ and thus the optimal mean-field  solution always factorizes as ${q}^*_{\textsc{mf}}(\bbeta) {q}^*_{\textsc{mf}}(\bz)={q}^*_{\textsc{mf}}(\bbeta) \prod_{i=1}^n{q}^*_{\textsc{mf}}(z_i)$. Replacing the densities of the above full-conditionals within \eqref{eq5}, it can be easily noted that $q_{\textsc{mf}}^*(\bbeta)$ and $ q_{\textsc{mf}}^*(z_i)$ $(i=1, \ldots, n)$, are Gaussian and truncated normal densities, respectively, with parameters obtained as in Algorithm \ref{algo1} \citep{consonni_2007}.   In Algorithm \ref{algo1}, the generic quantity $\phi_p(\bbeta-\bmu; \bSigma)$ refers to the density of a $p$-variate Gaussian for $\bbeta$ with mean $\bmu$ and variance-covariance matrix $\bSigma$.
\begin{algorithm*}[b!]
	\caption{Steps of the procedure to obtain ${q}^*_{\textsc{mf}}(\bbeta,\bz)={q}^*_{\textsc{mf}}(\bbeta) \prod_{i=1}^n {q}^*_{\textsc{mf}}(z_i)$.} 
	\label{algo1}
	\vspace{3pt}
	{[\bf 1]} {\small Initialize $\bar{\bz}^{(0)}=(\bar{z}^{(0)}_1, \ldots, \bar{z}^{(0)}_n)^{\intercal} \in \mathbb{R}^n$.}\\
	\vspace{1pt}
	{[\bf 2]}  \For(){\small $t$ \mbox{from} $1$ until convergence of ${\normalfont  \textsc{elbo}}\{q_{\textsc{mf}}^{(t)}(\bbeta, \bz)\}$}
	{
		\vspace{1pt}
		{\small  Set $\bar{\bbeta}^{(t)}=\bV\bX^\intercal\bar{\bz}^{(t-1)}$, with $\bar{\bz}^{(t-1)}= (\bar{z}^{(t-1)}_1, \ldots, \bar{z}^{(t-1)}_n)^{\intercal}$.}\\
		\vspace{-1pt}
		{\small Set  $\bar{z}_i^{(t)}= \bx_i^{\intercal}\bar{\bbeta}^{(t)}+(2y_i-1)\phi(\bx_i^{\intercal}\bar{\bbeta}^{(t)})/\Phi\{(2y_i-1)\bx_i^{\intercal}\bar{\bbeta}^{(t)}\}$ $(i=1, \ldots, n)$.
	}}
	\vspace{2pt}
	{\bf Output:} {\small Optimal mean-field variational approximations ${q}^*_{\textsc{mf}}(\bbeta) =\phi_p(\bbeta-\bar{\bbeta}^{*}; \bV)$ and ${q}^*_{\textsc{mf}}(z_i)={1}\{(2y_i-1)z_i>0\}\phi(z_i-\bx_i^{\intercal}\bar{\bbeta}^{*})/\Phi\{(2y_i-1)\bx_i^{\intercal}\bar{\bbeta}^{*}\}$ $(i=1, \ldots, n)$.}\\
	\vspace{2pt}
\end{algorithm*}

Algorithm \ref{algo1} relies on simple operations which basically require only to update $\bar{\bbeta}^{(t)}$ via matrix operations, and, unlike for most state-of-the-art alternative strategies, is computationally feasible in high-dimensional settings; see e.g., Table \ref{table1}, and refer to Section~\ref{sec_2.2} and to the Supplementary Material which clarify that the overall cost of Algorithm \ref{algo1} is dominated by the matrix  operations involving $\bV$ and $\bX$ in the pre-computation part. Due to the Gaussian form of ${q}^*_{\textsc{mf}}(\bbeta)$ also the calculation of the approximate posterior moments and predictive probabilities is straightforward. Recalling Lemma 7.1 in \citet{azzalini_2013}, the latter quantities can be expressed as
\begin{eqnarray}
\mbox{\normalfont pr}_{\textsc{mf}}(y_{\textsc{new}}=1\mid \by)
=E_{{q}^*_{\textsc{mf}}(\bbeta)}\{\Phi(\bx^{\intercal}_{\textsc{new}} \bbeta)\}
=\Phi\{\bx^{\intercal}_{\textsc{new}} \bar{\bbeta}^{*}(1+\bx^{\intercal}_{\textsc{new}}\bV  \bx_{\textsc{new}})^{-1/2}\},
\label{eq7}
\end{eqnarray}
where $\bx_{\textsc{new}}\in\mathbb{R}^p$ are the covariates of the new unit, and $\bar{\bbeta}^{*}=E_{{q}^*_{\textsc{mf}}(\bbeta)}(\bbeta)$.
However, as shown by the asymptotic results in Theorem \ref{teo_1}, the mean-field solution can lead to poor approximations of the posterior in high dimensions as $p \rightarrow \infty$, causing concerns on the quality of inference and prediction. These asymptotic results are derived under the following random design assumption.

\begin{assumption}
	Assume that the predictors $x_{ij}$ $(i=1, \ldots, n;j=1, \ldots, p)$, are independent random variables with $E(x_{ij})=0$, $E(x_{ij}^2)=\sigma_x^2$ and $E(x_{ij}^4)\leq c_x$ for some $c_x<\infty$.
	\label{ass_1}
\end{assumption}

The above random design assumption is common to asymptotic studies of regression models \citep[e.g.,][]{reiss_2008, qin_2019}. Moreover, the zero mean and constant variance assumption is a natural requirement in  binary regression, where the predictors are typically standardized \citep{gelman_2008,chopin_2017}. In Section \ref{sec_2.2}, we will relax Assumption~\ref{ass_1}  to prove that the proposed partially-factorized solution has desirable asymptotic properties also in settings including dependence among predictors. Since  the mean-field approximation has pathological asymptotic behaviors even under complete independence assumptions, we avoid such a relaxation in Theorem~\ref{teo_1} to simplify  presentation. As  illustrated in Sections \ref{sec_simu}--\ref{sec_3} and in the Supplementary Material, the empirical evidence on simulated and real data, where our assumptions might not hold, is coherent with the theory results stated in the article. To rule out pathological cases, we also require the following mild technical assumption on the behavior of $\nu_p^2$ as $p\to\infty$.

\begin{assumption}
	Assume that $\sup_p\{\nu_p^2\}<\infty$ and $ \lim_{p \to \infty} p \nu_p^2=\alpha \in (0,\infty]$.
	\label{ass_2}
\end{assumption}

Assumption \ref{ass_2} includes the two elicitations for the prior variance  of interest in these settings as highlighted in Section \ref{sec_2.0}, namely,  $\nu_p^2=\nu^2$ and $\nu_p^2=\nu^2/p$, with $\nu^2 <\infty$. Throughout the paper, $n$ is assumed to be fixed while $p \to \infty$. A discussion on possible future extensions of the theoretical results in regimes when $n$ is allowed to grow with $p$ is provided in Section~\ref{sec_4}.

\begin{theorem}
	Under Assumptions \ref{ass_1}--\ref{ass_2} with $\alpha=\infty$, e.g.,  $\nu_p^2=\nu^2/p^{\gamma}$,  $\gamma \in [0,1)$, it holds
	\begin{description}
		\item{(a) \  $\liminf_{p\to\infty}\textsc{kl}\{q_{\textsc{mf}}^{*}(\bbeta) \mid \mid p(\bbeta \mid \by)\}>0$ almost surely (a.s.).}
		\item{(b) \ $\nu_p^{-1}||{E}_{{q}^*_{\textsc{mf}}(\bbeta)}(\bbeta)|| \rightarrow 0$ a.s., whereas  $\nu_p^{-1}||{E}_{p(\bbeta \mid \by)}(\bbeta)|| \rightarrow   c \sqrt{n} >0$ a.s.,  as $p \rightarrow \infty$, where $c=2\int_{0}^{\infty}u\phi(u)\mbox{\normalfont d}u$ is a strictly positive constant and $||\cdot||$ denotes the Euclidean norm.}
		\item{(c) $\liminf_{p\to\infty }\sup_{\bx_{\textsc{new}}\in\mathbb{R}^p}|\mbox{\normalfont pr}_{\textsc{mf}}(y_{\textsc{new}}=1\mid \by)-\mbox{\normalfont pr}(y_{\textsc{new}}=1\mid \by)|> 0$ a.s., when $p\to\infty$, where  $\mbox{\normalfont pr}(y_{\textsc{new}}=1\mid \by)=E_{p(\bbeta|\by)}\{\Phi(\bx^{\intercal}_{\textsc{new}} \bbeta)\}$ denote the exact posterior predictive probability for a new unit with predictors $\bx_{\textsc{new}} \in \mathbb{R}^p$.}
	\end{description}	
	\label{teo_1}
\end{theorem}

According to Theorem \ref{teo_1}(a)--\ref{teo_1}(b), the mean-field solution causes over-shrinkage of the  approximate posterior means, that can result in an unsatisfactory  approximation of the whole posterior density $p(\bbeta \mid \by)$ in high dimensions. For instance, combining Theorem \ref{teo_1}(b) and \eqref{eq7}, the over-shrinkage of $\bar{\bbeta}^{*}$ towards $0$ may cause an excessively rapid concentration of $\mbox{\normalfont pr}_{\textsc{mf}}(y_{\textsc{new}}=1\mid \by)$ around $0.5$, thereby inducing bias; see Theorem \ref{teo_1}(c).  As shown with extensive empirical studies in Sections \ref{sec_simu}--\ref{sec_3}, this bias can be dramatic in high dimensions and is evident also beyond the specific regimes considered by Theorem \ref{teo_1}, including when $\alpha \in (0,\infty)$ which covers $\nu^2_p=\mathcal{O}(p^{-1})$. In addition,  although as $p \to\infty$ the prior  plays a more important role in the Bayesian updating, Theorem \ref{teo_1}(b) suggests that even few data points can induce evident differences between prior and posterior moments such as, for example, the expectations. As shown in Theorem~\ref{teo_3}, such issues can be provably solved by the partially-factorized solution proposed in Section~\ref{sec_2.2}.

\begin{remark}
	As discussed in  \citet{armagan_2011}, 
	$\bar{\bbeta}^{*}$ is also the  mode of the exact posterior $p(\bbeta \mid \by)$ and, hence, it also coincides with the ridge estimate in probit regression. Therefore, the above discussion suggests that the posterior mode should be used with caution when  $p \gg n$.  As a direct consequence, also Laplace approximation would yield unreliable inference since it is centered at the posterior mode.  
\end{remark}

The above results are in apparent contrast with the fact that the marginal posterior densities $p(\beta_j \mid \by)$ often exhibit negligible skewness and their modes $\arg\max_{\beta_j} p(\beta_j \mid \by)$ tend to be close to the corresponding mean ${E}_{p(\beta_j \mid \by)}(\beta_j)$; see  Figure \ref{f2} in the Supplementary Material. However, the same is not true for the joint posterior density $p(\bbeta	\mid \by)$, where little skewness is sufficient to induce dramatic differences between the joint posterior mode, $\arg\max_{\bbeta} p(\bbeta {\mid} \by)$, and the posterior expectation; see e.g., Figure \ref{f3}. In this sense, the results in Theorem \ref{teo_1} point towards caution in assessing Gaussianity of high-dimensional distributions based on the shape of the marginals.

Motivated by these considerations, in Section \ref{sec_2.2} we develop a new partially-factorized variational strategy that solves the above issues without increasing the computational cost.  In fact, the cost of our procedure is the same of the classical mean-field one but, unlike for such a strategy,  we obtain a substantially improved approximation whose \textsc{kl} divergence from the exact posterior goes to $0$ as $p \rightarrow \infty$. The magnitude of these gains is empirically illustrated in Sections \ref{sec_simu} and \ref{sec_3}.

\subsection{Partially-factorized variational Bayes with global and local variables}
\label{sec_2.2}
A natural option to improve the mean-field performance  is to relax the factorization assumptions on the approximating densities in a way that still allows simple optimization and inference. To accomplish this goal, we consider a partially-factorized representation $\mathcal{Q}_{\textsc{pfm}}=\{ q_{\textsc{pfm}}(\bbeta,\bz): q_{\textsc{pfm}}(\bbeta,\bz)=q_{\textsc{pfm}}(\bbeta\mid \bz)\prod_{i=1}^nq_{\textsc{pfm}}(z_i) \}$ which does not assume independence between the parameters  in $\bbeta$ and the local variables $\bz$, thus providing a more flexible family of approximating densities. This new enlarged family $\mathcal{Q}_{\textsc{pfm}}$ allows to incorporate more structure of the actual posterior relative to $\mathcal{Q}_{\textsc{mf}}$, while retaining tractability. In fact, following \citet{holmes_2006} and recalling that $\bV=(\nu_p^{-2}\bI_p+\bX^{\intercal}\bX)^{-1}$, the joint density $p(\bbeta,\bz \mid \by)$ for the augmented model~\eqref{eq2}  factorizes as $p(\bbeta,\bz \mid \by)=p(\bbeta \mid \bz)p(\bz \mid \by)$, where $p(\bbeta \mid \bz)=\phi_p(\bbeta-\bV\bX^\intercal\bz; \bV)$ and $p(\bz \mid \by)\propto\phi_n(\bz; \bI_n+\nu_p^2\bX\bX^{\intercal})\prod_{i=1}^n {1}\{(2y_i-1)z_i>0\}$ denote the densities of a $p$-variate Gaussian and an $n$-variate truncated normal, respectively. The main source of intractability in this factorization of the posterior  is the multivariate  truncated normal density $p(\bz \mid \by)$, which requires evaluation of cumulative distribution functions of $n$-variate Gaussians with full variance-covariance matrix for  inference \citep[e.g.,][]{botev_2017,durante_2019,cao2019}. The independence assumption among the augmented data in $\mathcal{Q}_{\textsc{pfm}}$ avoids this source of intractability, while being fully flexible on $q_{\textsc{pfm}}(\bbeta\mid \bz)$.  Crucially, the mean-field solution  ${q}^*_{\textsc{mf}}(\bbeta,\bz)$ belongs to $\mathcal{Q}_{\textsc{pfm}}$ and thus, by minimizing $\textsc{kl}\{q_{\textsc{pfm}}(\bbeta,\bz) \mid \mid p(\bbeta,\bz \mid \by)\}$ in $\mathcal{Q}_{\textsc{pfm}}$,  we are guaranteed to obtain  improved approximations of the joint posterior density relative to mean-field, as stated in Proposition~\ref{rem_1}.

\begin{proposition}
	Let $q^*_{\textsc{mf}}(\bbeta,\bz)$ and $q^*_{\textsc{pfm}}(\bbeta,\bz)$ denote the optimal  approximations for $p(\bbeta,\bz \mid \by)$ from model \eqref{eq2}, under mean-field and  partially-factorized  variational Bayes, respectively. Then $\textsc{kl}\{q^*_{\textsc{pfm}}(\bbeta,\bz) \mid \mid p(\bbeta,\bz \mid \by)\} \leq \textsc{kl}\{q^*_{\textsc{mf}}(\bbeta,\bz) \mid \mid p(\bbeta,\bz \mid \by)\}$.
	\label{rem_1}
\end{proposition}

Proposition~\ref{rem_1} follows from the fact that ${q}^*_{\textsc{mf}}(\bbeta,\bz)$ belongs to $\mathcal{Q}_{\textsc{pfm}}$, while ${q}^*_{\textsc{pfm}}(\bbeta,\bz)$ is the actual minimizer of $\textsc{kl}\{q(\bbeta,\bz) \mid \mid p(\bbeta,\bz \mid \by)\}$ within the family $\mathcal{Q}_{\textsc{pfm}}$. This suggests that our strategy may provide a promising direction to improve the quality of posterior approximation. However, to be useful in practice, the solution $q^*_{\textsc{pfm}}(\bbeta,\bz)$ should be simple to derive, and the  approximate posterior of interest $q^*_{\textsc{pfm}}(\bbeta)=\int_{\mathbb{R}^n}q^*_{\textsc{pfm}}(\bbeta \mid \bz)\prod_{i=1}^nq^*_{\textsc{pfm}}(z_i) \mbox{d}\bz={E}_{q^*_{\textsc{pfm}}(\bz)}\{q^*_{\textsc{pfm}}(\bbeta \mid \bz)\}$  should have a tractable form. Theorem \ref{teo_2} and Corollary \ref{cor_SUN} show that this is possible.

\begin{theorem}
	Under model \eqref{eq2}, the $\textsc{kl}$ divergence between \mbox{$q_{\textsc{pfm}}(\bbeta,\bz) \in \mathcal{Q}_{\textsc{pfm}}$} and $p(\bbeta,\bz \mid \by)$ is minimized at $q^*_{\textsc{pfm}}(\bbeta \mid\bz) \prod_{i=1}^n q^*_{\textsc{pfm}}(z_i)$ with
	\begin{eqnarray}
	\begin{split}
	&q^*_{\textsc{pfm}}(\bbeta \mid\bz)=p(\bbeta \mid \bz)=\phi_p(\bbeta-\bV\bX^\intercal\bz; \bV), \quad \bV=(\nu_p^{-2}\bI_p+\bX^{\intercal}\bX)^{-1},\\
	&q_{\textsc{pfm}}^*(z_i)= \frac{\phi(z_i-\mu^*_i; \sigma^{*2}_i)}{\Phi\{(2y_i-1)\mu^*_i/\sigma^{*}_i\}}{1}\{(2y_i-1)z_i>0\},   \  \sigma^{*2}_i=\frac{1}{1-\bx^{\intercal}_i\bV \bx_i} \quad   (i=1, \ldots, n),
	\end{split}
	\label{eq8}
	\end{eqnarray}
	where $\bmu^*=(\mu_1^*, \ldots, \mu_n^*)^{\intercal}$ solves the system  $\mu^*_i-\sigma^{*2}_i \bx_i^{\intercal}\bV\bX^{\intercal}_{-i}\bar{\bz}^*_{-i}=0,$ for each $i=1, \ldots, n$, with $\bX_{-i}$ denoting the design matrix without the $i$th row, while $\bar{\bz}^*_{-i}$ is an $n-1$ vector obtained by removing the $i$th element $\bar{z}^*_i=\mu^*_i+(2y_i-1)\sigma^{*}_i\phi(\mu^*_i/\sigma^{*}_i)/\Phi\{(2y_i-1)\mu^*_i/\sigma^{*}_i\}$ $(i=1, \ldots, n)$, from the vector $\bar{\bz}^*=(\bar{z}^*_1, \ldots, \bar{z}^*_n)^{\intercal}$.  
	\label{teo_2}
\end{theorem}

In Theorem \ref{teo_2}, the new solution for $q^*_{\textsc{pfm}}(\bbeta \mid\bz)$ follows by noting that 
$\textsc{kl}\{q_{\textsc{pfm}}(\bbeta,\bz) \mid \mid p(\bbeta,\bz \mid \by)\}=\textsc{kl}\{q_{\textsc{pfm}}(\bz) \mid \mid p(\bz \mid \by)\}+{E}_{q_{\textsc{pfm}}(\bz)}[\textsc{kl}\{q_{\textsc{pfm}}(\bbeta \mid \bz) \mid \mid p(\bbeta \mid \bz)\}]$ due to the chain rule for the $\textsc{kl}$ divergence.  Therefore, the second summand is $0$ if and only if $q^*_{\textsc{pfm}}(\bbeta \mid\bz)=p(\bbeta \mid \bz)$. 
The expressions for $q_{\textsc{pfm}}^*(z_i)$ $(i=1, \ldots, n)$, are instead a consequence of the closure under conditioning property of multivariate truncated Gaussians \citep{ horrace_2005} which allows to recognize the kernel of a univariate truncated normal in the  optimal solution $\exp[{E}_{q_{\textsc{pfm}}^*(\bz_{-i})} \{\log p(z_i \mid \bz_{-i},\by)\}]$ \citep{blei_2017}  for $q_{\textsc{pfm}}^*(z_i)$; see the Supplementary Material for the detailed proof. Algorithm \ref{algo2} outlines the steps of the optimization strategy to obtain ${q}^*_{\textsc{pfm}}(\bbeta,\bz)$. 

\begin{algorithm*}[t]
	\caption{Steps of the procedure to obtain ${q}^*_{\textsc{pfm}}(\bbeta,\bz)={q}^*_{\textsc{pfm}}(\bbeta \mid \bz) \prod_{i=1}^n {q}^*_{\textsc{pfm}}(z_i)$.} 
	\label{algo2}
	\vspace{3pt}	
	{[\bf 1]} {\small For each $i=1, \ldots, n$, set $\sigma^{*2}_i=(1-\bx^{\intercal}_i\bV \bx_i)^{-1}$ and initialize  $\bar{z}^{(0)}_i\in\mathbb{R}$.}\\
	\vspace{1pt}
	{{[\bf 2]}      \For(){\small $t$ \mbox{from} $1$ until convergence of ${\normalfont  \textsc{elbo}}\{q_{\textsc{pfm}}^{(t)}(\bbeta, \bz)\}$}
		{
			\vspace{1pt}
			\For(){\small $i$ \mbox{from} $1$ \mbox{to} $n$}
			{ \vspace{1pt}
				\small	Set $\mu^{(t)}_i=\sigma^{*2}_i \bx_i^{\intercal}\bV\bX^{\intercal}_{-i}(\bar{z}^{(t)}_1, \ldots, \bar{z}^{(t)}_{i-1}, \bar{z}^{(t-1)}_{i+1}, \ldots,  \bar{z}^{(t-1)}_{n})^{\intercal}$.\\
				\vspace{1pt}
				Set $\bar{z}^{(t)}_i=\mu^{(t)}_i+(2y_i-1)\sigma^{*}_i\phi(\mu^{(t)}_i/\sigma^{*}_i)/\Phi\{(2y_i-1)\mu^{(t)}_i/\sigma^{*}_i\}$. }
		}
	}
	\vspace{1pt}
	{\bf Output:} {\small Optimal partially-factorized approximations ${q}^*_{\textsc{pfm}}(\bbeta \mid \bz)= \phi_p(\bbeta-\bV\bX^\intercal\bz; \bV)$ and\\
		\quad \enspace  $q_{\textsc{pfm}}^*(z_i) ={1}\{(2y_i-1)z_i>0\}\phi(z_i-\mu^{*}_i; \sigma^{*2}_i)/\Phi\{(2y_i-1)\mu^{*}_i/\sigma^{*}_i\}$ $(i=1, \ldots, n)$.}
\end{algorithm*}

As for classical coordinate ascent variational inference, this routine optimizes the \textsc{elbo} sequentially with respect to  each density $q_{\textsc{pfm}}(z_i)$, keeping fixed the others  at their most recent update, thus producing a strategy that iteratively solves the system of equations for $\bmu^*$ in Theorem \ref{teo_2} via simple  expressions. Indeed, since the form of the approximating densities is already available as in Theorem \ref{teo_2}, Algorithm \ref{algo2}  reduces to update the vector of parameters $\bmu^*$ via simple functions and matrix operations; see the Supplementary Material for the expression of  ${\textsc{elbo}}\{q_{\textsc{pfm}}(\bbeta, \bz)\}$.

As stated in Corollary \ref{cor_SUN},  the optimal $q_{\textsc{pfm}}^*(\bbeta)$ of  interest can be derived from $q_{\textsc{pfm}}^*(\bbeta \mid \bz)$ and $\prod_{i=1}^nq_{\textsc{pfm}}^*(z_i)$, and  coincides with a tractable  \textsc{sun}  density \citep{arellano_2006}.
\begin{corollary}
	Let $\bar{\bY}=\mbox{\normalfont  diag}(2y_1-1, \ldots, 2y_n-1)$ and $\bsigma^{*}=\mbox{\normalfont diag}(\sigma^*_1, \ldots, \sigma^*_n)$, then,	under \eqref{eq8}, the approximate density $q_{\textsc{pfm}}^*(\bbeta)$  for $\bbeta$  coincides with that of the variable
	\begin{eqnarray}
	\bu^{(0)}+\bV\bX^{\intercal}\bar{\bY} \bsigma^* \bu^{(1)},
	\label{eq9}
	\end{eqnarray}
	where $\bu^{(0)} \sim {\normalfont \mbox{N}}_p(\bV\bX^{\intercal} \bmu^*,\bV)$, whereas $ \bu^{(1)}=(u^{(1)}_{1}, \ldots, u^{(1)}_{n})^\intercal$ with  $u^{(1)}_{i} \sim {\normalfont \textsc{tn}}[0,1,\{-(2y_i-1)\mu^*_i/\sigma^*_i, +\infty\} ]$, independently for every $i=1, \ldots, n$. Equivalently, $q_{\textsc{pfm}}^*(\bbeta)$ is the density of the $\textsc{sun}_{p,n}(\bxi, \bOmega, \bDelta, \bgamma, \bGamma)$ variable with parameters $\bxi=\bV\bX^{\intercal}\bmu^*, \bOmega=\bV+\bV\bX^{\intercal}\bsigma^{*2}\bX \bV, \bDelta=\bomega^{-1}\bV\bX^{\intercal}\bar{\bY}\bsigma^*, \bgamma=\bar{\bY}\bsigma^{*}\:^{-1}\bmu^*,  \bGamma=\bI_n$, where $\bomega$ denotes a $p \times p$ diagonal matrix containing the square roots of the diagonal elements in $\bOmega$, whereas $\bar{\bOmega}$ is the associated correlation matrix. 
	\label{cor_SUN}
\end{corollary}

Corollary \ref{cor_SUN} follows by noticing that,  under \eqref{eq8}, the approximate density for $\bbeta$ is the convolution  of a $p$-variate Gaussian and an $n$-variate truncated normal, thereby producing the density of a \textsc{sun} \citep{arellano_2006}. This class of variables generalizes the multivariate Gaussian family via a skewness-inducing mechanism controlled by the matrix $\bDelta$ which weights the skewing effect  produced by an $n$-variate truncated normal with covariance matrix $\bGamma$. Besides introducing asymmetric shapes, \textsc{sun}s have several closure properties which facilitate inference. However, evaluation of functionals requires the calculation of cumulative distribution functions of $n$-variate Gaussians \citep{arellano_2006, azzalini_2013}, which is prohibitive when $n$ is large, unless $\bGamma$ is diagonal. Recalling \citet{durante_2019}, this issue makes Bayesian inference rapidly impractical under the exact posterior $p(\bbeta \mid \by)$ when $n$ is more than a few hundreds, since $p(\bbeta \mid \by)$ is a  $\textsc{sun}$ density with non-diagonal $\bGamma_{\mbox{\scriptsize post}}$. Instead, the factorized form $\prod_{i=1}^n q_{\textsc{pfm}}(z_i)$ for $q_{\textsc{pfm}}(\bz)$ leads to a \textsc{sun} approximate density for $\bbeta$ in Corollary \ref{cor_SUN} which crucially relies on a diagonal $\bGamma=\bI_n$. This result allows approximate posterior inference  for every $n$ and $p$ via tractable expressions. In particular, exploiting Theorem~\ref{teo_2} and Corollary~\ref{cor_SUN}, the first two central moments of $\bbeta$ and the predictive probabilities are derived in Proposition~\ref{prop1}.

\begin{proposition}
	\label{prop1}
	Let $q_{\textsc{pfm}}^*(\bbeta)$ be the \textsc{sun} density in Corollary \ref{cor_SUN}, then
	\begin{eqnarray}
	\begin{split}
	{E}_{q_{\textsc{pfm}}^*(\bbeta)}(\bbeta)&=\bV \bX^{\intercal}\bar{\bz}^*, \\
	\mbox{\normalfont var}_{q_{\textsc{pfm}}^*(\bbeta)}(\bbeta)&=\bV+\bV\bX^{\intercal}\mbox{\normalfont diag}\{\sigma_1^{*2}-(\bar{z}_1^*-\mu_1^*)\bar{z}_1^*, \ldots, \sigma_n^{*2}-(\bar{z}_n^*-\mu_n^*)\bar{z}_n^*\}\bX \bV,
	\end{split}
	\label{eq10}
	\end{eqnarray}
	where $\bar{z}_i^*$, $\mu_i^*$ and $\sigma_i^{*}$ $(i=1, \ldots, n)$ are defined as in Theorem \ref{teo_2}. Moreover, the predictive probability $\mbox{\normalfont pr}_{\textsc{pfm}}(y_{\textsc{new}}=1\mid \by)={E}_{q^*_{\textsc{pfm}}(\bbeta)}\{\Phi(\bx^{\intercal}_{\textsc{new}} \bbeta)\}$ for a new unit with covariates $\bx_{\textsc{new}}$ is
	\begin{eqnarray}
	\mbox{\normalfont pr}_{\textsc{pfm}}(y_{\textsc{new}}=1\mid \by)={E}_{q^*_{\textsc{pfm}}(\bz)}[\Phi\{\bx_{\textsc{new}}^{\intercal}\bV\bX^{\intercal}\bz(1+\bx_{\textsc{new}}^{\intercal}\bV\bx_{\textsc{new}})^{-1/2}\}],
	\label{eq11}
	\end{eqnarray}
	where $q^*_{\textsc{pfm}}(\bz)$ can be expressed as the product $\prod_{i=1}^{n}q^*_{\textsc{pfm}}(z_i)$ of univariate truncated normal densities $q^*_{\textsc{pfm}}(z_i)=1\{(2y_i-1)z_i>0\}\phi(z_i-\mu^*_i; \sigma^{*2}_i)/\Phi\{(2y_i-1)\mu^*_i/\sigma^{*}_i\}$ $(i=1, \ldots, n)$.
\end{proposition}

Hence, unlike for inference under the exact posterior \citep{durante_2019}, the calculation of relevant approximate moments such as those in \eqref{eq10}, only requires evaluation of cumulative distribution functions of univariate Gaussians. Similarly, the predictive probabilities in  \eqref{eq11} can be easily evaluated via efficient Monte Carlo methods based on samples from $n$ independent univariate truncated normals with density $q^*_{\textsc{pfm}}(z_i)$ $(i=1, \ldots, n)$. Moreover, leveraging \eqref{eq9}, samples from $q^*_{\textsc{pfm}}(\bbeta)$ can be directly obtained via a linear combination between realizations from a $p$-variate Gaussian and from $n$ univariate truncated normals, as shown in Algorithm \ref{algo3}.  

\begin{algorithm*}
	\caption{\mbox{Strategy to sample a value from the approximate \textsc{sun} posterior in Corollary \ref{cor_SUN}.}}
	\label{algo3}
	\vspace{3pt}
	{[\bf 1]} {\small Draw $\bu^{(0)} \sim \mbox{N}_p(\bV\bX^{\intercal} \bmu^*,\bV)$.}\\
	{[\bf 2]} {\small Draw $u^{(1)}_{i} \sim \textsc{tn}[0,1,\{-(2y_i-1)\mu^*_i/\sigma^*_i, +\infty\} ]$  $(i=1, \ldots, n)$.}\\
	{[\bf 3]}  {\small Set $\bbeta=\bu^{(0)}+\bV\bX^{\intercal}\bar{\bY}\bsigma^* \bu^{(1)}$, with $\bu^{(1)}=(u^{(1)}_{1}, \ldots, u^{(1)}_{n})^\intercal$.}\\
	\vspace{3pt}
	
	{\bf Output:} {\small One sample of $\bbeta$ from the approximate posterior distribution with density $q^*_{\textsc{pfm}}(\bbeta)$.} 
\end{algorithm*}

This strategy  allows to study complex approximate functionals  through simple Monte Carlo methods. If the focus is only on  $q^*_{\textsc{pfm}}(\beta_j)$ $(j=1,\dots,p)$, one can avoid simulating from the joint $p$-variate Gaussian in Algorithm \ref{algo3} and just sample from the marginals of $\bu^{(0)}$ in the \textsc{sun} additive  representation  to produce draws from $q^*_{\textsc{pfm}}(\beta_j)$ $(j=1,\dots,p)$ at an  $\mathcal{O}(pn \cdot\min\{p,n\})$ cost. 

We conclude the presentation of our partially-factorized solution by studying its properties in high-dimensional settings as $p \rightarrow \infty$. As discussed in Section \ref{sec_2.1}, classical mean-field routines  \citep{consonni_2007} yield poor Gaussian approximations of the posterior density in high dimensions, which  affect quality of inference and prediction. By relaxing the mean-field assumption we obtain, instead, a skewed approximate density matching the exact posterior for $\bbeta$ when $p \rightarrow \infty$, as stated in Theorem \ref{teo_3}. 
We study asymptotic properties under the following assumption. 

\begin{assumption}
	Assume  $x_{ij}$ $(i=1, \ldots, n;j=1, \ldots, p)$, are random variables with $E(x_{ij})=0${,} ${E}(x_{ij}^2)=\sigma_x^2$, $E(x_{ij}^4)\leq c_x$ for some $c_x< \infty$, and such that for any $i=1,\dots,n$ and $i'\neq i$: 
	\begin{description}
		\item{(a) $p^{-1}\sum_{j=1}^p \mbox{\normalfont cov}(x_{ij},x_{i'j})\to 0$ as $p\to\infty$,}
		\item{(b) $p^{-2}\sum_{j=1}^p \sum_{j'\neq j} \mbox{\normalfont cov}(x_{ij}^2,x_{ij'}^2)\to 0$ as $p\to\infty$,}
		\item{(c) $p^{-2}\sum_{j=1}^p \sum_{j'\neq j} \mbox{\normalfont cov}(x_{ij}x_{i'j},x_{ij'}x_{i'j'})\to 0$ as $p\to\infty$.}
	\end{description}	
	\label{ass_3}
\end{assumption}

In Assumption \ref{ass_3} we do not require independence among predictors $\{x_{ij}\}$, but rather we impose asymptotic conditions on their covariance. Such conditions are trivially satisfied when predictors are independent and thus Assumption \ref{ass_3} is strictly weaker than Assumption \ref{ass_1} and incorporates the latter. The strongest assumption is arguably \ref{ass_3}(a), which requires the average covariance between   units to go to 0 as $p\to\infty$.  This assumption does not prevent statistical units from being strongly associated on a subset of predictors, but rather requires those to be a minority among all predictors, which is reasonable  in large $p$ scenarios. Conditions \ref{ass_3}(b)--\ref{ass_3}(c), which deal with covariance across predictors, are arguably much weaker and we expect them to be satisfied in almost any common practical scenario. In particular, these conditions allow for pairs of predictors with arbitrarily strong covariance, while requiring the number of such pairs to be $o(p^2)$ as $p\to\infty$. Also, the empirical results in Sections \ref{sec_simu} and \ref{sec_3} suggest that asymptotic exactness of our strategy for large $p$ might hold even beyond Assumption \ref{ass_3}; see also the Supplementary Material.

\begin{theorem}
	Under Assumptions \ref{ass_2}--\ref{ass_3},  we have that  $\textsc{kl}\{q_{\textsc{pfm}}^{*}(\bbeta) \mid \mid p(\bbeta \mid \by)\}\rightarrow 0$ in probability, as $p \rightarrow \infty$.
	\label{teo_3}
\end{theorem}

Hence, in the high-dimensional settings where current computational strategies are impractical  \citep{chopin_2017}, inference and prediction under the proposed partially-factorized variational approximation is practically feasible, and provides essentially the same results as those obtained under the exact posterior.  For instance, Corollary \ref{cor_predictive} states that, unlike for classical mean-field strategies,  our solution is guaranteed to provide increasingly accurate approximations of posterior predictive probabilities as $p\to\infty$. 

\begin{corollary}
	Let  $\mbox{\normalfont pr}(y_{\textsc{new}}=1\mid \by)=\int \Phi(\bx^{\intercal}_{\textsc{new}} \bbeta)p(\bbeta \mid \by)\mbox{\normalfont d}\bbeta$ be the exact posterior predictive probability for a new observation with predictors $\bx_{\textsc{new}} \in \mathbb{R}^p$. Then, under Assumptions  \ref{ass_2}--\ref{ass_3}, ${\sup_{\bx_{\textsc{new}}\in\mathbb{R}^p}}|\mbox{\normalfont pr}_{\textsc{pfm}}(y_{\textsc{new}}= 1\mid \by)-\mbox{\normalfont pr}(y_{\textsc{new}}= 1\mid \by)|\rightarrow 0$ in probability,  as $p\to\infty$. 	\label{cor_predictive}
\end{corollary}

Corollary \ref{cor_predictive} implies that the error made by the proposed method in approximating the posterior predictive probabilities goes to $0$ as $p\to\infty$, regardless of the choice of $\bx_{\textsc{new}} \in \mathbb{R}^p$. On the contrary, since by Theorem \ref{teo_1}(c) $\sup_{\bx_{\textsc{new}}\in\mathbb{R}^p}|\mbox{\normalfont pr}_{\textsc{mf}}(y_{\textsc{new}}=1\mid \by)-\mbox{\normalfont pr}(y_{\textsc{new}}=1\mid \by)|$ is bounded away from $0$ almost surely as $p\to\infty$, there always exists, for every large $p$, some $\bx_{\textsc{new}}$ such that the corresponding posterior predictive probability $\mbox{pr}_{\textsc{mf}}(y_{\textsc{new}}=1\mid \by)$ approximated under classical mean-field variational Bayes does not match the exact one.

Finally, as stated in Theorem \ref{teo_4},  the number of iterations required by  Algorithm \ref{algo2} to produce the optimal solution $q_{\textsc{pfm}}^*(\bbeta)$ converges to $1$ as $p \rightarrow \infty$.

\begin{theorem}
	Let $q^{(t)}_{\textsc{pfm}}(\bbeta)$  be the approximate density for  $\bbeta$ produced at step $(t)$ by Algorithm \ref{algo2}. Then, under Assumptions \ref{ass_2}--\ref{ass_3}, $\textsc{kl}\{q_{\textsc{pfm}}^{(1)}(\bbeta) \mid \mid p(\bbeta \mid \by)\}\rightarrow 0$ in probability, as $p \rightarrow \infty$.
	\label{teo_4}
\end{theorem}

Theorem \ref{teo_4} suggests that Algorithm \ref{algo2} needs increasingly less iterations to converge as $p$ grows, requiring essentially one iteration as $p \rightarrow \infty$.  This is coherent with our numerical studies, where very few iterations are sufficient to reach convergence  in large $p$ cases. The computational complexity of Algorithm \ref{algo2} is thus equal to that of a fixed number of iterations, which is dominated by an $\mathcal{O}(pn {\cdot}\min\{p,n\})$ pre-computation cost derived in the Supplementary Material, where we also highlight how the evaluation of the functionals in Proposition \ref{prop1} can be achieved at the same cost. More complex functionals of the approximate posterior can be instead obtained at higher cost via Monte Carlo methods  based on Algorithm~\ref{algo3}.  As discussed in the Supplementary Material, Algorithm~\ref{algo1} has the same  $\mathcal{O}(pn {\cdot}\min\{p,n\})$ pre-computation cost and $\mathcal{O}(n {\cdot}\min\{p,n\})$ per-iteration cost of Algorithm~\ref{algo2}. However, recalling  Sections \ref{sec_2.1}--\ref{sec_2.2}, our partially-factorized solution produces substantially more accurate approximations than classical mean-field. In Sections \ref{sec_simu}--\ref{sec_3}, we provide quantitative evidence for these arguments, and discuss how the theory in  Section \ref{sec_2} matches closely the empirical behavior observed in simulations and applications.

\section{Simulation studies}\label{sec_simu}
To illustrate the advantages of the partially-factorized approximation, we consider simulation studies that assess empirical evidence in finite dimension for the asymptotic results in Section \ref{sec_2}, and quantify the magnitude of the improvements in scalability and accuracy relative to state-of-the-art competitors. Consistent with these goals and with our main focus on high dimensions, we compare performance against mean-field variational Bayes \citep{consonni_2007}, which provides the only widely-implemented alternative, among those discussed in Section~\ref{sec_1}, that can scale to high-dimensional settings as effectively as our method. We refer to Table \ref{table1}  for  empirical evidence on the computational intractability in high dimensions of the other  relevant strategies discussed in Section~\ref{sec_1} \citep{hoff_2014,chopin_2017,durante_2019}.

Consistent with the above goals, we generate data under assumptions compatible with those in Section \ref{sec_2}. In particular, we simulate each  $y_{i} \in \{0;1\}$, for $i=1, \ldots, n$, from a $\mbox{Bern}\{\Phi(\bx^{\intercal}_i\bbeta)\}$, where $\bx_i=(1, x_{i2}, \ldots, x_{ip})^{\intercal}$ is a $p$-dimensional vector of predictors, comprising an intercept term and $p-1$ covariates simulated from standard normal $\mbox{N}(0,1)$ variables, independently for $i=1, \ldots, n$ and $j=2, \ldots, p$, and then standardized to have mean 0 and standard deviation 0.5, as suggested in \citet{gelman_2008} and \citet{chopin_2017}. Following the remarks in  \citet{gelman_2008}, the coefficients $\beta_1, \ldots, \beta_p$ in $\bbeta$ are instead simulated from $p$ independent uniforms on the range $[-5,5]$. To carefully illustrate performance under varying $(n,p)$ settings, we simulate such a dataset under three different sample sizes $n \in \{50; 100; 250\}$ and, for each $n$, we evaluate  performance at varying ratios $p/n \in \{0.5; 1; 2; 4\}$. These dimensions for $n$ and $p$ are much lower relative to those that can be easily handled under the partially-factorized and mean-field strategies analyzed. In fact, such moderate dimensions are required to make inference still feasible also under the widely-used \textsc{stan} implementation of Hamiltonian Monte Carlo \citep{hoff_2014}, which serves here as a benchmark to assess the accuracy of  the two variational methods in approximating key functionals of the exact posterior.

In performing Bayesian inference for the above scenarios, we  implement the probit regression model in~\eqref{eq1} under independent weakly informative Gaussian priors with mean $0$ and standard deviation $5$ for each $\beta_j$ $(j=1, \ldots, p)$ \citep{gelman_2008,chopin_2017}. Such priors are then updated with the probit likelihood of the simulated data to obtain a posterior for $\bbeta$ in each of the twelve $(n,p)$ settings, which we approximate with both mean-field variational Bayes and our novel partially-factorized solution. Figure \ref{fsimu1} illustrates the quality of these two approximations with a focus on three key functionals of the posterior for $\bbeta$, taking as benchmark the Monte Carlo estimates of such functionals computed from 20000 posterior samples of $\bbeta$, under the \texttt{rstan} implementation of the Hamiltonian Monte Carlo. The latter strategy leads to accurate estimates of the exact posterior functionals but, as highlighted in Section \ref{sec_1} and in Table~\ref{table1}, is clearly impractical in high dimensions. Therefore, it is of interest to quantify how far the approximations provided by the two more scalable variational methods are from  the \textsc{stan} Monte Carlo estimates, under different $(n,p)$ settings. In Figure \ref{fsimu1}, this assessment focuses on the first two posterior moments of each $\beta_j$ $(j=1, \ldots,p)$, and on the out-of-sample predictive probabilities of newly simulated  units $i'=1, \ldots, 100$.  For such functionals, we represent the median and the quartiles of the absolute differences between the corresponding \textsc{stan} Monte Carlo estimates and  the approximations provided by the two variational methods, at varying combinations of $(n,p)$. In the first and second panel, the three quartiles are computed from the $p$ absolute differences among the Monte Carlo and approximate estimates for the two posterior moments of each $\beta_j$ $(j=1, \ldots, p)$, respectively, whereas in the third panel these summary measures are computed from the $100$ absolute differences for the   predictive probabilities of the newly simulated units $i'=1, \ldots, 100$.

\begin{figure}[t]
	\captionsetup{font={small}}
	\centering
	\includegraphics[width=\linewidth]{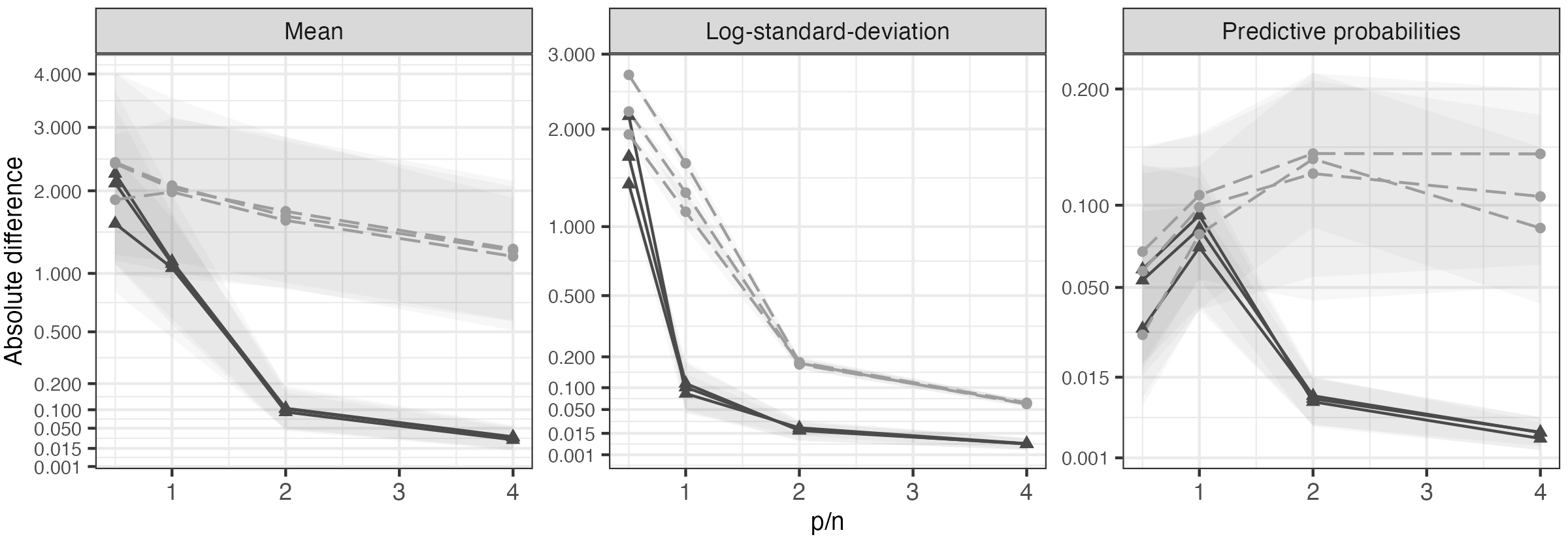}
	\caption{Accuracy in the approximation for three key  functionals of the posterior distribution for $\bbeta$. Trajectories for the median of  the  absolute differences between an accurate but expensive Monte Carlo estimate of such functionals and their approximations provided by  partially-factorized variational Bayes (dark grey solid lines) and mean-field variational Bayes (light grey dashed lines), respectively, for increasing values of the ratio $p/n$, i.e.,  $p/n \in \{0.5;1;2;4\}$. The different trajectories for each of the two methods correspond to three different settings of the sample size $n$,  i.e.,  $n  \in \{50;100;250\}$, whereas the grey areas denote the first and third quartiles computed from the absolute differences. For graphical purposes we consider a square-root scale for the $y$-axis.}
	\label{fsimu1}
\end{figure}

According to Figure \ref{fsimu1}, the partially-factorized solution uniformly improves the accuracy of  the mean-field approximation, and the magnitude of such gains  increases as $p$ grows relative to $n$. Notably, the median and the two quartiles of the absolute differences between the \textsc{stan}  estimates and  those provided by our novel partially-factorized approximation rapidly shrink toward $0$ for all the functionals under analysis, when $p>n$. This means that in such settings the proposed strategy matches almost perfectly the functionals of interest of the exact posterior, thus providing the same quality in inference and prediction as the one obtained, at massively higher computational costs, under state-of-the-art  Hamiltonian Monte Carlo implementations. For instance, while our approximation can be obtained in fractions of seconds in the setting $(n=250, p=1000)$, \textsc{stan} requires several hours, even when considering parallel implementations. On the contrary, the classical mean-field approximation always leads to biased estimates, even when $p>n$, not only for the scales, but also for the locations and the predictive probabilities. As  we will clarify in the  application in Section \ref{sec_3}, this bias can have dramatic consequences for inference and prediction. 

All the above empirical findings are fully coherent with the theoretical results presented in  Sections \ref{sec_2.1}--\ref{sec_2.2} and provide consistent evidence that the behaviors proved in asymptotic regimes are clearly visibile even in finite-dimensional $p>n$ settings, without requiring $p$ to be massively higher than $n$. Also Theorem \ref{teo_4} finds evidence in the simulation studies since, for all the $p>n$ settings under analysis, Algorithm~\ref{algo2} requires at most $6$ iterations to reach convergence, whereas Algorithm~\ref{algo1}  needs on average $103$ iterations. Finally, the almost perfect overlap in Figure  \ref{fsimu1}  of the trajectories for the different $n$ settings suggests that the quality of the two variational approximations may depend on $(p,n)$ just through the ratio $p/n$; see Section \ref{sec_4} for further discussion. 

As shown in Figure \ref{fsimu2} of the Supplementary Material, the above findings also apply to general data structures beyond the Assumptions~\ref{ass_1}--\ref{ass_3} we require to prove theory in Section~\ref{sec_2}. In fact, we obtain similar conclusions also when inducing either structured decaying pairwise correlation $\mbox{cor}(x_{ij},x_{i'j})=0.75^{|i-i'|}$ over the rows of $X$, or strong pairwise correlation $\mbox{cor}(x_{ij},x_{ij'})=0.75$ across the columns of $X$. Since in these simulations $\mbox{var}(x_{ij})=1$, then correlations coincide with covariances. Such results motivate additional explorations in real-world studies in Section~\ref{sec_3}.

\section{High-dimensional probit regression for medical data}\label{sec_3}
As shown in \citet{chopin_2017},  state-of-the-art computational methods for Bayesian binary regression, such as Hamiltonian Monte Carlo \citep{hoff_2014}, mean-field variational Bayes \citep{consonni_2007} and expectation-propagation \citep{chopin_2017} are feasible and powerful procedures in small-to-moderate $p$ settings, but tend to become rapidly impractical or inaccurate in large $p$ contexts, such as $p > 1000$.  The overarching focus of the present article is to close this gap and, consistent with this aim, we start by considering a large $p$ study to quantify the drawbacks encountered by the aforementioned strategies  along with the improvements provided by the proposed partially-factorized method in real-world applications.

Following the above remarks, we model  presence or absence of the Alzheimer's disease in its early stages as a function of demographic data, genotype and assay results; the source dataset is available in the \texttt{R} library \texttt{AppliedPredictiveModeling} \citep{craig_2011}. In the original article, the authors consider a variety of machine learning procedures to improve  flexibility relative  to a basic binary regression model. Here, we avoid excessively complex black-box algorithms and rely on an interpretable probit regression as in \eqref{eq1}, that improves flexibility by simply adding pairwise interactions, thus obtaining $p=9036$ predictors, including the intercept, collected for $333$ individuals. As done for the simulation studies in Section \ref{sec_simu}, the original measurements have been standardized to have mean $0$ and standard deviation $0.5$, before entering such variables and their interactions in the probit models.  Recalling Section \ref{sec_simu}, we recommend to always standardize the predictors and to include the intercept term when implementing the partially-factorized variational approximation  in real-world datasets. Besides being a common operation in routine implementations of probit regression  \citep[][]{gelman_2008,chopin_2017}, such a standardization typically reduces the covariance between units and thus also between the associated latent variables $z_i$, making the variational approximation more accurate. 

\begin{table}[b]
	\captionsetup{font={small}}
	\caption{Computational time of state-of-the-art routines in the Alzheimer's study. This includes the running time of the sampling or optimization procedure, and the time to compute means, standard deviations and predictive probabilities, for those routines that were feasible. \textsc{hmc}, Hamiltonian Monte Carlo; \textsc{ep}, expectation-propagation; \textsc{sun}, i.i.d.\ sampler from the  \textsc{sun} posterior; \textsc{mf}-\textsc{vb}, mean-field variational Bayes; \textsc{pfm}-\textsc{vb}, partially-factorized variational Bayes.
		\label{table1}}
	\centering
	\begin{tabular}{lccccc}
		& \textsc{stan} & \ \textsc{ep} & \ \textsc{sun} & \ \textsc{mf}-\textsc{vb} & \ \textsc{pfm}-\textsc{vb} \\ 
		Computational time in minutes $ \ $ &$>360.00$ &  \ $>360.00$ & \  $92.27$ & \ $0.06$ &  \ $0.06$ \\ 
	\end{tabular}
\end{table}

As for the simulation study, we perform Bayesian inference  by relying on independent Gaussian priors with mean $0$ and standard deviation $5$ for each $\beta_j$ $(j=1, \ldots, 9036)$ \citep{gelman_2008}. These priors  are updated with the probit likelihood of $n=300$ units, after holding out $33$ individuals to study the behavior of the posterior predictive probabilities in such large $p$ settings, along with the performance of the overall approximation of the posterior. Table \ref{table1} provides insights on the computational time of the two variational approximations, and highlights bottlenecks encountered by relevant routine-use competitors. These include the \texttt{rstan} implementation of Hamiltonian Monte Carlo, the expectation-propagation algorithm in the \texttt{R} library  \texttt{EPGLM},  and the Monte Carlo strategy based on $20000$ independent draws from the exact \textsc{sun} posterior using the algorithm in \citet{durante_2019}. As expected, these strategies are  impractical in such settings. In particular, \textsc{stan} and  expectation-propagation suffer from the large $p$, whereas sampling from the exact posterior is still feasible, but requires a non-negligible computational effort due to the moderately large $n$. Variational inference is orders of magnitude faster, thus providing the only viable approach in this study. Such results motivate our main  focus on the quality of the two variational approximations, taking as a benchmark Monte Carlo inference based on $20000$ independent samples from the exact posterior \citep{durante_2019}. In this example, Algorithm~\ref{algo2} requires only $6$  iterations to converge, instead of $212$ as for Algorithm~\ref{algo1}. This is in line with Theorem \ref{teo_4}.

As a first assessment, we evaluate the accuracy in approximating the whole posterior distributions by studying Monte Carlo estimates of the Wasserstein distances between the $p=9036$ exact posterior marginals and the associated approximations under the two variational methods.  Such quantities are computed with the \texttt{R} function \texttt{wasserstein1d}, which uses $20000$ values sampled  from the approximate  and exact marginals. The results of this analysis provide strong support in favor of the proposed partially-factorized solution that yields an overall Wasserstein distance, averaged across the $p=9036$ coefficients, of $0.07$, one order of magnitude lower than the one obtained under classical mean-field variational Bayes, which is $0.47$. This finding is in line with Proposition~\ref{rem_1}. More crucially, consistent with Theorem~\ref{teo_3}, the $94.2\%$ of the $p=9036$ Wasserstein distances under the partially-factorized approximation are within the quantiles $2.5\%$ and $97.5\%$ of the Wasserstein distances between two different  samples of 20000 draws from the same exact posterior marginals, meaning that,  in practice, our solution matches almost perfectly the exact posterior, with most of the variability  coming arguably from Monte Carlo error. This is not the case under the classical mean-field strategy for which the percentage of Wasserstain distances within the two quantiles notably drops from $94.2\%$ to $15.9\%$. The high accuracy of the new approximation is achieved despite the presence of 5000 pairs of highly associated predictors with correlation above 0.85, thus confirming its quality in real settings, beyond Assumption~\ref{ass_3}.

Consistent with Theorem \ref{teo_1}, the low quality of the classical mean-field solution is mainly due to a tendency to over-shrink towards $0$ the locations of the actual posterior. This is evident in the first panel of Figure \ref{f3}, that compares the posterior expectations computed from 20000 values sampled from the exact \textsc{sun} posterior with those provided by the closed-form expressions under the two competing variational approximations. Also the standard deviations are  under-estimated relative to our novel approximation, that notably removes bias also in the second order moments. Consistent with our previous results, the slight variability of the estimates provided by our solution in the second panel of Figure \ref{f3}  is arguably due to Monte Carlo error. In Figure  \ref{f3}, we also assess the quality in the approximation of the exact posterior predictive probabilities for the $33$ held-out individuals. Such measures are fundamental for prediction and, unlike for the first two marginal moments, their evaluation depends on the behavior of the entire posterior, since they rely on a non-linear mapping of a linear combination of the parameters $\bbeta$. In the third panel of Figure \ref{f3},  the proposed approximation essentially matches the exact posterior predictive probabilities, thus providing reliable classification and uncertainty quantification. Instead, as expected from the theoretical results in  Theorem  \ref{teo_1},  mean-field over-shrinks these quantities towards $0.5$. This leads to a test deviance $-\sum_{i=1}^{33} [y_{i,\textsc{new}}\log \hat{\mbox{pr}}(y_{i,\textsc{new}} =1\mid \by)+(1-y_{i,\textsc{new}})\log\{1-\hat{\mbox{pr}}(y_{i,\textsc{new}}=1 \mid \by)\}]$ for mean-field variational Bayes of $22.31$, dramatically increasing the test deviance of $14.62$ under the proposed partially-factorized solution. The latter is, instead, in line with the test deviance of $14.58$ obtained via Monte Carlo using independent samples from the exact  \textsc{sun} posterior.

\begin{figure}
	\captionsetup{font={small}}
	\centering
	\includegraphics[width=\linewidth]{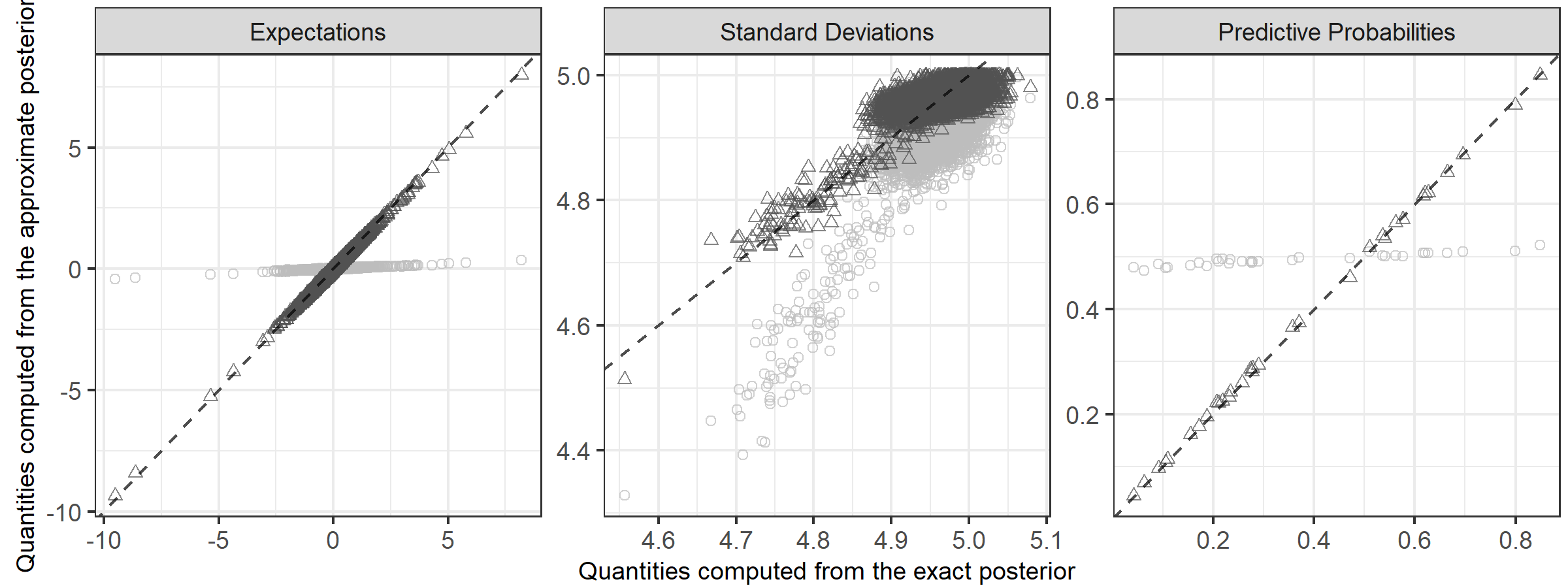}
	\caption{Scatterplots comparing the posterior expectations, standard deviations and predictive probabilities computed from 20000 values sampled from the exact \textsc{sun} posterior, with those provided by the mean-field variational Bayes (light grey circles) and partially-factorized variational Bayes (dark grey triangles), under the setting $\nu^2_p=25$.}
	\label{f3}
\end{figure}
\begin{figure}
	\captionsetup{font={small}}
	\centering
	\includegraphics[width=\linewidth]{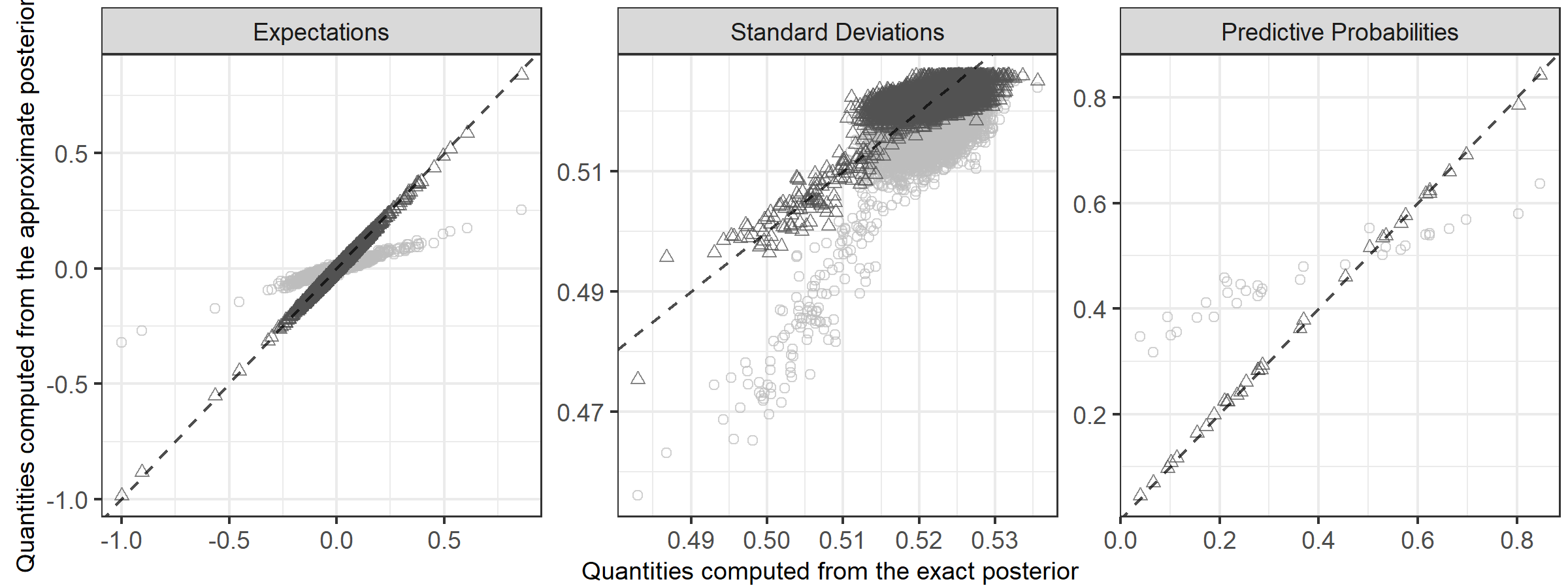}
	\caption{Scatterplots comparing the posterior expectations, standard deviations and predictive probabilities computed from 20000 values sampled from the exact \textsc{sun} posterior, with those provided by the mean-field variational Bayes (light grey circles) and partially-factorized variational Bayes (dark grey triangles), under a different setting for   $\nu^2_p$ controlling the variance of the linear predictor via $\nu^2_p=25 \cdot 100/p$, to induce increasing shrinkage.}
	\label{f4}
\end{figure}

The above conclusions are confirmed also in Figure~\ref{f4}, when setting a prior variance of order $\mathcal{O}(p^{-1})$, namely $\nu_p^2 =25\cdot 100/p$.  This choice controls, heuristically, the total variance of the linear predictor as if there were $100$ coefficients, out of $9036$,  with prior variance  $25$, while the others were fixed to zero, thus  inducing increasing shrinkage in high dimensions. As a consequence of this strong shrinkage effect also the exact posterior means increasingly shrink towards 0, thus mitigating the issues of classical mean-field approximations which, however, still maintain a bias that is not present under the proposed partially-factorized solution, in all settings. 

\begin{table}[b]
	\captionsetup{font={small}}
	\caption{For the  three methods analyzed, five-fold cross-validation estimates of the test deviance under different medical datasets. Italics values denote best performance. \textsc{mf}-\textsc{vb}, mean-field variational Bayes; \textsc{pfm}-\textsc{vb}, partially-factorized variational Bayes; \textsc{svb} sparse variational Bayes.
		\label{table2}}
	\begin{tabular}{lcccc}
		&\qquad  {\texttt{parkinson}}  &{ \texttt{voice}} & { \texttt{lesion}} &  { \texttt{alzheimer}} \\ 
		&\qquad {\footnotesize$(n=756,p=754)$} &{\footnotesize\ \ $(n=126,p=310)$} & {\footnotesize\ \ $(n=76,p=952)$}& {\footnotesize\ \ $(n=333,p=9036)$}\\ [5pt]
		\textsc{mf}-\textsc{vb}&  \qquad $309.82$ &  $59.38$ & $48.66$ & $228.71$ \\ 
		\textsc{pfm}-\textsc{vb} & \qquad ${\it 306.37}$ &  $46.35$ &   ${\it 27.24}$ & $187.52$ \\[2pt]
		\textsc{svb} & \qquad $317.42$ &  ${\it 42.16}$ &  $35.81$ & ${ \it126.24}$ \\  
	\end{tabular}
\end{table}

We conclude by providing further evidence for the behavior of the two variational approximations  in settings where the focus is accurate prediction. With this goal in mind, we consider the Alzheimers' study and three additional medical datasets \citep{sakar2019,tsanas2013,mesejo2016} with different combinations of $(n,p)$ available at the UCI Machine Learning Repository ({{\texttt{https://archive.ics.uci.edu/ml/datasets.php}}}). To each of these datasets we apply model \eqref{eq1} with the same settings initially considered for the Alzheimer's study, and evaluate predictive performance of the classical mean-field approximation and the proposed partially-factorized alternative. Moreover, we also compare to the sparse variational Bayes (\textsc{svb}) approximations for spike-and-slab Bayesian  logistic regression by \citet{ray2020spike}, with intercept and zero-mean Gaussian slabs having standard deviation $5$, so as to facilitate comparison. Unlike for the mean-field approximation and the proposed partially-factorized modification, this alternative competitor relies on a different model and, as a consequence, it approximates a different posterior. Hence, while changes in predictive performance between the first two strategies can be directly interpreted in terms of approximation accuracy, it is not possible to separate the quality of the approximation from model performance in  the comparison with sparse variational Bayes. Nonetheless, assessing predictive power against this strategy is useful in studying the effectiveness of the ridge-type shrinkage induced by model \eqref{eq1} relative to a sparse spike-and-slab one. Following common practice in machine learning, we consider five-fold cross-validation estimates $-\sum_{i=1}^n [y_i\log \hat{\mbox{pr}}(y_i =1\mid \by_{\textsc{fold}(-i)})+(1-y_i)\log\{1-\hat{\mbox{pr}}(y_i=1 \mid \by_{\textsc{fold}(-i)})\}]$  of test deviances under the three competitors, where $ \by_{\textsc{fold}(-i)}$ denotes the response vector for all the individuals that are not in the same fold of unit $i$. Such measures are reported in Table \ref{table2}, and confirm the higher quality of the partially-factorized solution relative to the mean-field one.  As expected, these gains are found in all the datasets and become substantial as $p$ grows relative to $n$. The comparison against sparse variational Bayes is, instead, less obvious since the two types of shrinkage induced are different. In fact, in Table~\ref{table2}, the sparse variational Bayes solution and the proposed  partially-factorized approximation are both competitive and the improved performance of one over the other inherently depends on the dataset analyzed, without a clear pattern in relation to $p/n$. These conclusions were confirmed when replicating the analyses under different settings of $\nu_p^2$, namely $\nu_p^2=25 \cdot 100/p$ and $\nu_p^2=25 \cdot  10/p$. More specifically, in 12 scenarios, corresponding to $3$ prior choices for $4$ datasets, the proposed partially-factorized solution ranked first $8$ times, whereas sparse variational Bayes over-performed the other methods the remaining  $4$ times. Considering Laplace slabs instead of Gaussian ones when implementing sparse variational Bayes in some of these experiments did not substantially change the final conclusions.

\section{Future research directions}
\label{sec_4}
While our contribution provides an important advancement in a non-Gaussian regression context where previously available Bayesian computational strategies are unsatisfactory \citep{chopin_2017}, the results in this article open several avenues for future research. For instance, the theory on mean-field approximations and mode estimators presented in Section~\ref{sec_2.1} for large $p$ settings points to the need of further theoretical studies on the use of such quantities in high-dimensional regression with non-Gaussian responses. In these contexts, our idea of relying on a partially-factorized approximating family could provide a general strategy to solve potential issues of current global-local approximations; see \citet{cao2020} and \citet{fasano2020} for ongoing research in this direction that has been motivated by the ideas and methods developed in this article. For instance, in the context of multinomial probit models,  \citet{fasano2020} further enlarge the partially-factorized mean-field family by replacing the independence assumption among all the unit-specific augmented data $z_i$, for $i=1, \ldots, n$, with a relaxed version which only assumes  independence between pre-specified low-dimensional blocks of units, while preserving dependence within each block. This leads to low-dimensional truncated normal approximating densities for the blocks of augmented data which preserve tractability, while further improving approximation accuracy. This strategy can be employed also in univariate probit regression and, recalling Section~\ref{sec_2}, is expected to yield additional gains, especially when grouping units with highly correlated predictors. We shall also emphasize that, although our focus on Gaussian priors with variance possibly decreasing with $p$ already allows to enforce shrinkage in high dimensions, this  setting also provides a key building-block which can be included in more complex scale mixtures of Gaussians priors to obtain improved theoretical and practical performance in state-of-the-art approximations under sparse settings. The comparison against sparse variational Bayes  \citep{ray2020spike} in Section~\ref{sec_3}, suggests that this is a relevant direction.

Finally, it would interesting to extend  Theorems \ref{teo_1}, \ref{teo_3} and \ref{teo_4} to settings in which $n$ grows with $p$ at some rate. Motivated by  Figure \ref{fsimu1}, we expect that $n$ growing sublinearly with $p$ is a sufficient condition to obtain asymptotic-exactness results analogous to Theorem \ref{teo_3}.  However, the proof of Theorem \ref{teo_3} exploits the fact that the dimension of  $z$ is fixed in our asymptotic regime, and thus would not directly extend to a regime where $n$ grows with $p$. One possibility to obtain such an extension would be to derive quantitative bounds on the  \textsc{kl} divergence between truncated normal distributions with growing dimension, but we are not aware of similar bounds in the literature.

\section*{Acknowledgement}
Augusto Fasano is further affiliated  to the de Castro Statistics Initiative at Collegio Carlo Alberto, Turin.
Daniele Durante is also affiliated to the Bocconi Institute for Data Science and Analytics (BIDSA) at Bocconi University, Milan, and acknowledges the support from MIUR–PRIN 2017 project 20177BRJXS. Giacomo Zanella is affiliated to the Innocenzo Gasparini Institute for Economic Research (IGIER) and the Bocconi Institute for Data Science and Analytics (BIDSA) at Bocconi University, Milan, and acknowledges the support from MIUR–PRIN project 2015SNS29B.

\appendix
\section*{Supplementary Material}
The Supplementary Material includes extended proofs of the theoretical results, details on computational complexity of the algorithms developed, additional simulations and further empirical evidence in real-world medical applications for the performance of the methods developed in the article.

\section{ELBO and computational cost of partially-factorized variational Bayes}
\label{app:comp_cost}

We discuss the computational cost of partially-factorized variational Bayes, showing that the whole routine requires matrix pre-computations with $\mathcal{O}(pn \cdot\min\{p,n\})$ cost and iterations with $\mathcal{O}(n\cdot \min\{p,n\})$ cost. Consider first Algorithm \ref{algo2}. When $p\geq n$, one can pre-compute $\bX\bV \bX^\intercal$ at $\mathcal{O}(pn^2)$ cost and then perform every iteration at $\mathcal{O}(n^2)$ cost. More specifically, by applying Woodbury's identity to $\bV$, it follows that \smash{$\bX\bV \bX^\intercal=\nu^{2}_pXX^\intercal(I_n+\nu^{2}_pXX^\intercal)^{-1}$}. Hence, the $\mathcal{O}(pn^2)$ cost of  matrix multiplication operations dominates the $\mathcal{O}(n^3)$ cost of inversion, when $p\geq n$. Instead, when $p<n$, one can pre-compute $\bX\bV$ at an $\mathcal{O}(p^2n)$ cost, and then perform each iteration at $\mathcal{O}(pn)$ cost noting that 
\begin{eqnarray*}
	\textstyle
	\mu^{(t)}_i=\sigma^{*2}_i \sum\nolimits_{j=1}^p(\bX\bV)_{ij}d^{(t)}_{ij}, \qquad d^{(t)}_{ij}=\sum\nolimits_{i'=1}^{i-1}x_{i' j}\bar{z}^{(t)}_{i'}+\sum\nolimits_{i'=i+1}^{n}x_{i' j}\bar{z}^{(t-1)}_{i'},
\end{eqnarray*}
where $d^{(t)}_{i}=(d^{(t)}_{i1}, \ldots, d^{(t)}_{ip})^{\intercal}$ can be computed at $\mathcal{O}(p)$ cost from $d_{i-1}^{(t)}$ via the recursive equation
\begin{eqnarray*}
	d^{(t)}_{ij}=d^{(t)}_{i-1,j}-x_{ij}\bar{z}^{(t-1)}_i+x_{i-1,j}\bar{z}^{(t)}_{i-1}.
\end{eqnarray*}
Therefore, computing $\mu^{(t)}_i$ $(i=1,\dots,n)$, which is the most expensive part of Algorithm \ref{algo2}, can be done in $\mathcal{O}(np)$ operations using $\bX\bV$ and \smash{$d_i^{(t)}$}. With simple calculations one can check that also computing the {\normalfont  \textsc{elbo}} requires $\mathcal{O}(n\cdot \min\{p,n\})$ operations, as it involves quadratic forms of $n\times n$ matrices with rank at most $\min\{p,n\}$. Indeed, recalling that \smash{$q_{\textsc{pfm}}^{*}(\beta \mid z)=p(\beta \mid z)$}, and letting $\bLambda = (\bI_n + \nu_p^2\bX\bX^\intercal)^{-1}$, it holds
\begin{eqnarray*}
	\begin{split}
		\textstyle
		&	{\normalfont  \textsc{elbo}}\{q_{\textsc{pfm}}^{(t)}(\bbeta, \bz)\} = {\normalfont  \textsc{elbo}}\{q_{\textsc{pfm}}^{(t)}(\bz)\}= \E_{q_{\textsc{pfm}}^{(t)}(\bz)}\{\log p(\bz,\by)\} - \E_{q_{\textsc{pfm}}^{(t)}(\bz)}\{\log q^{(t)}_{\textsc{pfm}}(\bz)\}\\
		\textstyle
		&  \textstyle = \mbox{const} + \sum\nolimits_{i=1}^n \log\Phi\{(2y_i-1)\mu_i^{(t)}/\sigma_i^*\}+ 0.5\E_{q_{\textsc{pfm}}^{(t)}(\bz)}\{(\bz-\mu^{(t)})^\intercal \bsigma^{*^{-1}} (\bz-\mu^{(t)})-\bz^{\intercal}\bLambda\bz\},\\
	\end{split}
\end{eqnarray*}	
where the calculation of the above expectation only requires computation of the means and variances for the $n$ univariate truncated normals obtained in step $t$ of Algorithm \ref{algo2}. Recalling basic properties of univariate truncated normals, both quantities are available in closed form as a function of \smash{$\mu^{(t)}_i$}, \smash{$\bar{z}^{(t)}_i$} and $\sigma^*_i$, $(i=1, \ldots, n)$, already computed in Algorithm \ref{algo2}.

Given the output of Algorithm \ref{algo2}, the expectation of $\bbeta$ under the partially-factorized approximation can be computed at $\mathcal{O}(p n \cdot  \min\{p,n\})$ cost noting that, by  \eqref{eq10}, 
$\E_{q_{\textsc{pfm}}^{*}(\bbeta)}(\bbeta)=\bV\bX^\intercal\bar{\bz}^*$, where $\bV\bX^\intercal$ can be computed at $\mathcal{O}(p n \cdot  \min\{p,n\})$ cost using either its definition, when $p\leq n$, or the equality \smash{$\bV\bX^\intercal = \nu_p^2\bX^\intercal\left(\bI_n+\nu_p^2\bX\bX^\intercal\right)^{-1}$}, when $p>n$. Given $\bV\bX^\intercal$, one can compute the covariance matrix of $\bbeta$ under partially-factorized variational Bayes  at $\mathcal{O}(p^2 n)$ cost using~\eqref{eq10}, and applying the Woodbury's identity to $\bV$, when $p> n$.
On the other hand, the marginal variances $\mbox{var}_{q_{\textsc{pfm}}^{*}(\beta_j)}(\beta_j)$ $(j=1,\dots, p)$, can be obtained at $\mathcal{O}(p n \cdot \min\{p,n\})$ cost by first  computing $\bV\bX^\intercal$, and then exploiting  \eqref{eq10} along with \smash{$V_{jj}=\nu_p^2\left\{1-\sum_{i=1}^n (\bV\bX^\intercal)_{ji} x_{ij} \right\}$}, which follows from  \smash{$\bV(\bI_p+\nu_p^2\bX^\intercal \bX)=\nu_p^2\bI_p$}.

Finally, the Monte Carlo estimates of the approximate predictive probabilities $\mbox{\normalfont pr}_{\textsc{pfm}}(y_{\textsc{new}}=1\mid \by)$ in \eqref{eq11} can be computed at $\mathcal{O}(pn \cdot \min\{p,n\}+nR)$ cost, where $R$ denotes the number of Monte Carlo samples. Indeed, simulating i.i.d. realizations $\bz^{(r)}$ $(r=1, \ldots, R)$, from $q^*_{\textsc{pfm}}(\bz)$ has an $\mathcal{O}(nR)$ cost, whereas computing $\Phi\{\bx_{\textsc{new}}^{\intercal}\bV\bX^{\intercal}\bz^{(r)}(1+\bx_{\textsc{new}}^{\intercal}\bV\bx_{\textsc{new}})^{-1/2}\}$ for $r=1,\dots,R$ has $\mathcal{O}(pn \cdot \min\{p,n\}+nR)$ cost because, given $\bV\bX^{\intercal}$, the computation of $\bx_{\textsc{new}}^{\intercal}\bV\bX^{\intercal}\bz^{(r)}$  for $r=1,\dots,R$ requires $\mathcal{O}(pn+nR)$ operations, while the computation of $\bx_{\textsc{new}}^{\intercal}\bV\bx_{\textsc{new}}$ can be done in $\mathcal{O}(pn \cdot \min\{p,n\})$ operations using either its original definition, when $p\leq n$, or  Woodbury's identity on $\bV$, when $p> n$, leading to
$\bx_{\textsc{new}}^{\intercal}\bV\bx_{\textsc{new}}
=
\nu_p^2\{\|\bx_{\textsc{new}}\|^2-
\nu_p^2(\bX\bx_{\textsc{new}})^{\intercal}
(\bI_n+\nu_p^2\bX\bX^\intercal)^{-1}(\bX\bx_{\textsc{new}})\}
$.

Similar derivations can be considered to prove that also Algorithm 1 for classical mean-field variational Bayes has the same $\mathcal{O}(pn \cdot\min\{p,n\})$ pre-computation cost and $\mathcal{O}(n\cdot \min\{p,n\})$ per-iteration cost of the proposed partially-factorized strategy.

\section{Proofs of Propositions, Lemmas, Theorems and Corollaries}\label{sec_proof}
We first prove some lemmas that are useful for the proofs of Theorems \ref{teo_1},  \ref{teo_3} and  \ref{teo_4}. 
We use the notation $\bM=o(p^d)$ for a matrix $\bM$ to indicate that all the entries of $\bM$ are $o(p^d)$ as $p\to\infty$.
Almost sure convergence and convergence in probability are denoted, respectively, by \smash{$\stackrel{a.s.}\to$} and \smash{$\stackrel{p}\to$}.
Moreover, convergence of matrices, both almost sure and in probability, is meant element-wise.

\vspace{10pt}
\renewcommand{\thelemma}{S1}
\begin{lemma}\label{lemma:SLLN}
	Let $(w_j)_{j\geq 1}$ be a sequence of independent random variables with mean $0$  and\\  ${\normalfont \mbox{sup}}_{j\geq 1}\{{\normalfont \mbox{var}}(w_j) \}< \infty$. Then \smash{$p^{-1/2-\delta}\sum_{j=1}^pw_j \stackrel{a.s.}\to 0$} as $p\to\infty$ for every $\delta>0$.
\end{lemma}

\vspace{5pt}

Lemma \ref{lemma:SLLN} is a  variant of the strong law of large numbers, and it follows from the Khintchine--Kolmogorov convergence theorem and Kronecker's lemma.
Being a classical result, we omit the proof for brevity.

\vspace{10pt}
\renewcommand{\thelemma}{S2}
\begin{lemma}\label{lemma:covariance}
	Under Assumptions \ref{ass_1}--\ref{ass_2} we have  that $(\sigma_x^2 p)^{-1}\bX\bX^\intercal \stackrel{a.s.}=\bI_n+o(p^{-1/2+\delta})$ for every $\delta>0$, and \smash{\mbox{$(1+\sigma_x^2p\nu_p^2)^{-1}(\bI_n+\nu_p^2\bX\bX^\intercal)\stackrel{a.s.}\to \bI_n$}} as $p\to\infty$.
	Under Assumptions \ref{ass_2}--\ref{ass_3}, we have   \smash{$(\sigma_x^2 p)^{-1}\bX\bX^\intercal \stackrel{p}\to \bI_n$} and, consequently,  \smash{$(1+\sigma_x^2p\nu_p^2)^{-1}(\bI_n+\nu_p^2\bX\bX^\intercal)\stackrel{p}\to \bI_n$} as $p\to\infty$.
\end{lemma}

\vspace{10pt}

\begin{proof}
	Considering the first part of the statement, by Assumption \ref{ass_1}, $(x_{ij}^2)_{j\geq 1}$ are independent random variables with mean $\sigma_x^2$ and variance bounded over $j$, for every  $i=1, \ldots, n$. Thus
	\begin{eqnarray*}
		\textstyle p^{1/2-\delta}\{(\sigma_x^2 p)^{-1}\bX\bX^\intercal -\bI_n\}_{ii} = p^{-1/2-\delta}\sum\nolimits_{j=1}^p(\sigma_x^{-2}x_{ij}^2 -1) \stackrel{a.s.}\to 0,
	\end{eqnarray*}
	as $p\to\infty$ by Lemma \ref{lemma:SLLN}. Similarly, when $i\neq i'$, $(x_{ij}x_{i'j})_{j\geq 1}$ are independent  variables with  mean $0$ and variance  $\sigma_x^4<\infty$. Hence
	\begin{eqnarray*}
		\textstyle p^{1/2-\delta}\{(\sigma_x^2 p)^{-1}\bX\bX^\intercal -\bI_n\}_{ii'}=\sigma_x^{-2} p^{-1/2-\delta}\sum\nolimits_{j=1}^p x_{ij}x_{i'j}\stackrel{a.s.}\to 0,
	\end{eqnarray*}
	as $p\to\infty$ by Lemma \ref{lemma:SLLN}. It follows that $(\sigma_x^2 p)^{-1}\bX\bX^\intercal \stackrel{a.s.}=\bI_n+o(p^{-1/2+\delta})$ as $p\to\infty$. Finally, since by Assumption \ref{ass_2}, $p\nu_p^2$ converges to a positive constant or  to infinity, in both cases we have 
	\begin{equation}\label{eq:manipulations}
	\begin{split}
	&(1+\sigma_x^2p\nu_p^2)^{-1}(\bI_n+\nu_p^2\bX\bX^\intercal)\\
	&= (1+\sigma_x^2p\nu_p^2)^{-1}\bI_n+(1+\sigma_x^2p\nu_p^2)^{-1} (\sigma_x^2p\nu_p^2)(\sigma_x^2p)^{-1}\bX\bX^\intercal \stackrel{a.s.}\to \bI_n, \quad \mbox{as } p\to\infty.
	\end{split}
	\end{equation}	
	
	Consider now the second part of the statement.
	Since \smash{$\{(\sigma_x^2 p)^{-1}\bX\bX^\intercal\}_{ii}=p^{-1}\sum_{j=1}^p \sigma_x^{-2} x_{ij}^2$} and $\E(\sigma_x^{-2} x_{ij}^2)=1$ for any $j=1,\ldots,p$, Chebyshev's inequality implies that, for any $\epsilon > 0$,
	\begin{equation*}
	\begin{split}
	&\textstyle \mbox{pr}\{|p^{-1}\sum\nolimits_{j=1}^p \sigma_x^{-2} x_{ij}^2 - 1|>\epsilon\} \\
	&\textstyle \le (p^2 \epsilon^2)^{-1} \mbox{\normalfont var} (\sum\nolimits_{j=1}^p \sigma_x^{-2} x_{ij}^2)= (\epsilon^2 p^2 \sigma_x^4)^{-1}\{\sum\nolimits_{j=1}^p \mbox{\normalfont var}(x_{ij}^2) + \sum\nolimits_{j=1}^p \sum\nolimits_{j'\neq j} \mbox{\normalfont cov}(x_{ij}^2,x_{ij'}^2) \}\\
	&\textstyle \le (\epsilon^2 p \sigma_x^4)^{-1}[\max_{j=1,\ldots,p}\{\E(x_{ij}^4)\}] + (\epsilon^2 p^2 \sigma_x^4)^{-1} \sum\nolimits_{j=1}^p \sum\nolimits_{j'\neq j} \mbox{\normalfont cov}(x_{ij}^2,x_{ij'}^2) \to 0, 
	\end{split}
	\end{equation*}
	as $p\to\infty$, by Assumption \ref{ass_3}(b) and $\sup_{j\geq 1}\{\E(x_{ij}^4)\} <\infty$.
	Thus $\{(\sigma_x^2 p)^{-1}\bX\bX^\intercal\}_{ii}\stackrel{p}{\to} 1$ as $p\to\infty$ for any $i=1,\ldots,n$.
	
	We now consider the off-diagonal and show \smash{$\{(\sigma_x^2 p)^{-1}\bX\bX^\intercal\}_{ii'}=p^{-1}\sum_{j=1}^p \sigma_x^{-2} x_{ij} x_{i'j} \stackrel{L^2}{\to} 0$} as $p\to\infty$ for any $i'\ne i$, which implies the desired result.
	Indeed, 
	$$\textstyle \E\{(p^{-1}\sum\nolimits_{j=1}^p \sigma_x^{-2} x_{ij} x_{i'j})^2\} = \mbox{\normalfont var}(p^{-1} \sum\nolimits_{j=1}^p \sigma_x^{-2} x_{ij} x_{i'j}) + \{\E(p^{-1} \sum\nolimits_{j=1}^p \sigma_x^{-2} x_{ij} x_{i'j})\}^2.$$
	By the Cauchy--Schwarz inequality we have 
	\begin{equation*}
	\begin{split}
	&\textstyle \mbox{\normalfont var}(p^{-1} \sum\nolimits_{j=1}^p \sigma_x^{-2} x_{ij} x_{i'j}) \\
	&=\textstyle
	(p\sigma_x^2)^{-2}
	\sum\nolimits_{j=1}^p\mbox{\normalfont var}\left(x_{ij} x_{i'j}\right) + (p\sigma^2_x)^{-2}\sum\nolimits_{j=1}^p\sum\nolimits_{j'\neq j}\mbox{\normalfont cov}(x_{ij} x_{i'j},x_{ij'} x_{i'j'})\\
	&\textstyle \leq
	(p\sigma_x^2)^{-2}
	\sum\nolimits_{j=1}^p[\E(x_{ij}^4)\E(x_{i'j}^4)]^{1/2}
	+ (p\sigma^2_x)^{-2}\sum\nolimits_{j=1}^p\sum\nolimits_{j'\neq j}\mbox{\normalfont cov}(x_{ij} x_{i'j},x_{ij'} x_{i'j'})
	\to 0,
	\end{split}
	\end{equation*}
	as $p\to 0$ by Assumption \ref{ass_3}(c) and $\sup_{i=1,\dots,n;j\geq 1}\{\E(x_{ij}^4)\} <\infty$.
	Moreover 
	$$\textstyle
	\E(p^{-1} \sum\nolimits_{j=1}^p \sigma_x^{-2} x_{ij} x_{i'j}) = p^{-1} \sigma_x^{-2} \sum\nolimits_{j=1}^p \mbox{\normalfont cov}(x_{ij},x_{i'j}) \to 0,
	$$
	as $p\to\infty$ by Assumption \ref{ass_3}(a).
	The above derivations imply that $\{(\sigma_x^2 p)^{-1}\bX\bX^\intercal\}_{ii'} \stackrel{L^2}{\to} 0$, from which it follows \smash{$\{(\sigma_x^2 p)^{-1}\bX\bX^\intercal\}_{ii'} \stackrel{p}{\to} 0$} as $p\to\infty$.
	Finally, calculations analogous to \eqref{eq:manipulations} imply \smash{$(1+\sigma_x^2p\nu_p^2)^{-1}(\bI_n+\nu_p^2\bX\bX^\intercal)\stackrel{p}\to \bI_n$} as $p\to\infty$.
\end{proof}

\vspace{15pt}
\renewcommand{\thelemma}{S3}
\begin{lemma}\label{lemma:H_asymp}
	Define $\bH$ as $\bH=\bX\bV \bX^\intercal $.
	Then, under Assumptions \ref{ass_1} and \ref{ass_2},  we have that $(1+\sigma_x^2p\nu_p^2)\left(\bI_n-\bH\right)\ \smash{\stackrel{a.s.}\to} \ \bI_n$  as $p\to\infty$.
	In particular, for $p\to\infty$, we have that $\bH \ \smash{\stackrel{a.s.}\to}\ \{(\alpha\sigma_x^2)/(1+\alpha\sigma_x^2)\}\bI_n$, where $(\alpha\sigma_x^2)/(1+\alpha\sigma_x^2)=1$ if $\alpha=\infty$.
	Under Assumptions \ref{ass_2}--\ref{ass_3}, the same holds, with a.s.\ convergence replaced by convergence in probability.
\end{lemma}
\vspace{5pt}

\begin{proof} 
	Applying Woodbury identity to $(\bI_n + \nu_p^2\bX\bX^\intercal)^{-1}$ and using $\bV=(\nu_p^{-2}\bI_p+\bX^\intercal \bX)^{-1}$, we have
	\begin{equation}
	\begin{split}
	\label{eq:H1}
	(\bI_n + \nu_p^2\bX\bX^\intercal )^{-1} &= \bI_n - \nu_p^2\bX(\bI_p  + \nu_p^2\bX^\intercal \bX)^{-1}\bX^\intercal\\
	&= \bI_n - \bX(\nu_p^{-2}\bI_p  + \bX^\intercal \bX)^{-1}\bX^\intercal = \bI_n - \bH\,.
	\end{split}
	\end{equation}
	Thus, under Assumptions \ref{ass_1} and \ref{ass_2}, 
	$$\{(1+\sigma_x^2p\nu_p^2)(\bI_n-\bH)\}^{-1}=(1+\sigma_x^2p\nu_p^2)^{-1}(\bI_n + \nu_p^2\bX\bX^\intercal)\stackrel{a.s.}\to \bI_n,$$
	as $p\to\infty$ by Lemma \ref{lemma:covariance}.
	Under Assumptions \ref{ass_2} and \ref{ass_3} the same convergence holds in probability again  by Lemma \ref{lemma:covariance}.
	The convergence of $(1+\sigma_x^2p\nu_p^2)\left(\bI_n-\bH\right)$ to the identity $\bI_n$ follows by the one of $\{(1+\sigma_x^2p\nu_p^2)(\bI_n-\bH)\}^{-1}$ and by the continuity of the inverse operator over the set of non-singular $n\times n$ matrices.
	
	Finally, to obtain the convergence of $\bH$, note that by (\ref{eq:H1}) we have
	\begin{equation}
	\label{eq:H2}
	\bH = \bI_n - (1+\sigma_x^2p\nu_p^2)^{-1} (1+\sigma_x^2p\nu_p^2) (\bI_n + \nu_p^2\bX\bX^\intercal)^{-1}.
	\end{equation}
	By Lemma \ref{lemma:covariance} and the above mentioned continuity of the inverse operator, under Assumptions \ref{ass_1} and \ref{ass_2}, $(1+\sigma_x^2p\nu_p^2) (\bI_n + \nu_p^2\bX\bX^\intercal)^{-1}\stackrel{a.s.}{\to} \bI_n$ as $p\to\infty$.
	Thus, if $\alpha = \lim_{p \to \infty} p \nu_p^2 = \infty$, $\bH \stackrel{a.s.}{\to} \bI_n$ as $p\to \infty$, since the second addend in (\ref{eq:H2}) converges a.s.\ to the null matrix. On the other hand, if $\alpha = \lim_{p \to \infty} p \nu_p^2 < \infty$, the second addend converges a.s.\ to $(1+\alpha\sigma_x^2)^{-1}\bI_n$ as $p\to \infty$, and thus we have that, for $p\to \infty$,
	$$
	\bH \stackrel{a.s.}{\to} \bI_n - (1+\alpha\sigma_x^2)^{-1}\bI_n = \{\alpha\sigma_x^2/(1+\alpha\sigma_x^2)\}\bI_n,
	$$
	as desired. Under Assumptions \ref{ass_2} and \ref{ass_3}, the same argument holds, replacing a.s.\ convergence with convergence in probability.
\end{proof}

\vspace{15pt}
\renewcommand{\thelemma}{S4}
\begin{lemma}\label{lemma:KL continuity}
	Let \smash{$\bmu^{(p)}_l\to\bzero$} and \smash{$\bSigma^{(p)}_l\to\bI_n$} as $p\to\infty$ for $l=1,2$, with $\bmu^{(p)}_l\in\mathbb{R}^n$ and $\bSigma^{(p)}_l\in \smash{\mathbb{R}^{n^2}}$ $(l=1,2)$.
	Then
	$$
	\textsc{kl}\{\textsc{tn}(\bmu^{(p)}_1,\bSigma^{(p)}_1, \mathbb{A})\mid \mid \textsc{tn}(\bmu^{(p)}_2,\bSigma^{(p)}_2, \mathbb{A})\}\to 0,
	$$
	as $p\to\infty$, with $\mathbb{A}$ denoting an orthant of $\mathbb{R}^n$ that, in the subsequent derivations, is defined as  $\mathbb{A}=\{\bz \in \mathbb{R}^n: (2y_i-1)z_i>0, \mbox{ for } i=1, \ldots, n \}$.
\end{lemma}

\vspace{8pt}

\begin{proof}
	By definition, $\textsc{kl}\{\textsc{tn}(\bmu^{(p)}_1,\bSigma^{(p)}_1, \mathbb{A})\mid \mid \textsc{tn}(\bmu^{(p)}_2,\bSigma^{(p)}_2, \mathbb{A})\}$ is equal to 
	\begin{eqnarray*}
		\textstyle
		\log\{(\psi^{(p)}_1)^{-1}\psi^{(p)}_2\}+
		0.5\log\{\det(\bSigma^{(p)}_1)^{-1}\det(\bSigma^{(p)}_2)\}
		+
		(\psi^{(p)}_1)^{-1}\det(2\pi\bSigma^{(p)}_1)^{-1/2}
		\int_{\mathbb{A}}
		f_p(\bu)
		d\bu
		\,,
	\end{eqnarray*}
	where $\psi^{(p)}_l=\mbox{pr}(\bu^{(p)}_l\in \mathbb{A})$ with $\bu^{(p)}_l\sim \mbox{N}_n(\bmu^{(p)}_l,\bSigma^{(p)}_l)$, for $l=1,2$, and
	\begin{eqnarray*}
		\begin{split}
			f_p(\bu)&=g_p(\bu)
			\exp\{-0.5(\bu -\bmu_1^{(p)})^\intercal (\bSigma^{(p)}_1)^{-1} (\bu -\bmu_1^{(p)})\}, \quad \mbox{with}\\
			g_p(\bu)&=-0.5 \{(\bu -\bmu_1^{(p)})^\intercal (\bSigma^{(p)}_1)^{-1} (\bu -\bmu_1^{(p)})-(\bu -\bmu_2^{(p)})^\intercal (\bSigma^{(p)}_2)^{-1}(\bu -\bmu_2^{(p)})\}.
		\end{split}
	\end{eqnarray*}
	Since $\bmu^{(p)}_l\to\bzero$ and $\bSigma^{(p)}_l\to\bI_n$ as $p\to\infty$, we have that $\textsc{N}_n(\bmu_l^{(p)},\bSigma_l^{(p)})\to\textsc{N}_n( 0,\bI_n)$ in distribution, and \smash{$\psi^{(p)}_l\to 2^{-n}$} by the Portmanteau theorem, thereby implying \smash{$\log\{(\psi^{(p)}_1)^{-1}\psi^{(p)}_2\}\to 0$}. 
	
	In addition, by the continuity of $\det(\cdot)$, it follows that $\det(\bSigma^{(p)}_l)\to \det(\bI_n)=1$ as $p\to\infty$, and, therefore, \smash{$\log\{\det(\bSigma^{(p)}_1)^{-1}\det(\bSigma^{(p)}_2)\}\to 0 \ \mbox{as } p\to\infty.$}
	Moreover, \smash{$\bSigma^{(p)}_l\to\bI_n$} implies that all the eigenvalues of \smash{$\bSigma^{(p)}_l$} converge to 1 as $p\to\infty$ for $l=1,2$, and thus are eventually bounded away from $0$ and $\infty$.  Therefore, there exist positive, finite constants $m$, $M$ and $k$ such that $$m\| \bu -\bmu_l^{(p)} \|^2\leq (\bu -\bmu_l^{(p)})^\intercal (\bSigma^{(p)}_l)^{-1} (\bu -\bmu_l^{(p)}) \leq M\| \bu -\bmu_l^{(p)} \|^2, \ \mbox{for } l=1,2 \ \mbox{and } p\geq k.$$ Calling $b=\sup_{p\geq 1,\  l\in\{1,2\}}\|\bmu_l^{(p)}\| <\infty$, and using standard properties of norms, we have
	\begin{eqnarray*}
		m(\|\bu\|^2-2b\|\bu\|)\le
		(\bu -\bmu_l^{(p)})^\intercal (\bSigma^{(p)}_l)^{-1}(\bu -\bmu_l^{(p)})\le 2M(\|\bu\|^2+b^2), \ \mbox{for } l=1,2 \ \mbox{and } p\geq k,
	\end{eqnarray*}
	from which we obtain that, for any $p\ge k$, 
	$
	|f_p(\bu)|\leq
	2M(\|\bu\|^2{+}b^2)\exp\{-m(\|\bu\|^2/2-b\|\bu\|)\}
	$,
	where the latter is an integrable function on $\mathbb{R}^n$.
	Therefore we can apply the dominated convergence theorem and obtain $\lim_{p \to \infty}\int_{\mathbb{A}}
	f_p(\bu)
	d\bu=
	\int_{\mathbb{A}}
	\lim_{p \to \infty} f_p(\bu)
	d\bu
	=0
	$ as desired.
\end{proof}

\subsection{Proof of Theorem \ref{teo_1}}
Let $
\bar{\bbeta}^{*}
=
\arg\max_{\bbeta\in\mathbb{R}^p}\ell(\bbeta)$, where $\ell(\bbeta)=-(2\nu_p^2)^{-1}\|\bbeta\|^2+\sum_{i=1}^n\log \Phi\{(2y_i{-}1)\bx_i^\intercal\bbeta\}$ denotes the log-posterior up to an additive constant under \eqref{eq1}.
Note that $\bar{\bbeta}^{*}$ is unique because $\ell(\bbeta)$ is strictly concave \citep{haberman1974}.

\vspace{12pt}
\renewcommand{\thelemma}{S5}
\begin{lemma}\label{lemma:MAP_to_0}
	Under Assumptions \ref{ass_1}--\ref{ass_2},  with $\alpha= \infty$, we have $\nu_p^{-1}\|\bar{\bbeta}^{*}\|\stackrel{a.s.}\to 0$ as $p\to\infty$.
\end{lemma}

\vspace{6pt}
\begin{proof}
	Since $\log\Phi\{(2y_i-1)\bx_i^\intercal\bbeta\}<0$ for any $i=1, \ldots, n$, we have $\ell(\bbeta)<-(2\nu_p^2)^{-1}\|\bbeta\|^2$ and thus $\nu_p^{-2}\|\bbeta\|^2< -2\ell(\bbeta)$ for any $\bbeta\in\mathbb{R}^p$ and any $p$.
	Therefore, it follows that $\nu_p^{-2}\|\bar{\bbeta}^{*}\|^2< -2\ell(\bar{\bbeta}^{*})=-2\sup_{\bbeta\in\mathbb{R}^p}\ell(\bbeta)$.
	
	We now prove that $\sup_{\bbeta\in\mathbb{R}^p}\ell(\bbeta) \stackrel{a.s.}\to 0$ as $p\to\infty$. For any $p$, define $\tilde{\bbeta}=(\tilde{\beta}_j)_{j=1}^p\in\mathbb{R}^p$ as
	\begin{eqnarray*}
		\tilde{\beta}_j= \nu_p^{2/3}p^{-2/3}(2y_{_{\lceil nj/p\rceil}}-1) x_{\lceil nj/p\rceil,j} \quad (j=1,\dots,p),
	\end{eqnarray*}
	where $\lceil a\rceil$ denotes the smallest integer larger than or equal to $a$.
	It follows that
	\begin{eqnarray*}
		\textstyle	(p \nu^2_p)^{-1/3}\bx_i^\intercal \tilde{\bbeta}
		=
		p^{-1}(2y_i-1)
		\sum\nolimits_{j\in D_i}x_{ij}^2
		+
		p^{-1}
		\sum\nolimits_{j\notin D_i}  \zeta_{ij},
	\end{eqnarray*}
	where \smash{$D_i=\{j\in\{1,\dots,p\}: (i-1)p/n<j\leq ip/n\}$}, while \smash{$\zeta_{ij}=x_{ij}x_{\lceil nj/p\rceil,j}(2y_{_{\lceil nj/p\rceil}}-1)$}. Now, note that \smash{$(x_{ij}^2)_{j\in D_i}$} and \smash{$(\zeta_{ij})_{j\notin D_i}$} are independent variables with bounded variance, the size of each $D_i$ is asymptotic to $n^{-1}p$ as $p\to\infty$ and $\E(\zeta_{ij})=0$ for $j\notin D_i$. Thus, Lemma \ref{lemma:SLLN} implies \smash{$\lim_{p\to\infty}(p \nu^2_p)^{-1/3}\bx_i^\intercal \tilde{\bbeta}\stackrel{a.s.}=n^{-1}(2y_i-1)\sigma_x^2$}.
	Since $p \nu^2_p \rightarrow \infty$ as $p \rightarrow \infty$, and assuming $\sigma_x^2>0$ without loss of generality --- when \smash{$ \sigma_x^2=0$} it holds \smash{$\bar{\bbeta}^{*}\stackrel{a.s.}=0$} --- it follows that
	\smash{$\bx_i^\intercal \tilde{\bbeta}\stackrel{a.s.}\to+\infty$} if $y_i=1$ and \smash{$\bx_i^\intercal \tilde{\bbeta}\stackrel{a.s.}\to-\infty$} if $y_i=0$ as $p\to\infty$ and therefore \smash{$\sum_{i=1}^n\log \Phi\{(2y_i{-}1)\bx_i^\intercal \tilde{\bbeta}\}\stackrel{a.s.}\to 0$} as $p\to\infty$.
	Moreover, $\nu_p^{-2}\|\tilde{\bbeta}\|^2=\smash{(p \nu_p^2)^{-1/3}(p^{-1}\sum_{j=1}^px_{\lceil nj/p\rceil,j}^2)\stackrel{a.s.}\to 0}$ as $p\to\infty$ by Lemma \ref{lemma:SLLN} and \smash{$p \nu^2_p \rightarrow \infty$}.
	Thus \smash{$0\geq \sup_{\bbeta\in\Re^p}\ell (\bbeta) \geq \ell(\tilde{\bbeta})\stackrel{a.s.}\to 0$} as $p\to\infty$ as desired.
\end{proof}

\vspace{12pt}
\renewcommand{\thelemma}{S6}
\begin{lemma}\label{lemma:KL_predictive}
	Let $q_1$ and $q_2$ be probability distributions on $\mathbb{R}^p$.
	Then, for any $\bx_{\textsc{new}}\in\mathbb{R} ^p$, we have
	$
	\textsc{kl}(q_1\mid \mid q_2)
	\geq
	2\, \left|
	\mbox{\normalfont pr}_{q_1}-
	\mbox{\normalfont pr}_{q_2}
	\right|^2
	$,
	where 
	$\mbox{\normalfont pr}_{q_l}=
	\int \Phi(\bx^{\intercal}_{\textsc{new}} \bbeta)q_l(\bbeta)\mbox{\normalfont d}\bbeta
	$ for $l=1,2$.
\end{lemma}

\vspace{6pt}
\begin{proof}
	From Pinsker's inequality we have that
	$
	\textsc{kl}(q_1\mid \mid q_2)\geq
	2\, \textsc{tv}(q_1, q_2)^2
	$
	where $\textsc{tv}(\cdot, \cdot)$ denotes the total variation distance between probability distributions. Now, recall that 
	$
	\textsc{tv}(q_1, q_2)=\sup_{h:\mathbb{R}^p\to [0,1]}|
	\int_{\mathbb{R} ^p}h(\bbeta)q_1(\bbeta)d \bbeta- \int_{\mathbb{R} ^p}h(\bbeta)q_2(\bbeta)d\bbeta
	|.
	$
	Taking  $h(\bbeta)=\Phi(\bx^{\intercal}_{\textsc{new}} \bbeta)$ 
	we obtain the desired statement.
\end{proof}

\vspace{10pt}

\begin{proof}[Proof of Theorem \ref{teo_1}.]
	As noted in \citet{armagan_2011}, the coordinate ascent variational inference algorithm for classical mean-field variational Bayes is equivalent to an \textsc{em} algorithm for  $p(\bbeta \mid \by)$ with missing data $\bz$, which in this case is guaranteed to converge to the unique maximizer of $p(\bbeta \mid \by)$ by, e.g., Theorem 3.2 of \citet{mclachlan2007algorithm}, and the fact that $p(\bbeta \mid \by)$ is strictly concave \citep{haberman1974}. Therefore $\E_{q_{\textsc{mf}}^{*}(\bbeta)}(\bbeta)=\bar{\bbeta}^{*}$ and Lemma \ref{lemma:MAP_to_0} implies that \smash{$\nu_p^{-1}\|\E_{q_{\textsc{mf}}^{*}(\bbeta)}(\bbeta)\|\stackrel{a.s.}\to 0$} as $p\to\infty$.

	We now prove that under Assumptions \ref{ass_1}--\ref{ass_2},  with $\alpha= \infty$, $\nu_p^{-2}\|\E_{p(\bbeta \mid \by)}(\bbeta)\|^2
	\stackrel{a.s.}\to c^2n$ as $p\to\infty$. By the law of total expectation we have that
	$\E_{p(\bbeta|\by)}(\bbeta)=\bV \bX^\intercal \E_{p(\bz|\by)}(\bz)$, which implies
	$$\|\E_{p(\bbeta|\by)}(\bbeta)\|^2= \E_{p(\bz|\by)}(\bz)^\intercal \bX\bV^\intercal\bV \bX^\intercal \E_{p(\bz|\by)}(\bz).$$
	Applying the Woodbury's identity to $\bV$ we have
	$
	\bV \bX^\intercal =\nu_p^2\bX^\intercal (\bI_n+\nu_p^2\bX\bX^\intercal )^{-1}
	$.
	Therefore, we can write \smash{$\bX\bV^\intercal\bV \bX^\intercal=(1{+}\sigma_x^2p\nu_p^2)^{-2} \sigma_x^2 p\nu_p^{4}\bS^\intercal\{(\sigma_x^2p)^{-1}\bX\bX^\intercal\} \bS$} with $\bS=\smash{(1+\sigma_x^2p\nu_p^2)}(\bI_n+\nu_p^2\bX\bX^\intercal )^{-1}$. Since $\alpha= \infty$, and \smash{$\bS^\intercal\{(\sigma_x^2p)^{-1}\bX\bX^\intercal\} \bS\stackrel{a.s.}{\to} \bI_n$} as $p\to\infty$ by Lemma \ref{lemma:covariance}, it follows that 
	\begin{eqnarray*}
		\lim_{p\to\infty}\nu_p^{-2}\|\E_{p(\bbeta|\by)}(\bbeta)\|^2&\stackrel{a.s.}=&\lim_{p\to\infty} \{(\sigma_x^2p\nu_p^2)/(1+\sigma_x^2p\nu_p^2)\} \|\E_{p(\bz|\by)}[(1+\sigma_x^2p\nu_p^2)^{-1/2}\bz]\|^2,\\
		&&=  \lim_{p\to\infty}  \|\E_{p(\bz|\by)}\{(1+\sigma_x^2p\nu_p^2)^{-1/2}\bz\}\|^2.
	\end{eqnarray*}
	Since $\left(\bz\mid\by\right)\sim \textsc{tn}\{0,(\bI_n{+}\nu_p^2\bX\bX^\intercal), \mathbb{A}\}$, it follows that $\{(1+\sigma_x^2p\nu_p^2)^{-1/2}\bz\mid \by \}\sim \textsc{tn}\{0,(1+\sigma_x^2p\nu_p^2)^{-1}(\bI_n+\nu_p^2\bX\bX^\intercal), \mathbb{A}\}$.
	Therefore, 
	\begin{eqnarray*}
		\textstyle
		\E_{p(z_i|\by)}\{(1+\sigma_x^2p\nu_p^2)^{-1/2}z_i\}=(\tilde{\psi}^{(p)})^{-1}\int_{\mathbb{A}}\tilde{u}_i \phi_n\{\tilde{\bu};(1+\sigma_x^2p\nu_p^2)^{-1}(\bI_n+\nu_p^2\bX\bX^\intercal)\}
		d\tilde{\bu},
	\end{eqnarray*}
	where $\tilde{\psi}^{(p)}=\mbox{pr}(\bu^{(p)}\in \mathbb{A})$, with $\bu^{(p)}\sim \mbox{N}_n\{0,(1+\sigma_x^2p\nu_p^2)^{-1}(\bI_n+\nu_p^2\bX\bX^\intercal)\}$.
	Hence, Lemma \ref{lemma:covariance} together with a domination argument similar to the one used in the proof of Lemma \ref{lemma:KL continuity} imply that, as $p\to \infty$,
	\begin{eqnarray*}
		\textstyle
		\E_{p(z_i|\by)}\{(1+\sigma_x^2p\nu_p^2)^{-1/2}z_i\}\stackrel{a.s.}\to 2^n\int_{\mathbb{A}} \tilde{u}_i \phi_n(\tilde{\bu}; \bI_n) d\tilde{\bu} = c (2y_i-1), \ \mbox{with } c=2\int_{0}^{\infty}u \phi(u)du.
	\end{eqnarray*}
	Hence, 
	\smash{$\lim_{p\to\infty}\nu_p^{-2}\|\E_{p(\bbeta|\by)}(\bbeta)\|^2
		\stackrel{a.s.}=\sum_{i=1}^n c^2=c^2n$}. Taking the square root of this expression, yields the desired result stated in point {\em (b)} of Theorem~\ref{teo_1}.
	
	Let us now focus on proving that  \smash{$\lim\inf_{p \to \infty}\textsc{kl}\{q_{\textsc{mf}}^{*}(\bbeta) \mid \mid p(\bbeta \mid \by)\}>0$}, almost surely. As a direct consequence of 
	Lemma \ref{lemma:KL_predictive} it follows that
	\smash{$
		\textsc{kl}\{q_{\textsc{mf}}^{*}(\bbeta) \mid \mid p(\bbeta \mid \by)\}
		\geq
		2\, \left|
		\mbox{\normalfont pr}_{\textsc{mf}}-\mbox{\normalfont pr}_{\textsc{sun}}
		\right|^2$},
	with 
	$\mbox{\normalfont pr}_{\textsc{sun}}=
	\int \Phi(\bx^{\intercal}_{\textsc{new}} \bbeta){p}(\bbeta|\by)\mbox{\normalfont d}\bbeta
	$
	and 
	\smash{$\mbox{\normalfont pr}_{\textsc{mf}}=\int \Phi(\bx^{\intercal}_{\textsc{new}} \bbeta){q}^*_{\textsc{mf}}(\bbeta)\mbox{\normalfont d}\bbeta$} for any $\bx_{\textsc{new}} \in \mathbb{R}^p$.
	To continue the proof, we  consider the input vector \smash{$\bx_{\textsc{new}}=(\sigma_x^2p\nu_p^2)^{-1/2}\bX^\intercal \bH^{-1}\bdelta$}, with $\bdelta=(2y_1-1,0,\dots,0)^{\intercal}$, and show that \smash{$\lim_{p\to\infty}\left|
		\mbox{\normalfont pr}_{\textsc{mf}}-\mbox{\normalfont pr}_{\textsc{sun}}
		\right|>0$}.  Here we can assume without loss of generality that  $\bH$ is invertible because \smash{$\bH \stackrel{a.s.}\to  \bI_n$} as $p\to\infty$ by $\alpha= \infty$ and Lemma \ref{lemma:H_asymp}, and the set of $n\times n$ non-singular matrices is open. This result implies that $\bH$ is  eventually invertible as $p\to\infty$, 
	almost surely.
	By definition of $\bx_{\textsc{new}}$ we have
	\begin{eqnarray*}
		\nu_p^2 \|\bx_{\textsc{new}}\|^2= \nu_p^2 
		\bx_{\textsc{new}}^\intercal \bx_{\textsc{new}}=
		\bdelta^\intercal \bH^{-1}\{(\sigma_x^2p)^{-1}\bX\bX^\intercal\} \bH^{-1}\bdelta\stackrel{a.s.}\to 1, \ \mbox{as } p\to\infty,
	\end{eqnarray*}
	because $\bH^{-1}\stackrel{a.s.}\to \bI_n$ and $(\sigma_x^2p)^{-1}\bX\bX^\intercal\stackrel{a.s.}\to\bI_n$ as $p\to\infty$ by Lemmas \ref{lemma:H_asymp} and \ref{lemma:covariance}, respectively, $\alpha= \infty$ and $\|\bdelta\|=1$. By \eqref{eq7} we have that
	\smash{$
		\mbox{\normalfont pr}_{\textsc{mf}}=\Phi\{\bx^{\intercal}_{\textsc{new}} \bar{\bbeta}^{*}(1{+}\bx^{\intercal}_{\textsc{new}}\bV  \bx_{\textsc{new}})^{-1/2}\}$}. Therefore, leveraging the Cauchy--Schwarz inequality and the fact that $\bx^{\intercal}_{\textsc{new}}\bV  \bx_{\textsc{new}}\geq 0$, it easily follows that
	\smash{$
		|\bx^{\intercal}_{\textsc{new}} \bar{\bbeta}^{*}(1+\bx^{\intercal}_{\textsc{new}}\bV  \bx_{\textsc{new}})^{-1/2}|
		\leq
		\|\bx_{\textsc{new}}\| \|\bar{\bbeta}^{*}\|
		\stackrel{a.s.}{\rightarrow}  0
		$} as $p\to\infty$, where the latter convergence is a consequence of
	$\nu_p\|\bx_{\textsc{new}}\|\stackrel{a.s.}{\rightarrow}  1$ and $\nu_p^{-1}\|\bar{\bbeta}^{*}\|\stackrel{a.s.}{\rightarrow}  0$.
	Therefore $\mbox{\normalfont pr}_{\textsc{mf}} \rightarrow  0.5$, almost surely, 
	as $p \rightarrow \infty$.
	
	Consider now $\mbox{\normalfont pr}_{\textsc{sun}}$. With derivations analogous to those of equation \eqref{eq11}, we can define $\mbox{\normalfont pr}_{\textsc{sun}}$ as
	$\mbox{\normalfont pr}_{\textsc{sun}}=E_{p(\bz|\by)}[\Phi\{\bx_{\textsc{new}}^{\intercal}\bV\bX^{\intercal}\bz(1+\bx_{\textsc{new}}^{\intercal}\bV\bx_{\textsc{new}})^{-1/2}\}]$.
	By definition of $\bx_{\textsc{new}}$, we have that
	$$\bx_{\textsc{new}}^\intercal \bV\bx_{\textsc{new}}=(\sigma_x^2 p\nu_p^2)^{-1}\bdelta^\intercal \bH^{-1}XVX^{\intercal} \bH^{-1}\bdelta=	(\sigma_x^2 p\nu_p^2)^{-1}\bdelta^\intercal \bH^{-1}\bdelta,$$ 
	because, as stated in Lemma \ref{lemma:H_asymp}, $H=X^{\intercal}VX$. Thus, since $\bH^{-1}\stackrel{a.s.}\to \smash{\bI_n}$, as $p\to\infty$ by Lemma \ref{lemma:H_asymp} and $\alpha= \infty$, we have \smash{$\bx_{\textsc{new}}^\intercal \bV\bx_{\textsc{new}}\stackrel{a.s.}\to 0$}.
	Moreover, $\bx_{\textsc{new}}^{\intercal}\bV\bX^{\intercal}\bz=\bdelta^\intercal \{(\sigma_x^2 p\nu_p^2)^{-1/2}\bz\}$, where \smash{$(\sigma_x^2 p\nu_p^2)^{-1/2}\bz\to \textsc{tn}(\bzero,\bI_n, \mathbb{A})$} in distribution as $p\to\infty$, almost surely, as a direct consequence of Lemma  \ref{lemma:covariance}.
	Combining these results with Slutsky's lemma and the fact that $\Phi(\cdot)$ is bounded and continuous, it follows that
	$E_{p(\bz|\by)}[\Phi\{\bx_{\textsc{new}}^{\intercal}\bV\bX^{\intercal}\bz(1+\bx_{\textsc{new}}^{\intercal}\bV\bx_{\textsc{new}})^{-1/2}\}]\to 
	\smash{E_{p(\tilde{\bz})}\{\Phi(\bdelta^{\intercal}\tilde{\bz})\}}$, almost surely,  with $\tilde{\bz}\sim\textsc{tn}(\bzero,\bI_n, \mathbb{A})$.
	Thus
	$$\textstyle \mbox{\normalfont pr}_{\textsc{sun}}\stackrel{a.s.}\to E_{p(\tilde{z}_1)}[\Phi\{(2y_1{-}1)\tilde{z}_1\}]= {\int_{0}^\infty} \Phi(z) 2\phi(z)dz,$$
	as $p\to\infty$, where ${\int_{0}^\infty} \Phi(z) 2\phi(z)dz>0.5$. Hence, it  follows that 
	$
	\liminf_{p\to\infty}
	\textsc{kl}\{q_{\textsc{mf}}^{*}(\bbeta) \mid \mid p(\bbeta \mid \by)\}
	\geq
	\smash{2\,\lim_{p\to\infty}
		\left|
		\mbox{\normalfont pr}_{\textsc{mf}}-\mbox{\normalfont pr}_{\textsc{sun}}
		\right|^2}
	>0
	$
	almost surely as $p\to\infty$.
	
	To conclude the proof of the Theorem~\ref{teo_1}, notice that, consistent with the above derivations, by setting
	$\bx_{\textsc{new}}=(\sigma_x^2p\nu_p^2)^{-1/2}\bX^\intercal \bH^{-1}\bdelta$, with $\bdelta=(2y_1-1,0,\dots,0)^{\intercal}$,
	for every $p$
	leads to 
	$
	\liminf_{p\to\infty }| \mbox{\normalfont pr}_{\textsc{mf}}-\mbox{\normalfont pr}_{\textsc{sun}}|
	> 0$, almost surely, thereby implying that such a result also holds for the worst case scenario  in Theorem~\ref{teo_1}{\em (c)}. This proves the result on the predictive probabilities. 
\end{proof}

\vspace{5pt}

\subsection{Proof of Theorem \ref{teo_2}, Corollary \ref{cor_SUN} and Proposition \ref{prop1}}
\vspace{5pt}
\begin{proof}[Proof of Theorem \ref{teo_2}.]
	Leveraging the chain rule of the $\textsc{kl}$ divergence we have that 
	$$
	\textsc{kl}\{q_{\textsc{pfm}}(\bbeta,\bz) \mid \mid p(\bbeta,\bz \mid \by)\}=\textsc{kl}\{q_{\textsc{pfm}}(\bz) \mid \mid p(\bz \mid \by)\}+E_{q_{\textsc{pfm}}(\bz)}[\textsc{kl}\{q_{\textsc{pfm}}(\bbeta \mid \bz) \mid \mid p(\bbeta \mid \bz)\}]\,,
	$$
	where $q_{\textsc{pfm}}(\bbeta \mid \bz)$ appears only in the second summand, which is  always non-negative and coincides with zero,  for any $q_{\textsc{pfm}}(\bz)$, if and only if $q^*_{\textsc{pfm}}(\bbeta \mid\bz)=p(\bbeta \mid \bz)=\phi_p(\bbeta-\bV\bX^\intercal\bz; \bV)$.
	
	The expression for $q^*_{\textsc{pfm}}(\bz) =\prod_{i=1}^nq_{\textsc{pfm}}^*(z_i)$ is instead a direct consequence of the closure under conditioning property of the multivariate truncated Gaussian \citep{ horrace_2005,holmes_2006}. In particular, adapting the results in \citet{holmes_2006}, it easily follows that 
	\begin{eqnarray*}
		p(z_i \mid \bz_{-i},\by) \propto \phi\{z_i-(1-\bx_i^{\intercal}\bV \bx_i)^{-1} \bx_i^{\intercal}\bV\bX^{\intercal}_{-i}\bz_{-i}; (1-\bx_i^{\intercal}\bV \bx_i)^{-1}\}{1}\{(2y_i-1)z_i>0\},
	\end{eqnarray*}
	for every $i=1, \ldots, n$, where $\bX_{-i}$ is the design matrix without row  $i$. To obtain the expression for $q_{\textsc{pfm}}^*(z_i)$ $(i=1, \ldots, n)$, note that, recalling e.g., \citet{blei_2017}, the optimal solution for $q_{\textsc{pfm}}(\bz)$ which minimizes $\textsc{kl}\{q_{\textsc{pfm}}(\bz) \mid \mid p(\bz \mid \by)\}$ within the family of distributions that factorize over $z_1, \ldots, z_n$ can be expressed as $\prod_{i=1}^nq_{\textsc{pfm}}^*(z_i)$ with $q_{\textsc{pfm}}^*(z_i) \propto \exp[ E_{q_{\textsc{pfm}}^*(\bz_{-i})} \{\log p(z_i \mid \bz_{-i},\by)\}]$ for every $i=1, \ldots, n$. Combining such a result with the above expression for $p(z_i \mid \bz_{-i},\by) $ we have that $\exp[ E_{q_{\textsc{pfm}}^*(\bz_{-i})} \{\log p(z_i \mid \bz_{-i},\by)\}]$ is proportional to
	\begin{eqnarray*}
		\exp \left\{-\frac{z_i^2-2z_i(1-\bx_i^{\intercal}\bV \bx_i)^{-1} \bx_i^{\intercal}\bV\bX^{\intercal}_{-i}E_{q_{\textsc{pfm}}^*(\bz_{-i})}(\bz_{-i})}{2(1-\bx_i^{\intercal}\bV \bx_i)^{-1}}\right\}{1}\{(2y_i-1)z_i>0\}, \ \ (i=1, \ldots, n).
	\end{eqnarray*}
	The above quantity coincides with the kernel of a Gaussian density having variance $\sigma^{*2}_i=(1-\bx_i^{\intercal}\bV \bx_i)^{-1}$, expectation $\mu^{*}_i=\sigma^{*2}_i \bx_i^{\intercal}\bV\bX^{\intercal}_{-i}E_{q_{\textsc{pfm}}^*(\bz_{-i})}(\bz_{-i})$ and truncation below zero if $y_i=1$ or above zero if $y_i=0$. Hence, each $q_{\textsc{pfm}}^*(z_i)$ is the density of a truncated normal with parameters specified as in Theorem \ref{teo_2}. The proof is concluded after noticing that the expression for $\bar{z}^*_i=E_{q_{\textsc{pfm}}^*(z_{i})}(z_i)$, $i=1, \ldots, n$, in  Theorem \ref{teo_2} follows directly from the mean of truncated normals. \end{proof}

\vspace{15pt}
\begin{proof}[Proof of Corollary \ref{cor_SUN}.]
	From \eqref{eq8}, we have that $q^*_{\textsc{pfm}}(\bbeta)$ coincides with the density of a random variable which has the same distribution of $\tilde{\bu}^{(0)}+\bV\bX^{\intercal} \tilde{\bu}^{(1)}$, where $\tilde{\bu}^{(0)} \sim \mbox{N}_p({0}, \bV)$ and $\tilde{\bu}^{(1)}$ is from an $n$-variate Gaussian with mean vector $\bmu^*$, diagonal covariance matrix $\bsigma^{*2}$ and generic $i$th component truncated either below or above zero depending of the sign of $(2y_i-1)$, for $i=1, \ldots, n$. Since $\tilde{\bu}^{(1)}$  has independent components, by standard properties of univariate truncated normal variables we obtain
	\begin{eqnarray*}
		\tilde{\bu}^{(0)}+\bV\bX^{\intercal} \tilde{\bu}^{(1)} \stackrel{d}{=}\bu^{(0)}+\bV\bX^{\intercal}\bar{\bY}\bsigma^* \bu^{(1)}, \quad \mbox{with } \bar{\bY}=\mbox{diag}(2y_1-1, \ldots, 2 y_n-1),
	\end{eqnarray*}
	where $\bu^{(0)} \sim \mbox{N}_p(\bV\bX^{\intercal} \bmu^*,\bV)$ and $ \bu^{(1)}$ is an $n$-variate Gaussian  with mean  $0$, covariance matrix $\bI_n$, and truncation below \smash{$-\bar{\bY} \bsigma^{*^{-1}}\bmu^*$}. Calling  \smash{$\bxi=\bV\bX^{\intercal}\bmu^*$, $\bOmega=\bomega \bar{\bOmega}\bomega =\bV+\bV\bX^{\intercal}\bsigma^{*2}\bX \bV$},  \smash{$\bDelta=\bomega^{-1}\bV\bX^{\intercal}\bar{\bY}\bsigma^*$, $\bgamma=\bar{\bY}\bsigma^{*^{-1}}\bmu^*$} and $\bGamma=\bI_n$, as in Corollary \ref{cor_SUN}, we have that
	\begin{eqnarray*}
		\bu^{(0)}+\bV\bX^{\intercal}\bar{\bY} \bsigma^* \bu^{(1)}\stackrel{d}{=}\bxi+\bomega(\bar{\bu}^{(0)}+\bDelta\bGamma^{-1} \bar{\bu}^{(1)}),
	\end{eqnarray*}
	with $\bar{\bu}^{(0)} \sim \mbox{N}_p(0,\bar{\bOmega}-\bDelta\bGamma^{-1}\bDelta^{\intercal})$, and $ \bar{\bu}^{(1)}$ distributed as a $n$-variate Gaussian random variable with mean vector $0$, covariance matrix $\bGamma$, and truncation below $-\bgamma$. Recalling \citet{arellano_2006} and \citet{azzalini_2013} such a stochastic representation coincides with the one of the  unified skew-normal random variable $\textsc{sun}_{p,n}(\bxi, \bOmega, \bDelta, \bgamma, \bGamma)$.
\end{proof}

\vspace{12pt}
\begin{proof}[Proof of Proposition \ref{prop1}.] To prove Proposition \ref{prop1}, first notice that by the results in equation \eqref{eq8} and in Theorem \ref{teo_2}, $\bz=(z_1, \ldots, z_n)^{\intercal}$ is a random vector whose entries have independent truncated normal approximating densities. Hence, \smash{$E_{q_{\textsc{pfm}}^*(z_{i})}(z_{i})=\bar{z}^*_i$} and $\smash{\mbox{var}_{q_{\textsc{pfm}}^*(z_{i})}(z_{i})}=\sigma^{*2}_i\{1-(2y_i{-}1)\eta^*_i \mu^*_i/\sigma^*_i-\eta^{*2}_i \}$ with $\eta_i=\phi(\mu^*_i/\sigma^{*}_i)\Phi\{(2y_i-1)\mu^*_i/\sigma^{*}_i\}^{-1}$ for $i=1, \ldots,n$. Using the parameters defined in Theorem \ref{teo_2}, $\smash{\mbox{var}_{q_{\textsc{pfm}}^*(z_{i})}(z_{i})}$ can be rewritten as $\smash{\mbox{var}_{q_{\textsc{pfm}}^*(z_{i})}(z_{i})}=\sigma^{*2}_i-(\bar{z}^*_i-\mu^*_i)\bar{z}_i^*$. Therefore, $E_{q_{\textsc{pfm}}^*(\bz)}(\bz)=\bar{\bz}^*$, whereas $\mbox{var}_{q_{\textsc{pfm}}^*(\bz)}(\bz)=\mbox{\normalfont diag}\{\sigma_1^{*2}{-}(\bar{z}_1^*-\mu_1^*)\bar{z}_1^*, \ldots, \sigma_n^{*2}-(\bar{z}_n^*-\mu_n^*)\bar{z}_n^*\}$, where $\bar{z}_i^*$, $\mu_i^*$ and $\sigma_i^{*}$ $(i=1, \ldots, n)$ are  defined in Theorem \ref{teo_2} and in Corollary \ref{cor_SUN}. Combining these results with equation  \eqref{eq8}, and using the law of  iterated expectations we have
	\begin{eqnarray*}
		E_{q_{\textsc{pfm}}^*(\bbeta)}(\bbeta)&=&E_{q_{\textsc{pfm}}^*(\bz)}\{E_{p(\bbeta \mid \bz)}(\bbeta)\}=E_{q_{\textsc{pfm}}^*(\bz)}(\bV \bX^{\intercal}\bz)=\bV \bX^{\intercal}E_{q_{\textsc{pfm}}^*(\bz)}(\bz)=\bV \bX^{\intercal}\bar{\bz}^*,\\
		\mbox{var}_{q_{\textsc{pfm}}^*(\bbeta)}(\bbeta)&=&E_{q_{\textsc{pfm}}^*(\bz)}\{\mbox{var}_{p(\bbeta \mid \bz)}(\bbeta)\}+\mbox{var}_{q_{\textsc{pfm}}^*(\bz)}\{E_{p(\bbeta \mid \bz)}(\bbeta)\}=\bV+\bV\bX^{\intercal}\mbox{var}_{q_{\textsc{pfm}}^*(\bz)}(\bz)\bX \bV\\
		&=&\bV+\bV\bX^{\intercal}\mbox{\normalfont diag}\{\sigma_1^{*2}-(\bar{z}_1^*-\mu_1^*)\bar{z}_1^*, \ldots, \sigma_n^{*2}-(\bar{z}_n^*-\mu_n^*)\bar{z}_n^*\}\bX \bV,
	\end{eqnarray*}
	thus proving equation \eqref{eq10}.
	
	To prove equation \eqref{eq11} it suffices to notice that $\mbox{\normalfont pr}_{\textsc{pfm}}(y_{\textsc{new}}=1\mid \by)=E_{q_{\textsc{pfm}}^*(\bbeta)}\{ \Phi(\bx^{\intercal}_{\textsc{new}} \bbeta)\}$. Hence, by applying again the law of  iterated expectations we have
	\begin{eqnarray*}
		E_{q_{\textsc{pfm}}^*(\bbeta)}\{ \Phi(\bx^{\intercal}_{\textsc{new}} \bbeta)\}&=&E_{q_{\textsc{pfm}}^*(\bz)}[E_{p(\bbeta \mid \bz)}\{ \Phi(\bx^{\intercal}_{\textsc{new}} \bbeta)\}]\\
		&=&E_{q^*_{\textsc{pfm}}(\bz)}[\Phi\{\bx_{\textsc{new}}^{\intercal}\bV\bX^{\intercal}\bz(1+\bx_{\textsc{new}}^{\intercal}\bV\bx_{\textsc{new}})^{-1/2}\}].
	\end{eqnarray*}
	The last equality follows from the fact that $p(\bbeta \mid \bz)$ is Gaussian and thus $E_{p(\bbeta \mid \bz)}\{ \Phi(\bx^{\intercal}_{\textsc{new}} \bbeta)\}$ can be derived in closed form; see e.g., Lemma 7.1 in \citet{azzalini_2013}.
\end{proof}

\vspace{5pt}
\subsection{Proof of Theorem \ref{teo_3} and Corollary \ref{cor_predictive}}
\vspace{3pt}
\begin{proof}[Proof of Theorem \ref{teo_3}.]
	As a consequence of the discussion after the statement of Theorem \ref{teo_2}, the density $q^*_{\textsc{pfm}}(\bz)$ minimizes the \textsc{kl} divergence to $p(\bz \mid \by)$ within the family of distributions that factorize over $z_1,\dots,z_n$.
	As a consequence
	\begin{align}\label{eq:eq1}
	\textsc{kl}\{q_{\textsc{pfm}}^*(\bz)\mid \mid p(\bz \mid \by)\}\leq \textsc{kl}[\textsc{tn}\{\bzero,(1+\sigma_x^2p\nu_p^2)\bI_n, \mathbb{A}\}\mid \mid p(\bz\mid \by)].
	\end{align}
	Since the \textsc{kl}  is invariant with respect to bijective transformations of $z$, and $p(\bz \mid \by)=\textsc{tn}(\bzero,
	\bI_n+\nu_p^2\bX\bX^\intercal, \mathbb{A})$, then
	rescaling each $z_i$ by $(1+\sigma_x^2p\nu_p^2)^{-1/2}$ we have 
	\begin{equation}
	\begin{split}
	&\textsc{kl}[\textsc{tn}\{\bzero,(1+\sigma_x^2p\nu_p^2)\bI_n, \mathbb{A}\}\mid \mid p(\bz \mid \by)]\\
	&=\textsc{kl}[\textsc{tn}(\bzero,\bI_n, \mathbb{A})\mid \mid \textsc{tn}\{\bzero,(1+\sigma_x^2p\nu_p^2)^{-1}(\bI_n+\nu_p^2\bX\bX^\intercal), \mathbb{A}\}]\,.
	\end{split}
	\label{eq:eq2}
	\end{equation}
	Lemma \ref{lemma:covariance} shows that, under Assumptions \ref{ass_2}--\ref{ass_3}, $(1+\sigma_x^2p\nu_p^2)^{-1}(\bI_n+\nu_p^2\bX\bX^\intercal)\stackrel{p}\to\bI_n$ and thus the continuous mapping theorem, combined with Lemma \ref{lemma:KL continuity}, implies that 
	\begin{align}\label{eq:eq3}
	\textsc{kl}[\textsc{tn}(\bzero,\bI_n, \mathbb{A})\mid \mid \textsc{tn}\{\bzero,(1+\sigma_x^2p\nu_p^2)^{-1}(\bI_n+\nu_p^2\bX\bX^\intercal), \mathbb{A}\}]\stackrel{p}\to 0,
	\end{align}
	as $p\to\infty$.
	Hence, combining \eqref{eq:eq1}--\eqref{eq:eq3}
	we obtain $\textsc{kl}\{q^*_{\textsc{pfm}}(\bz)\mid \mid p(\bz \mid \by)\}\stackrel{p}\to 0$ as $p\to\infty$. Recalling the proof of Theorem \ref{teo_2} and leveraging the chain rule of the \textsc{kl} divergence, this result implies \smash{$\textsc{kl}\{q^*_{\textsc{pfm}}(\beta)\mid \mid p(\bbeta \mid \by)\}\stackrel{p}\to 0$} as $p\to\infty$, since $\textsc{kl}\{q^*_{\textsc{pfm}}(\bz)\mid \mid p(\bz \mid \by)\}=\textsc{kl}\{q^*_{\textsc{pfm}}(\beta,\bz)\mid \mid p(\beta,\bz \mid \by)\} \geq \textsc{kl}\{q^*_{\textsc{pfm}}(\beta)\mid \mid p(\bbeta \mid \by)\}$.
\end{proof}

\begin{proof}[Proof of Corollary \ref{cor_predictive}.]
	Lemma \ref{lemma:KL_predictive} and Theorem \ref{teo_3} imply that, under Assumptions \ref{ass_2}--\ref{ass_3}, we have
	\begin{equation*}
	\mbox{pr}(\sup_{\bx_{\textsc{new}}\in\mathbb{R}^p}|\mbox{\normalfont pr}_{\textsc{pfm}} -  \mbox{\normalfont pr}_{\textsc{sun}}| > \epsilon) \le
	\mbox{pr}[\textsc{kl}\{q_{\textsc{pfm}}^{*}(\bbeta) \mid \mid p(\bbeta \mid \by)\} > 2\epsilon^2] \to 0,\quad \text{as $p\to\infty$\,}
	\end{equation*}
	for any $\epsilon>0$. This proves the convergence in probability statement in Corollary \ref{cor_predictive}.
\end{proof}

\vspace{5pt}

\subsection{Proof of Theorem \ref{teo_4}}
\vspace{3pt}

\renewcommand{\thelemma}{S7}
\begin{lemma}\label{lemma:m}
	Let $y\in\{0;1\}$ 
	and $\bar{z}=\mu+(2y -1)\sigma \phi(\mu/\sigma)/\Phi\{(2y -1)\mu/\sigma\}$ be a function of $\mu\in\mathbb{R}$ and $\sigma\geq 0$.
	Then $\sup_{\mu,\sigma}(|\mu|+\sigma)^{-1}|\bar{z}|<\infty$.
\end{lemma}

\vspace{5pt}
\begin{proof} To prove Lemma~\ref{lemma:m}, first  notice that $\Phi\{(2y -1)\mu/\sigma\} \geq \Phi(-|\mu|/\sigma)$. Combining this result with the  triangle inequality on $|\bar{z}|$, it follows
	\begin{eqnarray*}
		(|\mu|+\sigma)^{-1}|\bar{z}|
		\leq
		1 +(|\mu|+\sigma)^{-1}\sigma \phi(|\mu|/\sigma)
		/\Phi(-|\mu|/\sigma).
	\end{eqnarray*}
	If $|\mu|\leq\sigma$, then 
	$
	|\bar{z}|/(|\mu|+\sigma)
	\leq
	1 +1\cdot \phi\left(0\right)
	\left/\Phi\left(-1\right)\right.<\infty
	$.
	If $|\mu|> \sigma$, then setting $t=|\mu|/\sigma$ and using the bound 
	$
	\Phi(-t)\geq (2\pi)^{-1/2}t(t^2+1)^{-1}\exp(-t^2/2)$,
	which holds for every $t>0$, we have 
	\begin{eqnarray*}
		\begin{split}
			(|\mu|+\sigma)^{-1}|\bar{z}|
			&\leq
			1 +|\mu|^{-1}\sigma\phi(t)
			\left/\Phi\left(-t\right)\right.\\
			& \leq
			1 +t^{-1}
			\exp(-t^2/2)
			[(t^2+1)^{-1}t\exp(-t^2/2)]^{-1}
			=
			1 +t^{-2}(t^2+1)
			<
			3
		\end{split}
	\end{eqnarray*}
	where in the last inequality we have used $t>1$, since $|\mu|> \sigma$.
	Combining the above results it follows that $\sup_{\mu,\sigma}(|\mu|+\sigma)^{-1}|\bar{z}|<\infty$ as desired.
\end{proof}

\vspace{15pt}

\renewcommand{\thelemma}{S8}
\begin{lemma}\label{lemma:mu_i}
	Under Assumptions \ref{ass_2} and \ref{ass_3}, for every statistical unit $i=1, \ldots, n$, we have that \smash{$(1+\sigma_x^2p\nu_p^2)^{-1/2}\mu_i^{(1)}\stackrel{p}\to 0$ as $p\to\infty$, where $\mu_i^{(1)}$} is defined as in Algorithm \ref{algo2}.
\end{lemma}

\vspace{8pt}
\begin{proof}	\textsc{Case [a]} $\alpha=\infty$. 
	For any $i=1,\dots, n$, we have
	\begin{equation}\label{eq:upper_bound1}
	\dfrac{|\bar{z}_i^{(0)}|}{(1+\sigma_x^2p\nu_p^2)^{1/2}} = 
	\dfrac{|\bar{z}_i^{(0)}|}{|\mu_i^{(0)}| + \sigma_i^*}
	\dfrac{|\mu_i^{(0)}|  +\sigma_i^*}{(1+\sigma_x^2p\nu_p^2)^{1/2}}
	\leq
	C
	\dfrac{|\mu_i^{(0)}|  +\sigma_i^*}{(1+\sigma_x^2p\nu_p^2)^{1/2}},
	\end{equation}
	where $C=\sup_{\mu_i^{(0)},\sigma_i^*}(|\mu_i^{(0)}|+\sigma_i^*)^{-1}|\bar{z}_i^{(0)}|<\infty$, by Lemma \ref{lemma:m}.
	Moreover
	\begin{equation}\label{eq:upper_bound2}
	(1+\sigma_x^2p\nu_p^2)^{-1/2}(|\mu_i^{(0)}|  +\sigma_i^*)
	\stackrel{p}\to 
	1
	\quad \hbox{as }p\to\infty,
	\end{equation}
	because \smash{$(1+\sigma_x^2p\nu_p^2)^{-1/2} \sigma_i^*\stackrel{p}\to 1$} as  \smash{$p\to\infty$} and $(1+\sigma_x^2p\nu_p^2)^{-1/2} |\mu_i^{(0)}|\to 0$ as $p\to\infty$. The latter result is due to the fact that $\alpha=\infty$, whereas the former follows as a consequence of Lemma \ref{lemma:H_asymp} after noticing fact that $\sigma_i^* = (1 - H_{ii})^{-1/2}$.
	Note that we are implicitly assuming Algorithm \ref{algo2} to have fixed initialization 
	\smash{$(\mu^{(0)}_1,\dots,\mu^{(0)}_n)\in \mathbb{R}^n$}, obtained from \smash{$(\bar{z}^{(0)}_1, \ldots, \bar{z}^{(0)}_n)\in\mathbb{R}^n$}.
	Combining \eqref{eq:upper_bound1}--\eqref{eq:upper_bound2} if follows that 
	\begin{equation}\label{eq:bounded_in_prob}
	\mbox{pr}\{(1+\sigma_x^2p\nu_p^2)^{-1/2}|\bar{z}_i^{(0)}|> C+\epsilon\} \to 0, \quad\hbox{as }p\to \infty, \hbox{ for any }\epsilon>0\,.
	\end{equation}

	We now prove that  $(1{+}\sigma_x^2p\nu_p^2)^{-1/2}\mu_i^{(1)}\stackrel{p}\to 0$ and $\mbox{pr}\{(1{+}\sigma_x^2p\nu_p^2)^{-1/2}|\bar{z}_i^{(1)}|> C+\epsilon\} \to 0$ as $p\to \infty$, for any $\epsilon>0$ and for every $i=1,\dots,n$.
	We proceed by induction on $i$.
	Recalling the definition of \smash{$\mu_1^{(1)}$} in Algorithm \ref{algo2}, 
	we have that
	\begin{equation}\label{eq:induction_1}
	\textstyle
	(1+\sigma_x^2p\nu_p^2)^{-1/2}\mu_1^{(1)}
	=\sum_{i'=2}^n \sigma_1^{*2}H_{1i'}(1+\sigma_x^2p\nu_p^2)^{-1/2}\bar{z}_{i'}^{(0)}.
	\end{equation}
	Moreover, 
	\begin{equation}\label{eq:to_0_in_prob}
	\sigma_i^{*2}H_{ii'} = \{(1+\sigma_x^2p\nu_p^2)(1-H_{ii})\}^{-1}\cdot\{(1+\sigma_x^2p\nu_p^2)H_{ii'}\}\stackrel{p}{\to} 0,
	\end{equation}
	as $p\to\infty$ for every $i\neq i'$ since, by Lemma \ref{lemma:H_asymp}, the first and second factors converge in probability to $1$ and $0$, respectively.
	Combining the result in \eqref{eq:to_0_in_prob} with that in   \eqref{eq:bounded_in_prob}, we have
	$\sigma_1^{*2}H_{1i'}(1+\sigma_x^2p\nu_p^2)^{-1/2}\smash{\bar{z}_{i'}^{(0)}\stackrel{p}{\to} 0}$ as $p\to\infty$ for every $i'\geq 2$. Thus, by \eqref{eq:induction_1}, also 
	\smash{$(1+\sigma_x^2p\nu_p^2)^{-1/2}\mu_1^{(1)}\stackrel{p}{\to} 0$} as $p\to\infty$, which implies \smash{$(1+\sigma_x^2p\nu_p^2)^{-1/2}|\mu_1^{(1)}|\stackrel{p}{\to} 0$} as $p\to\infty$. Finally, the statement 
	\begin{equation*}
	\mbox{pr}\{(1+\sigma_x^2p\nu_p^2)^{-1/2}|\bar{z}_1^{(1)}|> C+\epsilon\} \to 0, \quad\hbox{as }p\to \infty, \hbox{ for any }\epsilon>0,
	\end{equation*}	
	follows from \smash{$(1+\sigma_x^2p\nu_p^2)^{-1/2}|\mu_1^{(1)}|\stackrel{p}{\to} 0$} as $p\to\infty$, and arguments analogous to the ones in \eqref{eq:upper_bound1} and \eqref{eq:upper_bound2}.
	We thus proved the desired induction statement for $i=1$.
	When $i\geq 2$ we have
	\begin{equation}\label{eq:general_i}
	\begin{split}
	\textstyle
	&(1+\sigma_x^2p\nu_p^2)^{-1/2}\mu_i^{(1)} \\
	&=	\textstyle
	\sum_{i'=1}^{i-1}\sigma_i^{*2}H_{ii'}(1+\sigma_x^2p\nu_p^2)^{-1/2}\bar{z}_{i'}^{(1)}+\sum_{i'=i+1}^n\sigma_i^{*2}H_{ii'}(1+\sigma_x^2p\nu_p^2)^{-1/2}\bar{z}_{i'}^{(0)}.
	\end{split}
	\end{equation}
	For $i'> i$, we have that \smash{$\sigma_i^{*2}H_{ii'}(1+\sigma_x^2p\nu_p^2)^{-1/2}\bar{z}_{i'}^{(0)}\stackrel{p}\to 0,$} as $p\to\infty,$ by the convergence results in \eqref{eq:bounded_in_prob} and in \eqref{eq:to_0_in_prob}. Instead, for $i'< i$, it follows that 
	\smash{$\sigma_i^{*2}H_{ii'}(1+\sigma_x^2p\nu_p^2)^{-1/2}\bar{z}_{i'}^{(1)} \stackrel{p}\to 0$}, as  $p\to\infty,$
	provided that, by \eqref{eq:to_0_in_prob}, \smash{$\sigma_i^{*2}H_{ii'}\to 0$}, in probability, as $p\to\infty$, and that
	$\mbox{pr}\{(1+\sigma_x^2p\nu_p^2)^{-1/2}\smash{|\bar{z}_{i'}^{(1)}|}> C+\epsilon\} \to 0$ as $p\to\infty$, for every $\epsilon>0$, as a direct consequence of the inductive argument.
	Combining these results, it follows by \eqref{eq:general_i} that \smash{$(1+\sigma_x^2p\nu_p^2)^{-1/2}\mu_i^{(1)}\stackrel{p}\to 0$} and thus
	that  \smash{$\mbox{pr}\{(1+\sigma_x^2p\nu_p^2)^{-1/2}|\bar{z}_{i}^{(1)}|>  C+\epsilon\} \to 0$} as $p\to\infty$ by arguments analogous to the ones in \eqref{eq:upper_bound1} and \eqref{eq:upper_bound2}.
	The thesis follows by induction.
	\vspace{3pt}
	
	\textsc{Case [b] $\alpha\in(0,\infty)$}. Here, the stronger result $\mu_i^{(1)}\stackrel{p}\to 0$ as $p\to\infty$ holds. The proof follows the same steps of \textsc{Case [a]}. First,  adapting the derivations above, for each $i=1,\ldots,n$, we have	\begin{equation}\textstyle \label{eq:upper_bound1b}
	|\bar{z}_i^{(0)}| = \dfrac{|\bar{z}_i^{(0)}|}{|\mu_i^{(0)}|+ \sigma_i^*}(|\mu_i^{(0)}|+ \sigma_i^*) \le C (|\mu_i^{(0)}|+ \sigma_i^*),
	\end{equation}
	where $C=\sup_{\mu_i^{(0)},\sigma_i^*}(|\mu_i^{(0)}|+\sigma_i^*)^{-1}|\bar{z}_i^{(0)}|<\infty$, by Lemma \ref{lemma:m}.
	Moreover,
	\begin{equation}\label{eq:upper_bound2b}
	\sigma_i^*\stackrel{p}{\to} (1+\alpha\sigma_x^2)^{1/2}, \quad \mbox{as } p \to \infty,
	\end{equation}
	by $\sigma_i^* = (1 - H_{ii})^{-1/2}$, and Lemma \ref{lemma:H_asymp}.
	Thus, calling $M = \max_{i=1,\ldots,n} |\mu_i^{(0)}|< \infty$, it follows by \eqref{eq:upper_bound1b} and \eqref{eq:upper_bound2b} that
	\begin{equation}\label{eq:bounded_in_probb}
	\mbox{pr}[|\bar{z}_i^{(0)}|> C\{M+ (1+\alpha\sigma_x^2)^{1/2}+\epsilon\}] \to 0, \quad\hbox{as }p\to \infty, \hbox{ for any }\epsilon>0\,.
	\end{equation}
	Adapting the proof of the previous case, we now prove that \smash{$\mu_i^{(1)}\stackrel{p}\to 0$} and $\mbox{pr}[\smash{|\bar{z}_i^{(1)}|}> C\{(1+\alpha\sigma_x^2)^{1/2}+\epsilon\}] \to 0$ as $p\to \infty$, for any $\epsilon>0$ and for every $i=1,\dots,n$.
	We proceed by induction on $i$.
	Recalling the definition of \smash{$\mu_1^{(1)}$} in Algorithm \ref{algo2}, 
	we have that
	\begin{equation}\label{eq:induction_1b} \textstyle
	\mu_1^{(1)}
	=\sum_{i'=2}^n \sigma_1^{*2}H_{1i'}\bar{z}_{i'}^{(0)} \, .
	\end{equation}
	Since the arguments leading to  \eqref{eq:to_0_in_prob} hold also in this case, by 
	\eqref{eq:to_0_in_prob} and \eqref{eq:bounded_in_probb} one has
	\smash{$\sigma_1^{*2}H_{1i'}\bar{z}_{i'}^{(0)}\stackrel{p}{\to} 0$} as $p\to\infty$ for every $i'\geq 2$. Thus, by \eqref{eq:induction_1b}, also 
	\smash{$\mu_1^{(1)}\stackrel{p}{\to} 0$} as $p\to\infty$, which implies \smash{$|\mu_1^{(1)}|\stackrel{p}{\to} 0$} as $p\to\infty$.
	Finally, the statement \smash{$\mbox{pr}[|\bar{z}_1^{(1)}|> C\{(1+\alpha\sigma_x^2)^{1/2}+\epsilon\}] \to 0$} as $p\to \infty$ for any $\epsilon>0$ follows from \smash{$|\mu_1^{(1)}|\stackrel{p}{\to} 0$} as $p\to\infty$ and arguments analogous to the ones in \eqref{eq:upper_bound1b} and \eqref{eq:upper_bound2b}.
	The desired induction statement for $i=1$ has thus been proved.
	When $i \geq 1$ we have
	\begin{equation}\label{eq:general_ib}\textstyle
	\mu_i^{(1)}
	=
	\sum_{i'=1}^{i-1}\sigma_i^{*2}H_{ii'}\bar{z}_{i'}^{(1)}+\sum_{i'=i+1}^n\sigma_i^{*2}H_{ii'}\bar{z}_{i'}^{(0)}.
	\end{equation}
	For $i'> i$, we have that $\sigma_i^{*2}H_{ii'}\bar{z}_{i'}^{(0)}\stackrel{p}\to 0$ as $p\to\infty$, by \eqref{eq:bounded_in_probb} and \eqref{eq:to_0_in_prob}.
	For $i'< i$, we have 
	\smash{$\sigma_i^{*2}H_{ii'}\bar{z}_{i'}^{(1)} \stackrel{p}\to 0$} as $p\to\infty$, 
	since \smash{$\sigma_i^{*2}H_{ii'}\stackrel{p}{\to}0$} as $p\to\infty$ by \eqref{eq:to_0_in_prob}, and 
	$\mbox{pr}[\smash{|\bar{z}_{i'}^{(1)}|}> C\{(1+\alpha\sigma_x^2)^{1/2}+\epsilon\}] \to 0$ as $p\to\infty$ for every $\epsilon>0$ by the inductive argument. Therefore,
	it follows by \eqref{eq:general_ib} that \smash{$\mu_i^{(1)}\stackrel{p}\to 0$} and thus
	that  \smash{$\mbox{pr}[|\bar{z}_{i}^{(1)}|> C\{(1+\alpha\sigma_x^2)^{1/2}+\epsilon\}] \to 0$ as $p\to\infty$} by arguments analogous to the ones in \eqref{eq:upper_bound1b} and \eqref{eq:upper_bound2b}.
	The thesis follows by induction.
\end{proof}

\vspace{10pt}
\begin{proof}[Proof of Theorem \ref{teo_4}.]
	By leveraging the chain rule for the \textsc{kl} divergence, combined with the fact that $q_{\textsc{pfm}}^{(1)}(\bbeta \mid \bz)=p(\bbeta \mid \by,\bz)$, we have
	\begin{align*}
	\textsc{kl}\{q_{\textsc{pfm}}^{(1)}(\bbeta) \mid \mid p(\bbeta \mid \by)\}
	&\leq
	\textsc{kl}\{q_{\textsc{pfm}}^{(1)}(\bbeta) \mid \mid p(\bbeta \mid \by)\}+E_{q_{\textsc{pfm}}^{(1)}(\bbeta)}[\textsc{kl}\{q_{\textsc{pfm}}^{(1)}( \bz\mid \bbeta ) \mid \mid p(\bz\mid \by,\bbeta)\}]\\
	&=\textsc{kl}\{q_{\textsc{pfm}}^{(1)}(\bz) \mid \mid p(\bz \mid \by)\}+E_{q_{\textsc{pfm}}^{(1)}(\bz)}[\textsc{kl}\{q_{\textsc{pfm}}^{(1)}(\bbeta \mid \bz) \mid \mid p(\bbeta \mid \by,\bz)\}]\\
	&=
	\textsc{kl}\{q_{\textsc{pfm}}^{(1)}(\bz)\mid\mid p(\bz \mid \by)\}\,.
	\end{align*}
	Since $q_{\textsc{pfm}}^{(1)}(\bz)=\textsc{tn}(\bmu^{(1)},\bsigma^{*2}, \mathbb{A})$, $p(\bz 		\mid \by)=\textsc{tn}\{0,(\bI_n+\nu_p^2\bX\bX^\intercal), \mathbb{A}\}$ and
	the \textsc{kl} divergence is invariant with respect to bijective transformations,  then, calling $k_p = (1+\sigma_x^2p\nu^2_p)$ and rescaling each $z_i$ by \smash{$k_p ^{-1/2}$}  we obtain
	$$
	\textsc{kl}\{q_{\textsc{pfm}}^{(1)}(\bz)\mid \mid p(\bz \mid \by)\}=
	\textsc{kl}[\textsc{tn}(k_p^{-1/2}\bmu^{(1)},k_p^{-1}\bsigma^{*2},\mathbb{A})\mid \mid \textsc{tn}\{0,k_p^{-1}(\bI_n+\nu_p^2\bX\bX^\intercal),\mathbb{A}\}].$$
	Lemma \ref{lemma:mu_i} implies that $k_p^{-1/2} \bmu^{(1)}\stackrel{p}\to 0$ as $p\to\infty$, while Lemmas \ref{lemma:covariance} and \ref{lemma:H_asymp} imply that both $k_p^{-1}(\bI_n+\nu_p^2\bX\bX^\intercal)$ and $k_p^{-1}\bsigma^{*2}$ converge in probability to $\bI_n$ as $p\to\infty$.
	Therefore
	$$
	\textsc{kl}[\textsc{tn}(k_p^{-1/2}\bmu^{(1)},k_p^{-1}\bsigma^{*2},\mathbb{A})\mid \mid \textsc{tn}\{0,k_p^{-1}(\bI_n+\nu_p^2\bX\bX^\intercal),\mathbb{A}\}]
	\stackrel{p}\to 0
	$$
	by Lemma \ref{lemma:KL continuity} and the continuous mapping theorem, implying
	$\textsc{kl}\{q_{\textsc{pfm}}^{(1)}(\bbeta)\mid \mid p(\bbeta \mid \by)\}
	\stackrel{p}\to 0
	$ as $p\to\infty$, as desired.
\end{proof}

\section{Additional simulation studies}\label{simu_1_sup}
\renewcommand{\thefigure}{S1}
\begin{figure}[b!]
	\captionsetup{font={small}}
	\includegraphics[width=\linewidth]{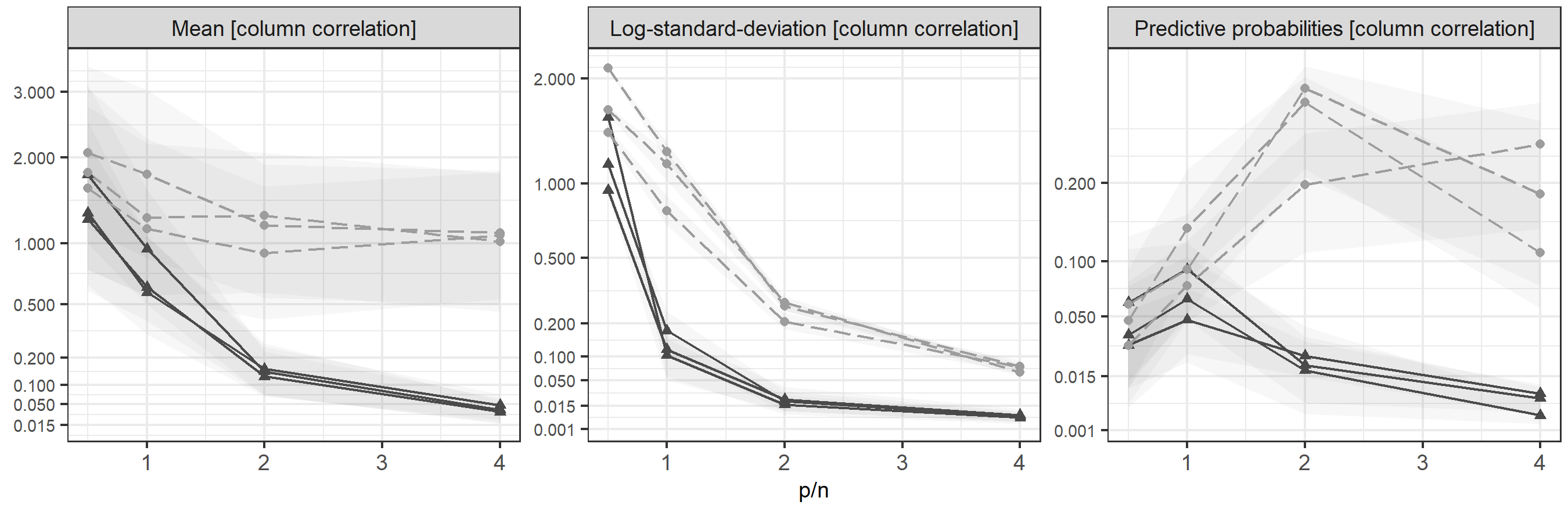}
	\includegraphics[width=\linewidth]{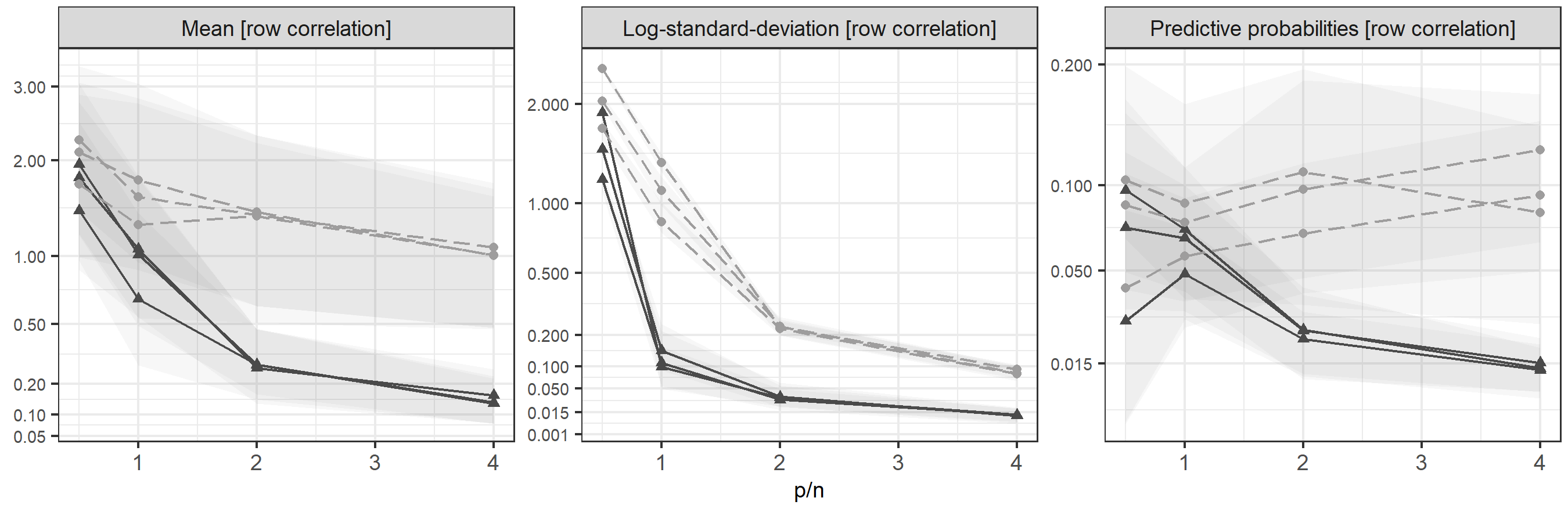}
	\caption{Accuracy in the approximation for three key  functionals of the posterior distribution for $\bbeta$. Trajectories for the median of the absolute differences between an accurate but expensive Monte Carlo estimate of such functionals and their approximations provided by  partially-factorized variational Bayes (dark grey solid lines) and mean-field variational Bayes (light grey dashed lines), respectively, for increasing values of the ratio $p/n$, i.e.,  $p/n \in \{0.5;1;2;4\}$. The different trajectories for each of the two methods correspond to three different settings of the sample size $n$,  i.e.,  $n  \in \{50;100;250\}$, whereas the grey areas denote the first and third quartiles computed from the absolute differences. The two row panels correspond to simulation scenarios in which there is dependence across the different predictors defining the columns of $\bX$ (first row), and dependence  between the statistical units denoting the rows of $\bX$ (second row). For graphical purposes we consider a square-root scale for the $y$-axis. }
	\label{fsimu2}
\end{figure}

We consider here additional simulation studies to check whether the findings reported in the article apply also to more general data structures which do not meet the assumptions we require to prove the theory in Sections \ref{sec_2.1}--\ref{sec_2.2} in the article. With this goal in mind, we replicate the simulation experiment in Section \ref{sec_simu} of the article under two alternative generative mechanisms for the predictors in the design matrix $\bX$, which avoid assuming independence across columns and rows. More specifically, in the first alternative scenario we induce dependence across columns by simulating the $n$ rows of $\bX$, excluding the intercept, from a $(p-1)$-variate Gaussian distribution with zero means, unit variances and strong pairwise correlations $\mbox{corr}(x_{ij},x_{ij'})=0.75$, for each $j=2, \ldots, p$ and $j'=1, \ldots, j-1$. In the second alternative scenario, we induce instead  correlation across statistical units by simulating the $p-1$ columns of $\bX$, excluding the intercept, from an $n$-variate Gaussian with zero means, unit variances and decaying pairwise correlations $\mbox{corr}(x_{ij},x_{i'j})=0.75^{|i-i'|}$, for each $i=1, \ldots, n$ and $i'=1, \ldots, n$. These two scenarios feature strong correlations, arguably more extreme relative to the ones found in real-world applications. Hence, such an assessment provides an important stress test to evaluate the performance of the proposed partially-factorized variational Bayes and the classical mean-field variational Bayes in challenging regimes which do not meet Assumptions 1 and 3; note that since in these simulations $\mbox{var}(x_{ij})=1$, then correlations coincide with covariances.  For instance, the more structured dependence that we induce among units is meant to provide a highly challenging scenario in which  standardization has almost no effect in reducing the dependence across rows. In fact, as discussed more in detail in Section \ref{sec_3} in the article, we found that the recommendations provided by \citet{gelman_2008} and \citet{chopin_2017} on standardization are beneficial in reducing the covariance between the rows of $\bX$ in routine real-world applications, thereby further improving the quality of the   partially-factorized approximation. When standardization is implemented, we recommend to always include the intercept term which is crucial to account for the centering effects \citep[e.g.,][]{gelman_2008, chopin_2017}.

Figure \ref{fsimu2} illustrates performance under these challenging regimes by  reproducing an assessment analogous to Figure \ref{fsimu1} in the article. According to Figure \ref{fsimu2}, the results are overall coherent with the evidence provided by Figure \ref{fsimu1}, for  both approximations, meaning that the proposed partially-factorized strategy provides a powerful solution under a variety of challenging settings, and the theoretical results in Sections \ref{sec_2.1}--\ref{sec_2.2} of the article match empirical evidence even beyond the assumptions we require for their proof. As one might expect, the partially-factorized approximation  is slightly less accurate when there is correlation across the $n$ units relative to scenarios with no correlation or correlation across the $p$ predictors. Nonetheless, the approximation error relative to \textsc{stan} Monte Carlo estimates is still very low also in such settings.

\section{Additional results for the medical applications}\label{app_1_sup}
\renewcommand{\thefigure}{S2}
\begin{figure}[b!]
	\captionsetup{font={small}}
	\centering
	\includegraphics[width=\linewidth]{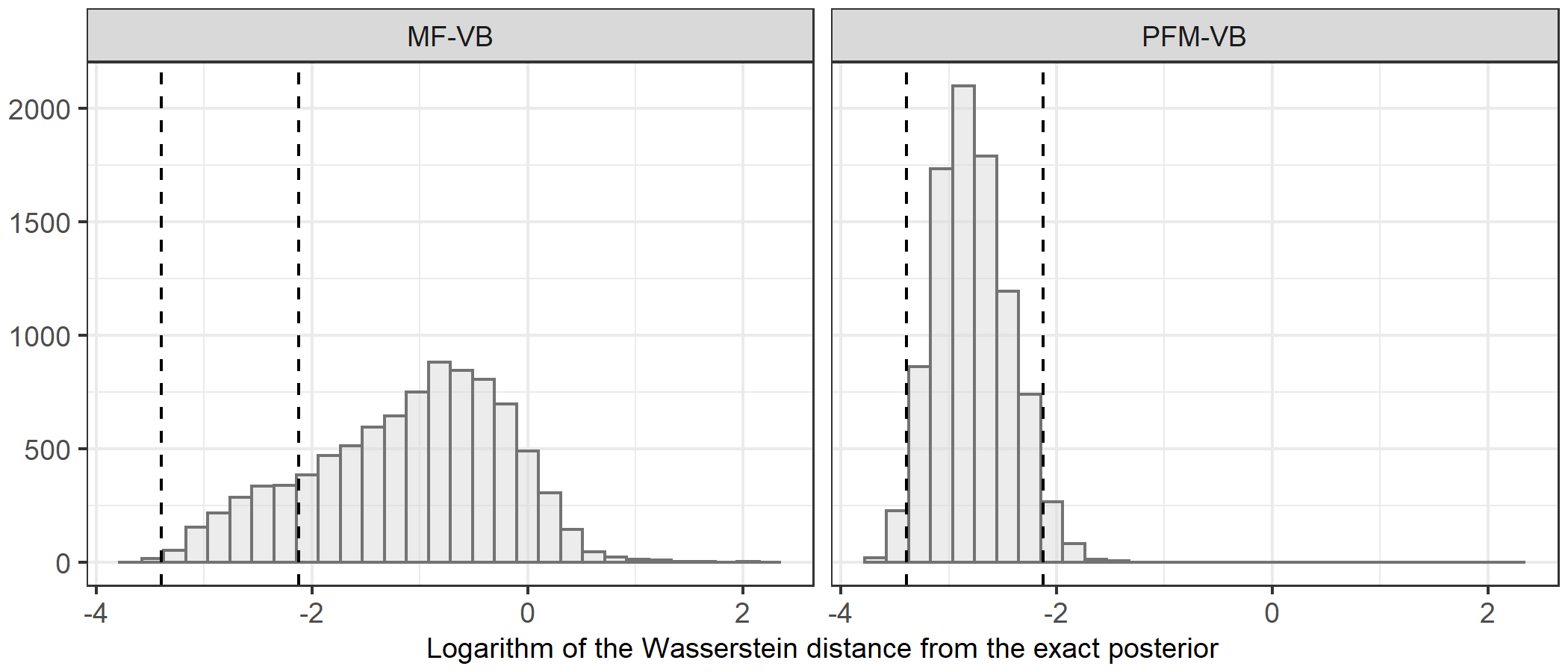}
	\caption{For the two variational approximations, histograms of the log-Wasserstein distances between the $p=9036$ approximate marginal densities provided by the two variational methods and the exact posterior marginals. These distances are computed leveraging 20000 samples from the approximate and exact marginals. To provide insights on Monte Carlo error, the dashed vertical lines represent the quantiles $2.5\%$ and $97.5\%$ of the log-Wasserstein distances between two different  samples of 20000 draws from the same exact posterior marginals.  \textsc{mf}-\textsc{vb}, mean-field variational Bayes; \textsc{pfm}-\textsc{vb}, partially-factorized variational Bayes.
		\vspace{-8pt}}
	\label{f1}
\end{figure}
We provide here additional analyses of the medical datasets  considered in Section \ref{sec_3} in the article. Figure \ref{f1} extends the comparison among the Wasserstein distances between the exact posterior marginals and the corresponding approximations provided by mean-field and partially-factorized variational Bayes. More specifically, Figure \ref{f1} displays the histograms of the log-Wasserstein distances among the $p=9036$ exact posterior marginals for the Alzheimers' dataset and the associated approximations under the two variational methods. As discussed in the article, such quantities are computed with the \texttt{R} function \texttt{wasserstein1d}, which uses $20000$ values sampled  from the approximate  and exact marginals. These histograms further clarify that  our partially-factorized solution improves the quality of the mean-field one and, in practice, it matches almost perfectly the exact posterior since it provides   distances within the range of values obtained by comparing two different  samples of 20000 draws from the same exact posterior marginals. 

\renewcommand{\thefigure}{S3}
\begin{figure}[t!]
	\captionsetup{font={small}}
	\includegraphics[width=\linewidth]{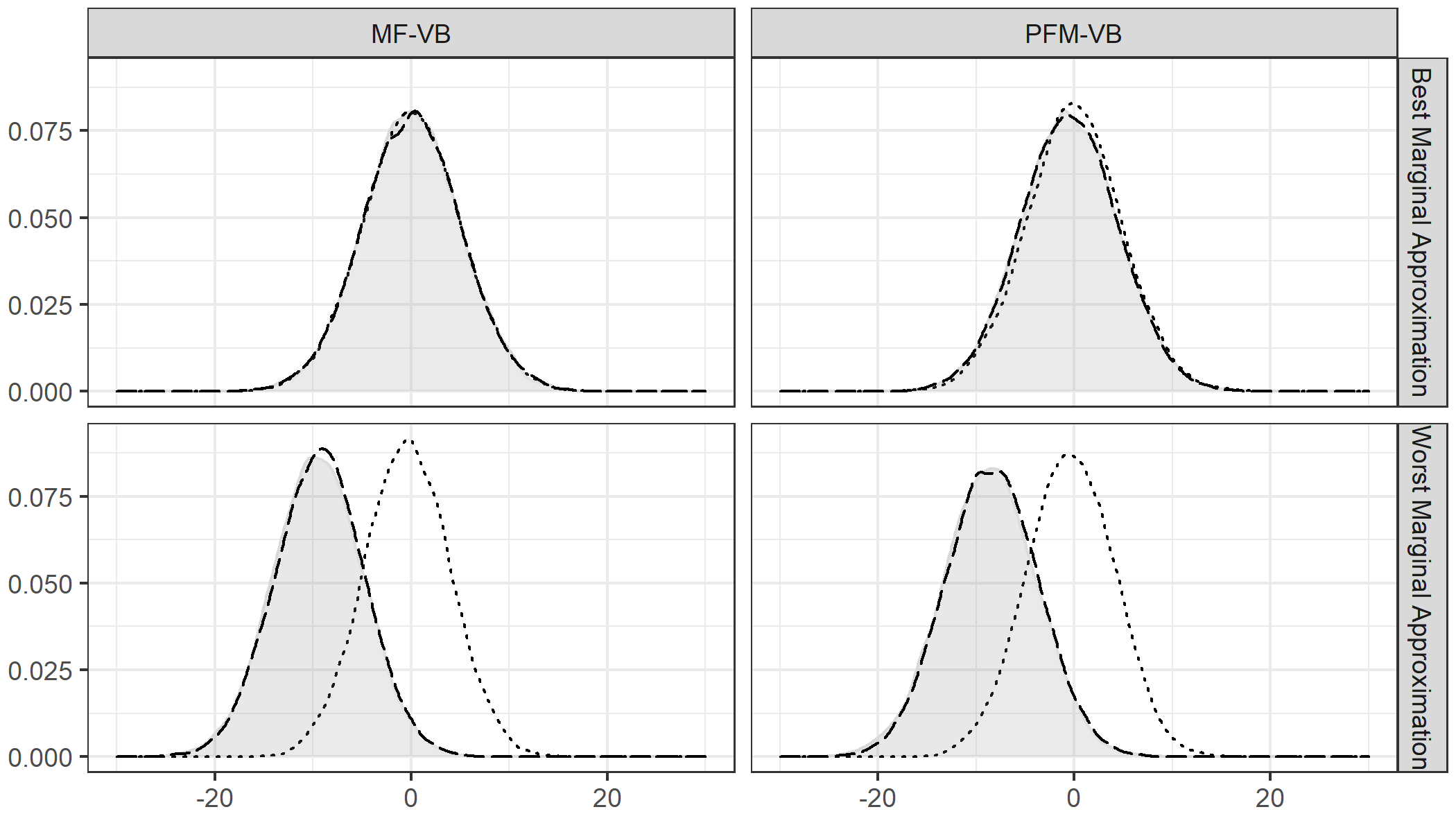}
	\caption{Quality of the marginal approximation for the coefficients associated with the highest and lowest Wasserstein distance from the exact posterior under mean-field variational Bayes (\textsc{mf}-\textsc{vb}) and partially-factorized variational Bayes (\textsc{pfm}-\textsc{vb}), respectively. The shaded grey area denotes the density of the exact posterior marginal, whereas the dotted and dashed lines represent the approximate densities provided by mean-field and partially-factorized variational Bayes, respectively. }
	\label{f2}
\end{figure}

\renewcommand{\thefigure}{S4}
\begin{figure}[b!]
	\captionsetup{font={small}}
	\includegraphics[width=\linewidth]{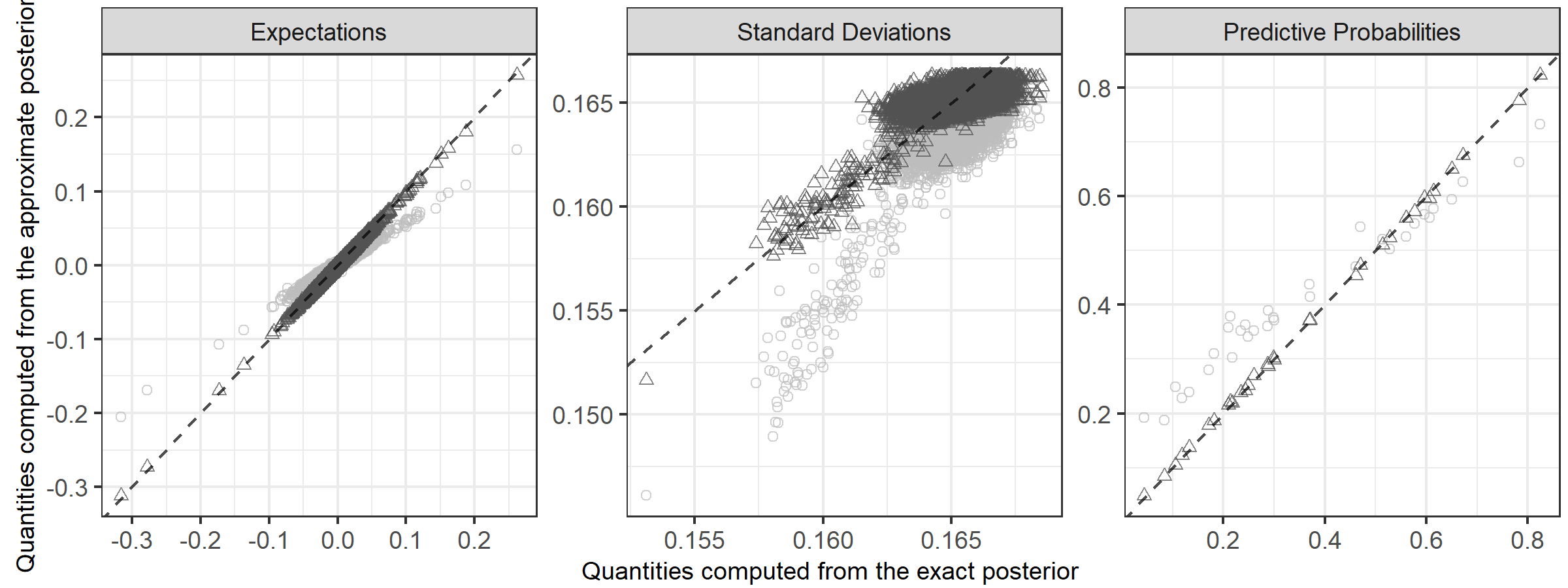}
	\caption{Scatterplots comparing the posterior expectations, standard deviations and predictive probabilities computed from 20000 values sampled from the exact \textsc{sun} posterior, with those provided by the mean-field variational Bayes (light grey circles) and partially-factorized variational Bayes (dark grey triangles), under a different setting for   $\nu^2_p$ controlling the variance of the linear predictor via $\nu^2_p=25 \cdot 10/p$, to induce increasing shrinkage. }
	\label{f4s}
\end{figure}

Figure \ref{f2}  complements the results in Figure  \ref{f1}  by  comparing graphically the quality of the marginal approximation for the coefficients associated with the highest and lowest Wasserstein distance from the exact posterior under the two variational approximations analyzed.
As is clear from Figure \ref{f2},  partially-factorized variational Bayes produces approximations which perfectly overlap with the exact posterior in all cases, including also the worst-case scenario with the highest Wasserstein distance.
Consistent with Theorem \ref{teo_1}, the mean-field approximation has instead reduced quality, mostly due to a tendency to over-shrink towards zero the locations of the actual posterior.  This  effect is evident from Figures  \ref{f3}--\ref{f4} in the article, and is further illustrated in Figure~\ref{f4s} which reproduces the analyses in  Figures \ref{f3}--\ref{f4} under an even more extreme setting in which the prior variance $\nu_p^2$ is set equal to $25\cdot 10/p$.  This choice controls, heuristically, the total variance of the linear predictor as if there were only $10$ coefficients, out of $9036$,  with prior variance  $25$, while the others were fixed to zero, thus  inducing an even stronger shrinkage effect in high dimensions relative to those considered in Figures  \ref{f3}--\ref{f4}. As a result, also the exact posterior means concentrate around 0, thus mitigating the issues of classical mean-field approximations. Nonetheless, even if the prior variance is $\nu_p^2 \approx 0.028$, the  classical mean-field solution still maintains a bias in the locations, standard deviations and predictive probabilities which is, instead, not present under the proposed partially-factorized solution.

To conclude our empirical studies, we replicate the analyses presented in Table \ref{table2} of the article, implementing now the alternative choices of $\nu^2_p$ considered in Figure \ref{f4} and Figure~\ref{f4s}, namely $25\cdot 100/p$ and $25\cdot 10/p$.   As previously discussed, these settings   induce increasing shrinkage in high dimension by controlling the variance of the whole linear predictor. Consistent with Table \ref{table2}, we compare predictive performance of classical mean-field variational Bayes, the proposed partially-factorized solution, and the sparse variational Bayes approximation for spike-and-slab Bayesian logistic regression  \citep{ray2020spike}; refer to Section \ref{sec_3} in the article and to the associated  repository \texttt{https://github.com/augustofasano/Probit-PFMVB} for details on the implementation of sparse variational Bayes. As mentioned in Section \ref{sec_3}, this competitor relies on a different model and, hence, it approximates a different posterior. Therefore, it is not possible to separate quality of the approximation from model performance in  the comparison with sparse variational Bayes. Results are reported in Table \ref{tableS2}, which also contains the case $\nu_p^2=25$ from Table  \ref{table2}  for completeness. Consistent with the findings discussed in the article, the partially-factorized solution over-performs the classical mean-field approximation in all settings, whereas the gains over sparse variational Bayes  depend on the specific dataset analyzed, without a clear pattern in relation to $p/n$. Nonetheless, in the 12 comparisons reported in Table \ref{tableS2}, the proposed partially-factorized solution ranks first $8$ times, whereas sparse variational Bayes over-performs the other methods the remaining  $4$ times. Overall, these results suggest that the ridge-type shrinkage associated with partially-factorized variational Bayes under model~(1) is competitive to state-of-the-art sparse variational methods, and, at the same time, it motivates future work to incorporate other types of shrinkage priors within our methodology, as discussed in Section \ref{sec_4} of the article. 

\renewcommand{\thetable}{S1}
\begin{table}[h]
	\caption{Under three different settings of $\nu^2_p$, five-fold cross-validation estimate of the test deviance in four different medical datasets. Italics values denote best performance. \textsc{mf}-\textsc{vb}, mean-field variational Bayes; \textsc{pfm}-\textsc{vb}, partially-factorized variational Bayes; \textsc{svb} sparse variational Bayes.
		\label{tableS2}}
	\begin{tabular}{llcccc}
		&&  {\texttt{parkinson}}  &{\texttt{voice}} & {\texttt{lesion}} &  {\texttt{alzheimer}} \\ 
		$\nu_p^2$&Method \qquad& {\scriptsize{$(n=756,p=754)$}} &\scriptsize{\ $(n=126,p=310)$} & \scriptsize{\ $(n=76,p=952)$}& \scriptsize{\ $(n=333,p=9036)$}\\ [5pt]
		$25$ & \textsc{mf}-\textsc{vb}&  $309.82$ &  $59.38$ & $48.66$ & $228.71$ \\ 
		&\textsc{pfm}-\textsc{vb} & ${\it 306.37}$ &  $46.35$ &   ${\it 27.24}$ & $187.52$ \\[2pt]
		&\textsc{svb} & $317.42$ &  ${\it 42.16}$ &  $35.81$ & ${ \it126.24}$ \\  [8pt]
		$25{\cdot}100/p$ \qquad & \textsc{mf}-\textsc{vb}& $298.13$ &   $52.39$ & $42.21$ & $215.98$  \\ 
		&\textsc{pfm}-\textsc{vb} &  $\it{294.91}$ &   $\it{45.74}$ &   $\it{27.24}$& $187.60$ \\[2pt]
		&\textsc{svb} & $305.18$ &   $46.32$ & $32.14$& $\it{139.12}$ \\  [8pt]
		$25{\cdot}10/p$ \qquad & \textsc{mf}-\textsc{vb}& $264.15$ &   $47.96$ &  $32.69$&  $202.28$  \\ 
		&\textsc{pfm}-\textsc{vb} &  $\it{261.87}$ &   $\it{45.74}$ &   $\it{27.43}$& $188.33$ \\[2pt]
		&\textsc{svb} &  $282.72$ &  $48.59$ &   $35.88$ & $\it{185.25}$ 
	\end{tabular}
\end{table}

%\bibliographystyle{apalike}
%\bibliography{paper-ref}

\end{document}